\RequirePackage[final]{graphicx}
\documentclass[runningheads]{llncs} 
\usepackage[T1]{fontenc}

\usepackage{amsfonts}
\usepackage{amsthm}
\usepackage{amsmath}
\usepackage{amssymb}
\usepackage{mathabx}
\usepackage{booktabs}
\usepackage{color}
\usepackage{xcolor}
\usepackage{wrapfig}  
\usepackage[draft]{fixme} \fxsetup{theme=color,mode=multiuser}
\FXRegisterAuthor{JL}{aJL}{JL}
\FXRegisterAuthor{NY}{aNY}{NY}
\usepackage{ifthen} 
\usepackage[inline]{enumitem}
\usepackage[linesnumbered]{algorithm2e}
\usepackage{fontawesome}
\usepackage{multibib}
\newcites{out}{Additional References for the Appendix}
\usepackage{thmtools}
\usepackage{thm-restate}
\usepackage{stmaryrd} 
\usepackage{cite}
\usepackage{url} 
\usepackage{comment} 
\usepackage{graphicx}
\usepackage{enumitem}
\usepackage{mathtools}
\mathtoolsset{showonlyrefs}
\usepackage{tikz}
\usetikzlibrary{arrows,shapes,automata,spy,positioning,
  calc,fit,arrows.meta,patterns,backgrounds,decorations.pathreplacing}
\makeatletter
\newcommand{\@chapapp}{\relax}\makeatother

\usepackage[title,toc,titletoc,header]{appendix}

\usepackage[hidelinks]{hyperref} 
\hypersetup{final}

\newif\iflong
\longtrue 
\newcommand{\appendixref}[1]{\iflong Appendix~\ref{#1}\else \cite{fullversion}\fi}

\newcommand{\proofsub}[1]{\noindent \textbf{(#1)}}

\newcommand{\CSA}{\textsc{csa}} \newcommand{\GMC}{{\textsc{gmc}}}
\newcommand{\DF}{{\textsc{sp}}}

\newcommand{\ER}{{\textsc{er}}}
\newcommand{\PG}{{\textsc{pg}}}

\newcommand{\MC}{-\textsc{mc}}
\newcommand{\SMC}{\textsc{smc}}
\newcommand{\PSPACE}{{\sc pspace}}
\newcommand{\FIFO}{{\sc fifo}}
\newcommand{\synk}{-synchronisable}

\newcommand{\exref}[1]{\ref{#1}}
\newcommand{\modifmark}{{\smash{$^\dagger$}}}
\newcommand{\asphi}{{\color{green!50!black!95}\varphi}}

\newcommand{\infBA}{\textsc{bi}}
\newcommand{\infOBI}{\textsc{obi}}
\newcommand{\infIBI}{\textsc{ibi}}
\newcommand{\infDIBI}{\textsc{sibi}}

\newcommand{\OBI}[1]{{#1}\textsc{-\infOBI}}
\newcommand{\IBI}[1]{{#1}\textsc{-\infIBI}}
\newcommand{\BA}[1]{{#1}\textsc{-\infBA}}

\newcommand{\DIBI}[1]{{#1}\textsc{-\infDIBI}}

\newcommand{\CIBI}[1]{{#1}\textsc{-cibi}}

\newcommand{\dependsin}[4]{#1 \vdash #2 \prec_{#3} #4}
\newcommand{\bindep}[3]{#1 \vdash #2 \prec #3}

\newcommand{\okmark}{{\color{green!50!black}\texttt{yes}}}
\newcommand{\komark}{{\color{red}\texttt{no}}}

\newcommand{\defi}{\overset{{\text{def}}}{=}}
\newcommand{\mmdef}{\defi}
\newcommand{\st}{\, \mid \,}
\newcommand{\qst}{:}
\DeclareMathOperator{\concat}{\cdot}

\newcommand{\naturals}{\mathbb{N}}

\newcommand{\ptp}[1]{{\color{blue!70!black!95}\mathtt{#1}}}
\newcommand{\msg}[1]{\mathit{#1}}

\newcommand{\p}{\ptp{p}}
\newcommand{\q}{\ptp{q}}
\newcommand{\rr}{\ptp{r}}
\newcommand{\s}{\ptp{s}}
\newcommand{\kleene}[1]{#1^{\ast}}

\newcommand{\subj}[1]{\mathit{subj}(#1)}
\newcommand{\obj}[1]{\mathit{obj}(#1)}

\newcommand{\chan}[1]{\mathit{chan}(#1)}

\newcommand{\action}{\ell} \newcommand{\acts}{\phi} \newcommand{\actsb}{\psi} \newcommand{\paset}{\Psi} \newcommand{\word}{w}
\newcommand{\emptyw}{\epsilon}
\newcommand{\vecemptyw}{\vec{\emptyw}}

\newcommand{\PSet}{\mathcal{P}} \newcommand{\CSet}{\mathcal{C}} \newcommand{\ASet}{\mathcal{A}} \newcommand{\ASetSend}{\mathcal{A}_!} \newcommand{\ASetRcv}{\mathcal{A}_?} \newcommand{\ASetC}{\ASet^\ast}
\newcommand{\ASetSendC}{\kleene{\ASetSend}}
\newcommand{\ASetRcvC}{\kleene{\ASetRcv}}
\newcommand{\ASigma}{\Sigma} \newcommand{\ASigmaC}{\ASigma^{\ast}} 
\newcommand{\PSENDRECEIVE}[2]{\ptp{#1} \dagger \msg{#2}}
\newcommand{\PSEND}[2]{\ptp{#1} ! \msg{#2}}
\newcommand{\PRECEIVE}[2]{\ptp{#1} ? \msg{#2}}

\newcommand{\csconf}[2]{(\vec{#1}; \vec{#2})}
\newcommand{\stablecsconf}[1]{(\vec{#1}; \vecemptyw)}
\newcommand{\TRANSS}[1]{\xrightarrow{#1}}
\newcommand{\TRANS}{\TRANSS{}}

\newcommand{\TRANSR}{\TRANS\negthickspace^\ast}

\newcommand{\bTRANSS}[2]{\xrightarrow{#1}_{#2}}
\newcommand{\bTRANSR}[1]{\TRANS\negthickspace^\ast_{#1}}

\newcommand{\kTRANSS}[1]{\bTRANSS{#1}{k}}
\newcommand{\rkTRANSS}[1]{\rbTRANSS{#1}{k}}
\newcommand{\rbTRANSS}[2]{\xrightharpoondown{#1}_{#2}}
\newcommand{\rkTRANSR}{\rkTRANSS{}\negthickspace^\ast}
\newcommand{\kTRANSR}{\kTRANSS{}\negthickspace^\ast}
\newcommand{\RS}{\mathit{RS}}

\newcommand{\grock}{\msg{r}}
\newcommand{\gpaper}{\msg{p}}
\newcommand{\gscissors}{\msg{s}}
\newcommand{\fontnode}[1]{\texttt{#1}}

\newcommand{\mbounded}{match-bounded}
\newcommand{\bounded}[2]{#1 \!\! \mid_{#2}}
\newcommand{\kexchange}[2]{#1 \!\! \parallel_{#2}}
\newcommand{\equivclass}[2]{[#1]_{#2}}

\DeclareMathOperator{\bisim}{\sim}
\DeclareMathOperator{\wbisim}{\approx}

\DeclareMathOperator{\resche}{\Bumpeq}

\newcommand{\nclosed}[3]{#2 \text{ is } \ensuremath{#1}\text{-closed for } #3}
\newcommand{\kclosed}[2]{\nclosed{k}{#1}{#2}}
\newcommand{\pclosed}[2]{\nclosed{k{+}1}{#1}{#2}}

\DeclareMathOperator{\eqpeerop}{\asymp}
\newcommand{\onpeer}[2]{\proj{#2}{#1}}
\newcommand{\eqpeer}[2]{#1 \eqpeerop #2}
\newcommand{\TS}[2]{\mathit{TS}_{#1}(#2)}
\newcommand{\kTS}[1]{\TS{k}{#1}}
\newcommand{\RTSACRO}{\mathit{RTS}}

\newcommand{\RTS}[2]{\RTSACRO_{#1}(#2)}
\newcommand{\kRTS}[1]{\RTS{k}{#1}}

\newcommand{\supermetaproj}[4]{#1_{#2}^{#3}(#4)}
\newcommand{\metaproj}[3]{\supermetaproj{\pi}{#1}{#2}{#3}}

\newcommand{\proj}[2]{\metaproj{\ptp{#2}}{}{#1}} 
\newcommand{\epsproj}[2]{\supermetaproj{\pi}{\ptp{#2}}{\epsilon}{#1}}

\newcommand{\sndproj}[2]{\supermetaproj{\sigma}{\ptp{#2}}{!}{#1}}
\newcommand{\sndrcvproj}[3]{\metaproj{\ptp{#2}}{#3}{#1}}

\newcommand{\esndproj}[2]{\metaproj{\ptp{#2}}{!}{#1}} \newcommand{\ercvproj}[2]{\metaproj{\ptp{#2}}{?}{#1}}

\newcommand{\lts}{\mathcal{T}}

\newcommand{\flabel}[1]{\ifthenelse{\equal{#1}{}}{\lambda}{\lambda(#1)}}

\newcommand{\sync}[1]{\mathit{sync}(#1)}

\newcommand{\metaszero}{t_0}

\newcommand{\spartition}[1]{\mathit{partition}(#1)}
\newcommand{\algoAssign}[0]{\ensuremath{\leftarrow}\ }
\newcommand{\sstack}{\mathit{stack}}
\newcommand{\svisited}{\mathit{visited}}
\newcommand{\sacc}{\mathit{accum}}
\newcommand{\pop}[1]{\mathit{pop}(#1)}
\newcommand{\push}[2]{\mathit{push}(#1,#2)}
\newcommand{\head}[1]{\mathit{head}(#1)}
\newcommand{\tail}[1]{\mathit{tail}(#1)}
\newcommand{\successor}[2]{\mathit{succ}(#1,#2)}
\newcommand{\pair}[2]{\langle #1, #2 \rangle}
\newcommand{\appfun}[2]{\mathit{f}(#1,#2)}
\newcommand{\rcvdir}[1]{\textit{rcv-dir}(#1)}
\newcommand{\snddir}[1]{\textit{snd-dir}(#1)}

\newcommand{\gatedistancein}{3pt}
\newcommand{\gatedistanceinand}{2pt}
\tikzset{
    every state/.style={minimum size=1pt,inner sep=1.2pt, initial text={}},
  mycfsm/.style={
    font=\scriptsize,
    initial where=left,
    initial distance=0.25cm,
    ->,>=stealth,auto, node distance=0.8cm and 0.8cm,
    scale=1, every node/.style={transform shape},
    baseline=(current  bounding  box.center)
  },
  ogate/.style = {
    diamond, draw, fill=white,
    minimum size=4mm,
    inner sep=0pt,
    postaction={path picture={        \draw[black]
        ([yshift=\gatedistancein]path picture bounding box.south) -- ([yshift=-\gatedistancein]path picture bounding box.north)
        ([xshift=-\gatedistancein]path picture bounding box.east) -- ([xshift=\gatedistancein]path picture bounding box.west)
        ;}}, drop shadow},
  agate/.style={draw,rectangle,
    minimum size=3mm,
    inner sep=0pt,
    fill=white,
    postaction={path picture={        \draw[black]
        ([yshift=\gatedistanceinand]path picture bounding box.south) --
        ([yshift=-\gatedistanceinand]path picture bounding box.north) ;}}, drop shadow},
  source/.style={draw,circle,fill=white,
    minimum size=3mm,
    inner sep=0pt, drop shadow},
  sink/.style={draw,circle,double,fill=white,
    minimum size=3mm,
    inner sep=0pt, drop shadow},
  intera/.style = {rectangle, draw=black, align=center, fill=white, rounded corners=0.1cm,
    minimum height=12,
    inner sep=2pt, drop shadow},
  line/.style = {draw,->, rounded corners=0.07cm,>=latex},
  venn/.style={preaction={fill, #1},opacity=0.6,anchor=south},
  exnode/.style={circle,draw=black,inner sep=1pt,fill=gray!20},
}

\begin{document}
\title{Verifying Asynchronous Interactions via Communicating Session
  Automata}

\renewcommand{\varnothing}{\emptyset}

\titlerunning{Verifying Asynchronous Interactions via Communicating Session
  Automata}
\author{Julien Lange\inst{1}
  \and
  Nobuko Yoshida \inst{2}}
\institute{University of Kent, Canterbury, United Kingdom 
  \and
  Imperial College London, London, United Kingdom}

{\def\addcontentsline#1#2#3{}\maketitle}

\begin{abstract}
  This paper proposes a sound procedure to verify properties of
communicating session automata (\CSA), i.e., communicating automata
that include multiparty session types.
We introduce a new \emph{asynchronous} compatibility property for
\CSA, called $k$-multiparty compatibility ($k$\MC), which is a strict
superset of the synchronous multiparty compatibility used in theories
and tools based on session types.
It is decomposed into two bounded properties: ($i$) a condition called
{\em $k$-safety} which guarantees that, within the bound, all sent
messages can be received and each automaton can make a move; 
and ($ii$) a condition called {\em $k$-exhaustivity} which guarantees
that all $k$-reachable send actions can be fired within the bound.
We show that $k$-exhaustivity implies existential boundedness, and
\emph{soundly and completely} characterises systems where each
automaton behaves equivalently under bounds greater than or equal to
$k$.
We show that checking $k$\MC\ is \PSPACE{}-complete, and demonstrate
its scalability empirically over large systems (using partial order
reduction).

   \end{abstract}

\section{Introduction}\label{sec:intro}

Communicating automata are a Turing-complete model of asynchronous
interactions~\cite{cfsm83} that has become one of the most prominent
for studying point-to-point communications over unbounded
first-in-first-out channels.
This paper focuses on a class of communicating automata, called
\emph{communicating session automata} (\CSA), which strictly includes
automata corresponding to \emph{asynchronous multiparty session
  types}~\cite{HondaYC08}.
Session types originated as a typing discipline for the
$\pi$-calculus~\cite{THK94,HVK98}, where a session type dictates the
behaviour of a process wrt.\ its communications.
Session types and related theories have been applied to the
verification and specification of concurrent and distributed systems
through their integration in several mainstream programming languages,
e.g., Haskell~\cite{LindleyM16,OrchardY16}, Erlang~\cite{NY2017},
F$\sharp$~\cite{NHYA2018},
Go~\cite{LangeNTY17,NgY16,LNTY2018,CastroHJNY19},
Java~\cite{HuY16,HY2017,KouzapasDPG16,SivaramakrishnanQZNE13},
OCaml~\cite{Padovani17}, C~\cite{NgYH12},
Python~\cite{DemangeonHHNY15,NY2017b,NBY2017},
Rust~\cite{JespersenML15}, and Scala~\cite{SY2016,ScalasDHY17}.
Communicating automata and asynchronous multiparty session
types~\cite{HondaYC08} are closely related: the latter can be seen as
a syntactical representation of the former~\cite{DenielouY12} where a
sending state corresponds to an internal choice and a
receiving state to an external choice.
This correspondence between communicating automata and multiparty
session types has become the foundation of many tools centred on
session types, e.g., for generating communication API from multiparty
session (global) types~\cite{HuY16,HY2017,ScalasDHY17,NHYA2018}, for
detecting deadlocks in message-passing
programs~\cite{TaylorTWD16,NgY16}, and for monitoring session-enabled
programs~\cite{BocchiCDHY17,NY2017,DemangeonHHNY15,NY2017b,NBY2017}.
These tools rely on a property called \emph{multiparty
  compatibility}~\cite{DY13,LTY15,BLY15}, which guarantees that
communicating automata representing session types interact correctly,
hence enabling the identification of correct protocols or the
detection of errors in endpoint programs.
Multiparty compatible communicating automata validate two essential
requirements for session types frameworks: every message that is sent
can be eventually received and each automaton can always eventually
make a move.
Thus, they satisfy the \emph{abstract} safety invariant $\asphi$ for
session types from~\cite{ScalasY19}, a prerequisite for session type
systems to guarantee safety of the typed processes.
Unfortunately, multiparty compatibility suffers from a severe
limitation: it requires that each execution of the system has a
synchronous equivalent. Hence, it rules out many correct systems.
Hereafter, we refer to this property as \emph{synchronous multiparty
  compatibility} (\SMC)
and explain its main limitation with Example~\ref{ex:intro}.

\begin{example}\label{ex:intro}
  The system in Figure~\ref{fig:running-example} contains an
  interaction pattern that is {\em not} supported by any definition of
  \SMC~\cite{DY13,LTY15,BLY15}.
    It consists of a client ($\ptp{c}$), a server ($\ptp{s}$), and a
  logger ($\ptp{l}$), which communicate via unbounded \FIFO\ channels.
    Transition $\PSEND{sr}{a}$ denotes that $\ptp{s}$ender puts
  (asynchronously) message $\msg{a}$ on channel $\ptp{sr}$; and
  transition $\PRECEIVE{sr}{a}$ denotes the consumption of $\msg{a}$
  from channel $\ptp{sr}$ by $\ptp{r}$eceiver.
        The $\ptp{c}$lient sends a $\msg{req}$uest and some $\msg{data}$ in
  a fire-and-forget fashion, before waiting for a response from the
  $\ptp{s}$erver.
                                        Because of the presence of this simple pattern, the system cannot be
  executed synchronously (i.e., with the restriction that a send
  action can only be fired when a matching receive is enabled), hence
    it is rejected by all definitions of \SMC\ from previous works, even
  though the system is safe (all sent messages are received and no
  automaton gets stuck).
              \end{example}
Synchronous multiparty compatibility is reminiscent of a strong form
of existential boundedness.
Among the existing sub-classes of communicating automata
(see~\cite{Muscholl10} for a survey), existentially $k$-bounded
communicating automata~\cite{GenestKM07} stand out because they can be
model-checked~\cite{bollig2010,GenestKM06} and they restrict the model
in a natural way: any execution can be rescheduled such that the
number of pending messages \emph{that can be received} is bounded by
$k$.
However, existential boundedness is generally
\emph{undecidable}~\cite{GenestKM07}, even for a fixed bound $k$. This
shortcoming makes it impossible to know when theoretical results are
applicable.

To address the limitation of \SMC\ and the shortcoming of existential
boundedness, we propose a (decidable) sufficient condition for
existential boundedness, called \emph{$k$-exhaustivity}, which serves
as a basis for a wider notion of new compatibility, called
\emph{$k$-multiparty compatibility} ($k$\MC) where
$k \in \naturals_{>0}$ is a bound on the number of pending messages in
each channel.
A system is $k$\MC\ when it is ($i$) \emph{$k$-exhaustive}, i.e., all
$k$-reachable send actions are enabled within the bound, and ($ii$)
\emph{$k$-safe}, i.e., within the bound $k$, all sent messages can be
received and each automaton can always eventually progress.
For example, the system in Figure~\ref{fig:running-example} is
$k$-multiparty compatible for any $k \in \naturals_{>0}$, hence it
does not lead to communication errors, see
Theorem~\ref{thm:soundness}.
The $k$\MC\ condition is a natural constraint for real-world
systems. 
Indeed any finite-state system is $k$-exhaustive (for $k$ sufficiently
large), while any system that is not $k$-exhaustive (resp.\ $k$-safe)
for any $k$ is unlikely to work correctly.
Furthermore, we show that if a system of \CSA\ validates $k$-exhaustivity,
then each automaton locally behaves equivalently under any bound
greater than or equal to $k$, a property that we call \emph{local
  bound-agnosticity}.
We give a \emph{sound and complete} characterisation of
$k$-exhaustivity for \CSA\ in terms of local
bound-agnosticity, see Theorem~\ref{thm:completeness}.
Additionally, we show that the complexity of checking $k$\MC\ is
\PSPACE-complete (i.e., no higher than related algorithms) and we
demonstrate empirically that its cost can be mitigated through (sound
and complete) partial order reduction.

In this paper, 
we consider \emph{communicating session automata} (\CSA), 
which cover the most common form of asynchronous multiparty session
types~\cite{CDPY2015} (see \emph{Remark~\ref{remark:basic-scribble}}),
and have been used as a basis to study properties and extensions of
session
types~\cite{DY13,HuY16,HY2017,LangeY16,BLY15,LY2017,NY2017,NY2017b,BYY14,NBY2017}.
More precisely, \CSA\ are deterministic automata, whose every state is
either sending (internal choice), receiving (external choice), or
final.
We focus on \CSA\ that preserve the intent of internal and external
choices from session types.
In these
\CSA, whenever an automaton is in a sending state, it can fire any
transition, no matter whether channels are bounded; when it is in a
receiving state then at most one action must be enabled.

\paragraph{Synopsis}
In \S~\ref{sec:cfsm}, we give the necessary background on
communicating automata and their properties, and introduce the
notions of output/input bound independence which guarantee that
internal/external choices are preserved in bounded semantics.
In \S~\ref{sec:kmc}, we introduce the definition of $k$-multiparty
compatibility ($k$\MC) and show that $k$\MC\ systems are safe for
systems which validate the bound independence properties.
In \S~\ref{sec:exist-bounded}, we formally relate existential
boundedness~\cite{kuske14,GenestKM07},
synchronisability~\cite{Bouajjani2018}, and $k$-exhaustivity.
In \S~\ref{sec:implementation} we present an implementation
(using partial order reduction) and an experimental evaluation of our
theory.
We discuss related works in \S~\ref{sec:related} and conclude in
\S~\ref{sec:conc}.

\iflong
Our implementation and benchmark data are available
online~\cite{kmc}.

\else
See~\cite{fullversion} for a full version of this paper (including
proofs and additional examples).
Our implementation and benchmark data are available
online~\cite{kmc}.
\fi

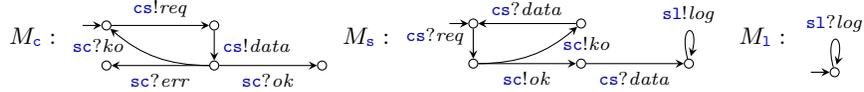
\begin{figure}[t]
  \centering
   \begin{tabular}{ccc}
    $ M_\ptp{c}:$
    \begin{tikzpicture}[mycfsm, node distance = 0.4cm and 1.3cm]
      \node[state, initial, initial where=left] (s0) {};
      \node[state, right=of s0] (s1) {};
      \node[state, below =of s1] (s2) {};
      \node[state, left=of s2] (s3) {};
      \node[state, right=of s2] (s4) {};
            \path
      (s0) edge node [above] {$\PSEND{cs}{req}$} (s1)
      (s1) edge node [right] {$\PSEND{cs}{data}$} (s2)
      (s2) edge [bend left=20] node [left, near end] {$\PRECEIVE{sc}{ko}$} (s0)
      (s2) edge node [below] {$\PRECEIVE{sc}{err}$} (s3)
      (s2) edge node [below] {$\PRECEIVE{sc}{ok}$} (s4)
      ;
    \end{tikzpicture}
    \quad
    &
      $ M_\ptp{s}:$
      \begin{tikzpicture}[mycfsm, node distance = 0.4cm and 1.3cm]
        \node[state, initial, initial where=left] (s0) {};
        \node[state, below=of s0] (s1) {};
        \node[state, right =of s0] (s2) {};
        \node[state, right=of s1] (s3) {};
        \node[state, right=of s3] (s4) {};
                        \path
        (s0) edge node [left, near start] {$\PRECEIVE{cs}{req}$} (s1)
        (s1) edge [bend right=20] node [right,near end] {$\PSEND{sc}{ko}$} (s2)
        (s2) edge node [above] {$\PRECEIVE{cs}{data}$} (s0)
        (s1) edge node [below] {$\PSEND{sc}{ok}$} (s3)
        (s3) edge node [below] {$\PRECEIVE{cs}{data}$} (s4)
        (s4) edge [loop above] node [above] {$\PSEND{sl}{log}$} (s4)
        ;
      \end{tikzpicture}
      \quad
    &
      $M_\ptp{l}:$
      \begin{tikzpicture}[mycfsm, node distance = 0.5cm and 1.3cm]
        \node[state, initial, initial where=left] (s0) {};
                \path
        (s0) edge [loop above] node [above] {$\PRECEIVE{sl}{log}$} (s0)
        ;
      \end{tikzpicture}
  \end{tabular}
          \caption{Client-Server-Logger example.}\label{fig:running-example}
\end{figure}

 \section{Communicating Automata and Bound Independence}
\label{sec:cfsm}
This section introduces notations and definitions of communicating
automata (following~\cite{LTY15,halfduplex}), as well as the notion of
output (resp.\ input) bound independence which enforces the intent of
internal (resp.\ external) choice in \CSA.

Fix a finite set $\PSet$ of \emph{participants} (ranged over by $\p$,
$\q$, $\rr$, $\ptp s$, etc.) and a finite alphabet $\ASigma$.
The set of \emph{channels} is
$\CSet \ \smash{\mmdef} \ \{\p\q \st \p,\q \in \PSet \text{ and }\p \neq \q \}$,
$\ASet \ \smash{\mmdef} \ \CSet \times \{!,?\} \times \ASigma$ is the set of
\emph{actions} (ranged over by $\action$), $\ASigmaC$ (resp. $\ASetC$)
is the set of finite words on $\ASigma$ (resp. $\ASet$).
Let $\word$ range over $\ASigmaC$, and $\acts, \actsb$ range over
$\ASetC$.
Also, $\emptyw$ ($\notin \ASigma \cup \ASet$) is the empty word,
$|\word|$ denotes the length of $\word$, and $\word \concat \word'$ is the
concatenation of $\word$ and $\word'$ (these notations are overloaded for
words in $\ASetC$).

\begin{definition}[Communicating automaton]\label{def:cfsm2} 
  A \emph{communicating automaton} is a finite transition
  system given by a triple $M = (Q,q_0,\delta)$
  where $Q$ is a finite set of {\em states}, $q_0\in Q$ is the initial
  state, and $\delta\ \subseteq \ Q \times \ASet \times Q$ is a set of
  \emph{transitions}.
\end{definition}

The transitions of a communicating automaton are labelled by actions
in $\ASet$ of the form $\PSEND{sr}{a}$, representing the
\emph{emission} of message $a$ from participant $\s$ to $\rr$, or
$\PRECEIVE{sr}{a}$ representing the \emph{reception} of $a$ by $\rr$.
Define $\subj{\PSEND{pq}{a}} = \subj{\PRECEIVE{qp}{a}} = \p$,
$\obj{\PSEND{pq}{a}} = \obj{\PRECEIVE{qp}{a}} = \q$, and
$\chan{\PSEND{pq}{a}} = \chan{\PRECEIVE{pq}{a}} = \ptp{pq}$.
The projection of $\action$ onto $\p$ is defined as
$\proj{\action}{p} = \action$ if $\subj{\action} = \p$ and
$\proj{\action}{p} = \emptyw$ otherwise.
Let $\dagger$ range over $\{!,?\}$, we define:
$\sndrcvproj{\PSENDRECEIVE{pq}{a}}{pq}{\dagger} = \msg{a}$
and
$\sndrcvproj{\PSENDRECEIVE{sr}{a}}{pq}{\dagger'}=\emptyw$ 
if either $\ptp{pq} \neq \ptp{sr}$ or $\dagger \neq \dagger'$.  We
extend these definitions to sequences of actions in the natural way.

A state $q\in Q$ with no outgoing transition is \emph{final}; $q$ is
\emph{sending} (resp.~\emph{receiving}) if it is not final and all its
outgoing transitions are labelled by send (resp. receive) actions,
and $q$ is \emph{mixed} otherwise.
$M=(Q, q_{0},\delta)$ is \emph{deterministic} if
$\forall (q,\action,q'), (q,\action',q'') \in \delta \qst \action =
\action' \xRightarrow{\quad} q' =q''$.
$M=(Q, q_{0},\delta)$ is \emph{send} (resp.\ \emph{receive})
\emph{directed} if for all sending (resp.\ receiving) $q \in Q$
and
$(q,\action,q'), (q,\action',q'') \in \delta \qst \obj{\action} =
\obj{\action'}$.
$M$ is \emph{directed} if it is send and receive directed.

\begin{remark}
  In this paper, we consider only deterministic communicating automata
  without mixed states, and call them \emph{Communicating Session
    Automata} (\CSA).
          We discuss possible extensions of our results beyond this class in
  Section~\ref{sec:conc}.
  \end{remark}

\begin{definition}[System]\label{def:cs}
  Given a communicating automaton $M_\p=(Q_\p, q_{0\p},\delta_\p)$
  for each $\p \in \PSet$, the tuple $S=(M_\p)_{\p\in\PSet}$ is a
  \emph{system}.
    A \emph{configuration} of $S$ is a pair $s = \csconf q w$ where
  $\vec{q}=(q_\p)_{\p\in\PSet}$ with $q_\p\in Q_\p$ and where
  $\vec{w}=(w_{\p\q})_{\p\q\in\CSet}$ with $w_{\p\q}\in
  \ASigma^\ast$; component $\vec q$ is the \emph{control state} and
  $q_\p\in Q_\p$ is the \emph{local state} of automaton $M_\p$. 
      The \emph{initial configuration} of $S$ is
  $s_0 = \stablecsconf{q_0}$ where
  $\vec{q_0} = (q_{0\p})_{\p\in\PSet}$ and we write ${\vecemptyw}$
  for the $\lvert \CSet \rvert$-tuple $(\emptyw, \ldots, \emptyw)$.
\end{definition}
Hereafter, we fix a communicating session automaton
$M_\p=(Q_\p, q_{0\p},\delta_\p)$ for each $\p \in \PSet$
and let $S=(M_\p)_{\p\in\PSet}$ be the corresponding system whose
initial configuration is $s_0$.
For each $\p \in \PSet$, we assume that
$\forall (q,\action,q') \in \delta_\p \qst \subj{\action} = \p$.
We assume that the components of a configuration are named
consistently, e.g., for $s' = \csconf{q'}{w'}$, we implicitly assume
that $\vec{q'}=(q'_\p)_{\p\in\PSet}$ and
$\vec{w'}=(w'_{\p\q})_{\p\q\in\CSet}$.

\begin{definition}[Reachable configuration]\label{def:rs}
  Configuration $s'=(\vec{q}';\vec{w}')$ is 
    \emph{reachable} from configuration $s=(\vec{q};\vec{w})$ by {\em
    firing transition $\action$}, written $s \TRANSS{\action} s'$ (or
  $s \TRANS{} s'$ when $\action$ is not relevant), if there are
  $\s, \rr \in \PSet$ and $a \in \ASigma$ such that either:
  \begin{enumerate}
  \item 
        (a) $\action= \PSEND{sr}{a}$ and $(q_\s,\action,q_\s')  \in \delta_{\s}$, 
        (b) $q_{\p}' = q_{\p}$ for all ${\p} \neq \s$, 
        (c) $w_{\s \rr}' = w_{\s \rr} \concat a$ and $w_{\p\q}'=w_{\p\q}$
    for all ${\p\q} \neq \s\rr$; or
          \item 
        (a) $\action= \PRECEIVE{sr}{a}$ and
    $(q_\rr,\action,q_\rr')\in \delta_\rr$,
            (b)  $q_{\p}' = q_{\p}$ for all ${\p} \neq \rr$,
        (c) $w_{\s \rr} = a \concat w_{\s \rr}' $, and
    $w_{\p\q}'=w_{\p\q}$ for all ${\p\q} \neq \ptp{sr}$.
  \end{enumerate} 
    \end{definition}

\begin{remark}
  Hereafter, we assume that any bound $k$ is finite and
  $k \in \naturals_{>0}$.
\end{remark}

We write $\TRANSR$ for the reflexive and transitive closure of
$\TRANSS{}$.
Configuration $\csconf{q}{w}$ is $k$-bounded if
$\forall \ptp{pq} \in \CSet \qst \lvert \word_{\p\q} \rvert \leq k$.
We write $s_1 \TRANSS{\action_1\cdots \action_n} s_{n+1}$ when
$s_1 \TRANSS{\action_1} s_2 \cdots s_n \TRANSS{\action_n} s_{n+1}$, for
some $s_2,\ldots,s_n$ (with $n \geq 0$); and
say that the \emph{execution} $\action_1\cdots \action_n$ is
\emph{$k$-bounded from} $s_1$ if $\forall 1 \leq i \leq n{+}1 \qst s_i$ is
$k$-bounded.
Given $\acts \in \ASetC$, we write $\p \notin \acts$ iff
$\acts = \acts_0 \concat \action \concat \acts_1 \implies
\subj{\action} \neq \p$.
We write $s \kTRANSS{\acts} s'$ if 
$s'$ is reachable with a $k$-bounded execution $\acts$ from $s$.
The set of \emph{reachable configurations of $S$} is $\RS(S)=\{
s \st s_0 \TRANSR s\}$.
The \emph{$k$-reachability set of $S$} is the largest subset
$\RS_k(S)$ of $\RS(S)$ within which each configuration $s$ can be
reached by a $k$-bounded execution from $s_0$.

Definition~\ref{def:k-safety} streamlines notions of safety from
previous works~\cite{BLY15,LTY15,DY13,halfduplex} (absence of
deadlocks, orphan messages, and unspecified receptions).
\begin{definition}[$k$-Safety]\label{def:k-safety}
  $S$ is $k$-\emph{safe} if the following holds 
  $ \forall \csconf q w \in \RS_k(S)$:
  
    \begin{description}
  \item[(\ER)]  \;     $\forall \p\q \in \CSet$, if $w_{\p\q} = \msg{a} \cdot w'$,
    then $\csconf q w \kTRANSR \kTRANSS{\PRECEIVE{pq}{a}}$.
  \item[(\PG)]  \;     $\forall\p \in \PSet$, if $q_\p$ is \emph{receiving},
    then $\csconf q w \kTRANSR \kTRANSS{\PRECEIVE{qp}{a}}$ for $\q \in
    \PSet$ and $\msg{a} \in \ASigma$.
  \end{description}
      We say that $S$ is \emph{safe} if it validates the unbounded version
  of $k$-safety ($\infty$-safe).
\end{definition}
Property (\ER), called \emph{eventual reception}, requires that any
sent message can always eventually be received (i.e., if $\msg{a}$ is
the head of a queue then there must be an execution that consumes
$\msg{a}$), and
Property (\PG), called \emph{progress}, requires
that any automaton in a receiving state can eventually make a move
(i.e., it can always eventually receive an \emph{expected} message).

We say that a configuration $s$ is \emph{stable} iff
$ s = \stablecsconf{q}$, i.e., all its queues are empty.
Next, we define the \emph{stable property} for systems of
communicating automata, following the definition from~\cite{DY13}.
\begin{definition}[Stable]\label{def:deadlockfree}
  $S$ has the \emph{stable property} (\DF) if
  $ \forall s \in \RS(S) \qst \exists \stablecsconf{q} \in \RS(S) \qst
  s \TRANSR \stablecsconf{q}$.
\end{definition}
A system has the stable property if it is possible to reach a stable
configuration from any reachable configuration. This property is
called \emph{deadlock-free} in~\cite{GenestKM07}.
The stable property implies the eventual reception property, but not
safety (e.g., an automaton may be waiting for an input in a stable
configuration, see Example~\ref {ex:dlf-not-safe}), and safety does
not imply the stable property, see Example~\ref{ex:fig-safe-notkmc}.

\begin{example}\label{ex:dlf-not-safe}
    The following system has the stable property, but it is not safe.
              \[
  \begin{array}{c@{\qquad}c@{\qquad}c}
    M_\s:
    \begin{tikzpicture}[mycfsm, node distance=0.6 and 1cm]
      \node[state, initial, initial where=above] (s2) {};
      \node[state, right=of s2] (s3) {};
      \node[state, left=of s2] (s4) {};
            \path
      (s2) edge  node [above] {$\PSEND{pq}{b}$} (s3)
      (s2) edge  node [above] {$\PSEND{pq}{a}$} (s4)
      ;
    \end{tikzpicture}
    \;
    &
      M_\ptp{q}:
      \begin{tikzpicture}[mycfsm, node distance=0.6 and 1cm]
        \node[state, initial, initial where=above] (s0) {};
        \node[state, left=of s0] (s1) {};
        \node[state, right=of s0] (s2) {};
        \node[state, right=of s2] (s3) {};
                \path
        (s0) edge node [above] {$\PRECEIVE{pq}{a}$} (s1)
        (s0) edge node [above] {$\PRECEIVE{pq}{b}$} (s2)
                (s2) edge node [above] {$\PSEND{qr}{c}$} (s3)
        ;
      \end{tikzpicture}
      \;
    & 
      M_\ptp{r}:
      \begin{tikzpicture}[mycfsm, node distance=0.6 and 1cm]
        \node[state, initial, initial where=above] (s0) {};
        \node[state, right=of s0] (s1) {};
                \path
        (s0) edge node [above] {$\PRECEIVE{qr}{c}$} (s1)
        ;
      \end{tikzpicture}
              \end{array}
  \]
        \end{example}

Next, we define two properties related to \emph{bound
  independence}. They specify classes of \CSA\ whose branching
behaviours are not affected by channel bounds.
\begin{definition}[$\OBI{k}$]\label{def:kobi}
  $S$ is $k$-\emph{output bound independent} ($\OBI{k}$), if
  $\forall s = \csconf{q}{w} \in \RS_k(S)$ and $\forall \p \in \PSet$, if
  $s \kTRANSS{\PSEND{pq}{a}}$, then
  $\forall (q_\p, \PSEND{pr}{b}, q'_\p) \in \delta_\p \qst
  s\kTRANSS{\PSEND{pr}{b}}$.
\end{definition}

\begin{definition}[$\IBI{k}$]\label{def:kibi}
  $S$ is $k$-\emph{input bound independent} ($\IBI{k}$), if 
  $\forall s = \csconf{q}{w} \in \RS_k(S)$ and $\forall \p \in \PSet$, if
  $s \kTRANSS{\PRECEIVE{qp}{a}}$, then
  $\forall \action \in \ASet \qst s \kTRANSS{\action} \, \land \,
  \subj{\action} = \p \xRightarrow{\quad} \action = \PRECEIVE{qp}{a}$.
\end{definition}

If $S$ is $\OBI{k}$, then any automaton that reaches a sending state
is able to fire any of its available transitions, i.e., sending states
model \emph{internal choices} which are not constrained by bounds
greater than or equal to $k$.
Note that the unbounded version of $\OBI{k}$ ($k=\infty$) is
trivially satisfied for any system due to unbounded asynchrony.
If $S$ is $\IBI{k}$, then any automaton that reaches a receiving state
is able to fire at most one transition, i.e., receiving states model
\emph{external choices} where the behaviour of the receiving automaton
is controlled exclusively by its environment.
We write $\infIBI$ for the unbounded version of $\IBI{k}$
($k = \infty$). 

Checking the $\infIBI$ property is generally undecidable.
However, systems consisting of (send and receive) \emph{directed}
automata are trivially $\IBI{k}$ and $\OBI{k}$ for all $k$,
this subclass of \CSA\ was referred to as \emph{basic}
in~\cite{DY13}.
We introduce larger decidable approximations of $\infIBI$ with
Definitions~\ref{def:non-csa-mc-dep} and~\ref{def:non-csa-mc}.

\begin{restatable}{proposition}{propdirectediobi}\label{prop:directed-iobi}
  (1)~If $S$ is send directed, then $S$ is $\OBI{k}$ for all
  $k \in \naturals_{>0}$.
    (2)~If $S$ is receive directed, then $S$ is $\infIBI$ (and $\IBI{k}$
  for all $k \in \naturals_{>0}$).
\end{restatable}

\begin{remark}
\label{remark:basic-scribble}
\CSA\ validating $\OBI{k}$ and $\infIBI$ strictly include the most
common forms of asynchronous multiparty session types, e.g., the
directed \CSA\ of~\cite{DY13}, and systems obtained by projecting
Scribble specifications (global types) which need to be receive
directed (this is called ``consistent external choice subjects''
in~\cite{HY2017}) and which validate $\OBI{1}$ by construction since
they are projections of synchronous specifications where choices must
be located at a unique sender.
  \end{remark}

 \section{Bounded Compatibility for \CSA}\label{sec:kmc}
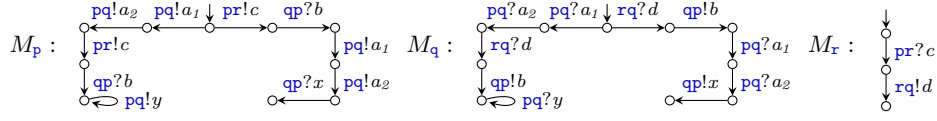
\begin{figure}[t]
  \centering
  \begin{tabular}{c@{\;}c@{\;}c}
    $M_\ptp{p}:$ \;
    \begin{tikzpicture}[mycfsm, node distance = 0.35cm and 0.7 cm]
      \node[state, initial, initial where=above,fill=white] (s0) {};
      \node[state, left=of s0] (sr1) {};
      \node[state, left=of sr1] (sr2) {};
      \node[state, below=of sr2] (sr3) {};
      \node[state, below=of sr3] (sr4) {};
                  \node[state, right=of s0] (sl1) {};
      \node[state, right=of sl1] (sl2) {};
      \node[state, below=of sl2] (sl3) {};
      \node[state, below=of sl3] (sl4) {};
      \node[state, left=of sl4,xshift=0pt] (sl5) {};
            \path
      (s0) edge  node [above] {$\PSEND{pq}{a_1}$} (sr1)
      (sr1) edge node [above] {$\PSEND{pq}{a_2}$} (sr2)
      (sr2) edge node [] {$\PSEND{pr}{c}$} (sr3)
      (sr3) edge node [] {$\PRECEIVE{qp}{b}$} (sr4)
      (sr4) edge [loop right] node [right] {$\PSEND{pq}{y}$} (sr4)
            (s0) edge  node [above] {$\PSEND{pr}{c}$} (sl1)
      (sl1) edge node [above] {$\PRECEIVE{qp}{b}$} (sl2)
      (sl2) edge node [] {$\PSEND{pq}{a_1}$} (sl3)
      (sl3) edge node [] {$\PSEND{pq}{a_2}$} (sl4)
      (sl4) edge node [above] {$\PRECEIVE{qp}{x}$} (sl5)
      ;
    \end{tikzpicture}
    &
      $M_\ptp{q}:$ \;
    \begin{tikzpicture}[mycfsm, node distance = 0.35cm and 0.7cm]
      \node[state, initial, initial where=above,fill=white] (s0) {};
      \node[state, left=of s0] (sr1) {};
      \node[state, left=of sr1] (sr2) {};
      \node[state, below=of sr2] (sr3) {};
      \node[state, below=of sr3] (sr4) {};
                  \node[state, right=of s0] (sl1) {};
      \node[state, right=of sl1] (sl2) {};
      \node[state, below=of sl2] (sl3) {};
      \node[state, below=of sl3] (sl4) {};
      \node[state, left=of sl4,xshift=0pt] (sl5) {};
            \path
      (s0) edge node [above] {$\PRECEIVE{pq}{a_1}$} (sr1)
      (sr1) edge node [above] {$\PRECEIVE{pq}{a_2}$} (sr2)
      (sr2) edge node [] {$\PRECEIVE{rq}{d}$} (sr3)
      (sr3) edge node [] {$\PSEND{qp}{b}$} (sr4)
      (sr4) edge [loop right] node [right] {$\PRECEIVE{pq}{y}$} (sr5)
            (s0) edge node [above] {$\PRECEIVE{rq}{d}$} (sl1)
      (sl1) edge node [above] {$\PSEND{qp}{b}$} (sl2)
      (sl2) edge node [] {$\PRECEIVE{pq}{a_1}$} (sl3)
      (sl3) edge node [] {$\PRECEIVE{pq}{a_2}$} (sl4)
      (sl4) edge node [above] {$\PSEND{qp}{x}$} (sl5)
      ;
    \end{tikzpicture}
    &
      $M_\ptp{r}:$ \;
      \begin{tikzpicture}[mycfsm, node distance = 0.35cm and 0.8cm]
        \node[state, initial, initial where=above] (s0) {};
        \node[state, below=of s0] (s1) {};
        \node[state, below=of s1] (s2) {};
                \path
        (s0) edge node [right] {$\PRECEIVE{pr}{c}$} (s1)
        (s1) edge node [right] {$\PSEND{rq}{d}$} (s2)
        ;
      \end{tikzpicture}
          \end{tabular}
  \caption{Example of a \emph{non}-\infIBI\ and \emph{non}-safe system.}
  \label{fig:ex-nondirected}
\end{figure}

In this section, we introduce \emph{$k$-multiparty compatibility}
($k$\MC) and study its properties wrt.\ safety of communicating
session automata (\CSA) which are $\OBI{k}$ and $\infIBI$.
Then, we soundly and completely characterise $k$-exhaustivity in terms
of local bound-agnosticity, a property which guarantees that
communicating automata behave equivalently under any bound greater
than or equal to $k$.

\subsection{Multiparty Compatibility}
The definition of $k$\MC\ is divided in two parts: ($i$)
\emph{$k$-exhaustivity} guarantees that the set of $k$-reachable
configurations contains enough information to make a sound decision
wrt.\ safety of the system; and ($ii$) \emph{$k$-safety}
(Definition~\ref{def:k-safety}) guarantees that a subset of all
possible executions is free of any communication errors.
Next, we define $k$-exhaustivity, then $k$-multiparty compatibility.
Intuitively, a system is $k$-exhaustive if for all $k$-reachable
configurations, whenever a send action is enabled, then it can be
fired within a $k$-bounded execution.

\begin{definition}[$k$-Exhaustivity]\label{def:exhaustive}
  $S$ is $k$-\emph{exhaustive} if
  $\forall  \csconf{q}{w} \in \RS_k(S)$ and $\forall \p \in \PSet$,
    if $q_\p$ is \emph{sending}, then
    $
  \forall (q_\p, \action, q'_\p) \in \delta_\p \qst \exists
  \acts \in \ASetC
  \qst
  \csconf{q}{w}  \kTRANSS{\acts}\kTRANSS{\action}  \land \p \notin
  \acts.
  $
    \end{definition}
\begin{definition}[$k$-Multiparty compatibility]\label{def:compa}
    $S$ is $k$-\emph{multiparty compatible} ($k$\MC) if it is
    $k$-safe and $k$-exhaustive.
    \end{definition}
Definition~\ref{def:compa} is a natural extension of the definitions
of \emph{synchronous} multiparty compatibility given
in~\cite[Definition 4.2]{DY13} and~\cite[Definition 4]{BLY15}.
The common key requirements are that \emph{every send} action must be
matched by a receive action (i.e., send actions are universally
quantified), while \emph{at least one receive} action must find a
matching send action (i.e., receive actions are existentially
quantified).
Here, the universal check on send actions is done via the eventual
reception property (\ER) and the $k$-exhaustivity condition; while the
existential check on receive actions is dealt with by the progress
property (\PG).

Whenever systems are $\OBI{k}$ and $\infIBI$, then $k$-exhaustivity
implies that $k$-bounded executions are sufficient to make a sound
decision wrt.\ safety.
This is not necessarily the case for systems outside of this class, see
 Examples~\ref{ex:non-dir-not-safe} and~\ref{ex:not-obi}.

\begin{example}\label{ex:non-dir-not-safe}
  The system $(M_\p,M_\q,M_\rr)$ in Figure~\ref{fig:ex-nondirected} is
  $\OBI{k}$ for any $k$, but not \infIBI\ (it is $\IBI{1}$ but not
  $\IBI{k}$ for any $k \geq 2$).
    When executing with a bound strictly greater than $1$, there is a
  configuration where $M_\q$ is in its initial state and \emph{both}
  its \emph{receive} transitions are enabled.
      The system is $1$-safe and $1$-exhaustive (hence $1$\MC) but it is
  \emph{not} $2$-exhaustive nor $2$-safe.
    By constraining the automata to execute with a channel bound of $1$,
  the left branch of $M_\p$ is prevented to execute together with the
  right branch of $M_\q$.
    Thus, the fact that the $\msg{y}$ messages are not received in this
  case remains invisible in $1$-bounded executions.
    This example
    can be easily extended so that it is $n$-exhaustive (resp.\ safe)
  but not $n{+1}$-exhaustive (resp.\ safe) by sending/receiving
  $n{+}1$ $\msg{a_i}$ messages.
\end{example}

\begin{figure}[t]
  \centering
  \begin{tabular}{c@{\qquad}c@{\quad}c@{\quad}c}
        $M_\p$:
    \begin{tikzpicture}[mycfsm, node distance=0.5cm and 0.7cm]
      \node[state, initial, initial where=above] (s0) {};
      \node[state, right=of s0] (s1) {};
      \node[state, below=of s0] at  ($(s0)!0.5!(s1)$) (s2) {};
            \path
      (s0) edge node [above] {$\PSEND{pq}{a}$} (s1)
      (s1) edge node [right] {$\PSEND{pq}{a}$} (s2)
      (s2) edge node [left] {$\PRECEIVE{qp}{b}$} (s0)
      ;
    \end{tikzpicture}
    &
      $M_\q$:
      \begin{tikzpicture}[mycfsm, node distance=0.5cm and 0.7cm]
        \node[state, initial, initial where=above] (s0) {};
        \node[state, below=of s0] (s1) {};
                \path
        (s0) edge [bend left] node [] {$\PSEND{qp}{b}$} (s1)
        (s1) edge [bend left] node [] {$\PRECEIVE{pq}{a}$} (s0)
        ;
      \end{tikzpicture}
    &
      $N_\q$:
      \begin{tikzpicture}[mycfsm, node distance=0.5cm and 0.7cm]
        \node[state, initial, initial where=above] (s0) {};
        \node[state, right=of s0] (s1) {};
        \node[state, below=of s0] at  ($(s0)!0.5!(s1)$) (s2) {};
                \path
        (s0) edge node [above] {$\PSEND{qp}{b}$} (s1)
        (s1) edge node [right] {$\PRECEIVE{pq}{a}$} (s2)
        (s2) edge node [left] {$\PRECEIVE{pq}{a}$} (s0)
        ;
      \end{tikzpicture}
    &
      $N'_\q$: \
      \begin{tikzpicture}[mycfsm, node distance=0.5cm and 0.7cm]
        \node[state] (s0) {};
        \node[state, right=of s0] (s1) {};
        \node[state, below=of s0] at  ($(s0)!0.5!(s1)$) (s2) {};
        \node[state, left=of s0, initial, initial where=above] (s3) {};
                \path
        (s3) edge node [above] {$\PSEND{qp}{b}$} (s0)
        (s0) edge node [above] {$\PSEND{qp}{b}$} (s1)
        (s1) edge node [right] {$\PRECEIVE{pq}{a}$} (s2)
        (s2) edge node [left] {$\PRECEIVE{pq}{a}$} (s0)
        ;
      \end{tikzpicture}
  \end{tabular}
  \caption{$(M_\p, M_\q)$ is non-exhaustive, $(M_\p, N_\q)$ is
    $1$-exhaustive, $(M_\p, N'_\q)$ is $2$-exhaustive.} 
  \label{fig:ex-unbounded}
\end{figure}
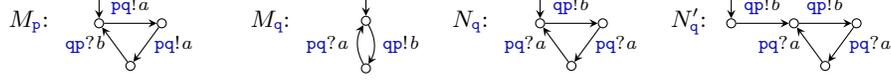

\begin{example}\label{ex:fig-safe-notkmc}
  The system in Figure~\ref{fig:running-example} is \emph{directed}
  and $1$\MC.
    The system $(M_\p, M_\q)$ in Figure~\ref{fig:ex-unbounded}
  is safe but \emph{not} $k$\MC\ for any finite $k \in
  \naturals_{>0}$.
    Indeed, for any execution of this system, at least one of the queues
  grows arbitrarily large.
    The system $(M_\p, N_\q)$ is $1$\MC\, while the system
  $(M_\p, N'_\q)$ is \emph{not} $1$\MC\ but it is $2$\MC.
\end{example}
 
\begin{example}\label{ex:not-obi}
  The system in Figure~\ref{fig:ex-obi-main} (without the dotted
  transition) is $1$\MC, but not $2$-safe; it is not $\OBI{1}$ but it is
  $\OBI{2}$.
    In $1$-bounded executions, $M_\rr$ can execute
  $\PSEND{rs}{b} \concat \PSEND{rp}{z}$, but it cannot fire
  $\PSEND{rs}{b} \concat \PSEND{rs}{a}$ (queue $\ptp{rs}$ is
  full), which
    violates the $\OBI{1}$ property.
              The system with the dotted transition is not $\OBI{1}$, but it is
  $\OBI{2}$ and $k$\MC\ for any $k \geq 1$.
    Both systems are receive directed, hence $\infIBI$.
\end{example}

\begin{restatable}{theorem}{thmsoundness}\label{thm:soundness}
  If $S$ is $\OBI{k}$, $\infIBI$, and $k$\MC, then it is safe.
\end{restatable}

\begin{remark}
  It is undecidable whether there exists a bound $k$ for which an
  arbitrary system is $k$\MC. This is a consequence of the Turing
  completeness of communicating (session)
  automata~\cite{cfsm83,FinkelM97,LY2017}.
\end{remark}

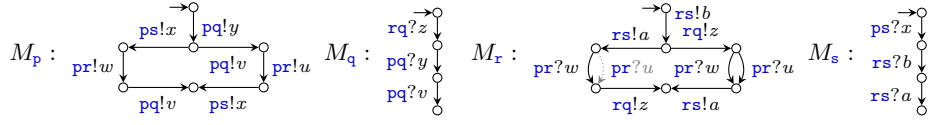
\begin{figure}[t]
\centering
\begin{tabular}{c@{\,}c@{\,}c@{\,}c} 
    $M_\ptp{p}:$     \begin{tikzpicture}[mycfsm, node distance = 0.4cm and 0.8cm]
      \node[state, initial, initial where=left] (s0) {};
      \node[state, below=of s0] (s1) {};
      \node[state, right=of s1] (s2) {};
      \node[state, below=of s2] (s6) {};
      \node[state, left=of s6] (s4) {};
      \node[state, left=of s1] (s3) {};
      \node[state, below=of s3] (s5) {};
            \path
      (s0) edge node [right] {$\PSEND{pq}{y}$} (s1)
      (s1) edge node [below] {$\PSEND{pq}{v}$} (s2)
      (s1) edge node [above] {$\PSEND{ps}{x}$} (s3)
      (s2) edge node [right] {$\PSEND{pr}{u}$} (s6)
      (s6) edge node [below] {$\PSEND{ps}{x}$} (s4)
      (s5) edge node [below] {$\PSEND{pq}{v}$} (s4)
      (s3) edge node [left] {$\PSEND{pr}{w}$} (s5)
      ;
    \end{tikzpicture}
    &
      $M_\ptp{q}:$       \begin{tikzpicture}[mycfsm, node distance = 0.3cm and 0.4cm]
        \node[state, initial, initial where=left] (s0) {};
        \node[state, below=of s0] (s1) {};
        \node[state, below=of s1] (s2) {};
        \node[state, below=of s2] (s3) {};
                \path
        (s0) edge node [left] {$\PRECEIVE{rq}{z}$} (s1)
        (s1) edge node [left] {$\PRECEIVE{pq}{y}$} (s2)
        (s2) edge node [left] {$\PRECEIVE{pq}{v}$} (s3)
        ;
      \end{tikzpicture}
    &
      \quad
      $M_\ptp{r}:$       \begin{tikzpicture}[mycfsm, node distance = 0.4cm and 0.8cm]
        \node[state, initial, initial where=left] (s0) {};
        \node[state, below=of s0] (s1) {};
        \node[state, right=of s1] (s2) {};
        \node[state, below=of s2] (s6) {};
        \node[state, left=of s6] (s4) {};
        \node[state, left=of s1] (s3) {};
        \node[state, below=of s3] (s5) {};
                \path
        (s0) edge node [right, near start,yshift=3pt] {$\PSEND{rs}{b}$} (s1)
                (s1) edge node [above] {$\PSEND{rq}{z}$} (s2)
                (s2) edge [bend left] node [right] {$\PRECEIVE{pr}{u}$} (s6)
        (s2) edge [bend right] node [left] {$\PRECEIVE{pr}{w}$} (s6)
                (s6) edge node [below] {$\PSEND{rs}{a}$} (s4)
                (s1) edge node [above] {$\PSEND{rs}{a}$} (s3)
        (s3) edge [black!50, densely dotted, bend left] node [right] {$\PRECEIVE{pr}{u}$} (s5)
        (s3) edge [bend right] node [left] {$\PRECEIVE{pr}{w}$} (s5)
        
        (s5) edge node [below] {$\PSEND{rq}{z}$} (s4)
        ;
      \end{tikzpicture}
  &
    $M_\ptp{s}:$     \begin{tikzpicture}[mycfsm, node distance = 0.3cm and 0.3cm]
      \node[state, initial, initial where=left] (s0) {};
      \node[state, below=of s0] (s1) {};
      \node[state, below=of s1] (s2) {};
      \node[state, below=of s2] (s3) {};
            \path
      (s0) edge node [left] {$\PRECEIVE{ps}{x}$} (s1)
      (s1) edge node [left] {$\PRECEIVE{rs}{b}$} (s2)
      (s2) edge node [left] {$\PRECEIVE{rs}{a}$} (s3)
      ;
    \end{tikzpicture}
    \end{tabular}
\caption{Example of a system which is not $\OBI{1}$.}
\label{fig:ex-obi-main}
\end{figure}

Although the $\infIBI$ property is generally undecidable, it is
possible to identify sound approximations, as we show below.
We adapt the dependency relation from~\cite{LTY15} and say that action
$\action'$ depends on $\action$ from $s = \csconf{q}{w}$, written
$\bindep{s}{\action}{\action'}$, iff
$\subj{\action} = \subj{\action'} \lor 
(\chan{\action} = \chan{\action'} \land \word_{\chan{\action}} =
\emptyw)$.
Action $\action'$ depends on $\action$ in $\acts$ from $s$, written
$\dependsin{s}{\action}{\acts}{\action'}$, if the following holds:
\[
\dependsin{s}{\action}{\acts}{\action'}
\iff 
\begin{cases}
  (
  \bindep{s}{\action}{\action''}
  \land
  \dependsin{s}{\action''}{\actsb}{\action'}
  )
  \lor
  \dependsin{s}{\action}{\actsb}{\action'}
  & \text{if } \acts = \action'' \concat \actsb
  \\
  \bindep{s}{\action}{\action'} & \text{otherwise}
\end{cases}
\]

\begin{definition}\label{def:non-csa-mc-dep}
  $S$ is $k$-\emph{chained input bound independent}
  ($\CIBI{k}$) if $\forall s = \csconf{q}{w} \in \RS_k(S)$ and
  $\forall \p \in \PSet$, if $s \kTRANSS{\PRECEIVE{qp}{a}} s'$, then
  $ \forall (q_\p, \PRECEIVE{sp}{b}, q'_\p) \in \delta_\p \qst \s \neq
  \q \implies \!  
  \neg ( s \kTRANSS{\PRECEIVE{sp}{b}})
  \land
  (
  \forall \acts \in \ASetC \qst 
  s' \kTRANSS{\acts} \kTRANSS{\PSEND{sp}{b}}
  \implies \dependsin{s}{\PRECEIVE{qp}{a}}{\acts}{\PSEND{sp}{b}})
    $.
\end{definition}

\begin{definition}\label{def:non-csa-mc}
  $S$ is $k$-\emph{strong input bound independent}
  ($\DIBI{k}$) if $\forall s = \csconf{q}{w} \in \RS_k(S)$ and
  $\forall \p \in \PSet$, if $s \kTRANSS{\PRECEIVE{qp}{a}} s'$, then
  $ \forall (q_\p, \PRECEIVE{sp}{b}, q'_\p) \in \delta_\p \qst \s \neq
  \q \implies \!  \neg ( s \kTRANSS{\PRECEIVE{sp}{b}} \, \lor \, s'
  \kTRANSR \kTRANSS{\PSEND{sp}{b}})$.
\end{definition}

Definition~\ref{def:non-csa-mc-dep} requires that whenever $\p$ can
fire a receive action, at most one of its receive actions is enabled
at $s$, and no other receive transition from $q_\p$ will be enabled
until $\p$ has made a move. This is due to the existence of a
dependency chain between the reception of a message
(${\PRECEIVE{qp}{a}}$) and the matching send of another possible
reception (${\PSEND{sp}{b}}$).
Property $\DIBI{k}$ (Definition~\ref{def:non-csa-mc}) is a
stronger version of $\CIBI{k}$,
which can be checked more efficiently.
\begin{restatable}{lemma}{lemkrikexhimpinfripBOTH}\label{lem:kri-kexh-imp-infrip-both}
  If $S$ is $\OBI{k}$, $\CIBI{k}$ (resp.\ $\DIBI{k}$) and
  $k$-exhaustive, then it is $\infIBI$.
\end{restatable}

The decidability of $\OBI{k}$, $\IBI{k}$, $\DIBI{k}$, $\CIBI{k}$,
and $k$\MC\ is straightforward since both $\RS_k(S)$ (which
has an exponential number of states wrt. $k$) and $\kTRANSS{}$ are
finite, given a finite $k$.
Theorem~\ref{thm:decidability-all} states the space complexity of the
procedures, except for $\CIBI{k}$ for which a complexity class is yet
to be determined.
We show that the properties are \PSPACE{} by reducing to an instance
of the reachability problem over a transition system built following
the construction of Bollig et al.\cite[Theorem 6.3]{bollig2010}.
The rest of the proof follows from similar arguments in Genest et
al.~\cite[Proposition 5.5]{GenestKM07} and Bouajjani et
al.~\cite[Theorem 3]{Bouajjani2018}.
 
\begin{restatable}{theorem}{thmdecidabilityall}\label{thm:decidability-all}
    The problems of checking the $\OBI{k}$, $\IBI{k}$, $\DIBI{k}$,
  $k$-safety, and $k$-exhaustivity properties are all decidable and
  \PSPACE{}-complete (with $k \in \naturals_{>0}$ given in unary).
    The problem of checking the $\CIBI{k}$ property is decidable.
\end{restatable}

\subsection{Local Bound-Agnosticity}\label{sec:new-completeness}
We introduce local bound-agnosticity and show that it fully
characterises $k$-exhaustive systems.
Local bound-agnosticity guarantees that each communicating automaton
behave in the same manner for any bound greater than or equal to some
$k$.
Therefore such systems may be executed transparently under a bounded
semantics (a communication model available in Go and Rust).
\begin{definition}[Transition system]\label{def:kts}
  The $k$-bounded transition system of $S$ is the labelled transition
  system (LTS) $\kTS{S} = (N, s_0, \Delta)$ such that $N = \RS_k(S)$, $s_0$ is
  the initial configuration of $S$,
  $\Delta \subseteq N {\times} \ASet {\times} N$ is the transition
  relation, and $(s, \action, s') \in \Delta$ if and only if
  $s \kTRANSS{\action} s'$.
    \end{definition}

\begin{definition}[Projection]\label{def:kts-projections}
  Let $\lts$ be an LTS over $\ASet$.
      The \emph{projection} of $\lts$ onto $\p$, written
  $\epsproj{\lts}{p}$, is obtained by replacing each label $\action$
  in $\lts$ by $\proj{\action}{p}$.
      \end{definition}
Recall that the projection of action $\action$, written
$\proj{\action}{p}$, is defined in Section~\ref{sec:cfsm}.
The automaton $\epsproj{\kTS{S}}{p}$ is essentially the \emph{local}
behaviour of participant $\p$ within the transition system $\kTS{S}$.
When each automaton in a system $S$ behaves equivalently for any bound
greater than or equal to some $k$, we say that $S$ is {locally
  bound-agnostic}.
Formally, $S$ is \emph{locally bound-agnostic for $k$} when
$\epsproj{\kTS{S}}{p}$ and $\epsproj{\TS{n}{S}}{p}$ are weakly
bisimilar ($\wbisim$) for each participant $\p$ and any $n \geq k$.
For $\OBI{k}$ and $\infIBI$ systems, local bound-agnosticity is a
\emph{necessary and sufficient} condition for $k$-exhaustivity, as
stated in Theorem~\ref{thm:completeness} and
Corollary~\ref{cor:wbisim-csa}.

\begin{restatable}{theorem}{thmcompleteness}\label{thm:completeness}
  Let $S$ be a system.  
    \begin{itemize}
  \item[(1)\!] If
    $\exists k \in \naturals_{>0} \qst \forall \p \in \PSet \qst
    \epsproj{\kTS{S}}{p} \wbisim \epsproj{\TS{k{+}1}{S}}{p}$,
    then $S$ is $k$-exhaustive.
  \item[(2)\!] If $S$ is $\OBI{k}$, $\infIBI$, and $k$-exhaustive, then
    $\forall \p \in \PSet  \! \qst \! \epsproj{\kTS{S}}{p} \wbisim
    \epsproj{\TS{k{+}1}{S}}{p}$.
  \end{itemize}
      \end{restatable}

\begin{restatable}{corollary}{corbisimcsa}\label{cor:wbisim-csa}
  Let
  $S$ be $\OBI{k}$ and $\infIBI$ s.t.\
          $\forall \p \in \PSet \qst
  \epsproj{\kTS{S}}{p} \wbisim \epsproj{\TS{k{+}1}{S}}{p}$, 
    then $S$ is locally bound-agnostic for $k$.

    \end{restatable}

Theorem~\ref{thm:completeness} (1) is reminiscent of the
(\PSPACE-complete) checking procedure for existentially bounded
systems with the stable property~\cite{GenestKM07} (an
\emph{undecidable} property).
Recall that $k$-exhaustivity is not sufficient to guarantee safety,
see Examples~\ref{ex:non-dir-not-safe} and~\ref{ex:not-obi}.
We give an effective procedure (based on partial order reduction) to
check $k$-exhaustivity and related properties in
\appendixref{app:por}.

 \section{Existentially Bounded and Synchronisable Automata}\label{sec:exist-bounded}

\subsection{Kuske and Muscholl's Existential Boundedness}\label{sub:kuske-exist}
Existentially bounded communicating
automata~\cite{GenestKM06,GenestKM07,kuske14} are a class of
communicating automata whose executions can always be scheduled in
such a way that the number of pending messages is bounded by a given
value.
Traditionally, existentially bounded communicating automata are
defined on communicating automata that feature (local) accepting
states and in terms of \emph{accepting runs}. An accepting run is an
execution (starting from $s_0$) which terminates in a configuration
$\csconf{q}{w}$ where each $q_\p$ is a local accepting state.
In our setting, we simply consider that every local state $q_\p$ is an
accepting state, hence any execution $\acts$ starting from $s_0$ is an
accepting run.
We first study existential boundedness as defined in~\cite{kuske14} as
it matches more closely $k$-exhaustivity, we study 
the ``classical'' definition of existential
boundedness~\cite{GenestKM07} in Section~\ref{sub:classical-exist}.

Following~\cite{kuske14}, we say that an execution $\acts \in \ASetC$
is \emph{valid} if for any prefix $\actsb$ of $\acts$ and any channel
$\ptp{pq} \in \CSet$, we have that
$\ercvproj{\actsb}{pq}$ is a prefix of $\esndproj{\actsb}{pq}$, i.e.,
an execution is valid if it models the {\sc fifo}
semantics of communicating automata.
\begin{definition}[Causal equivalence~\cite{kuske14}]
    Given $\acts, \actsb \in \ASetC$, we define:
  ${\acts}\resche{\actsb}$ \emph{iff}
  $\acts$ and $\actsb$ are \emph{valid} executions
    and
    $\forall \p \in \PSet \qst
  \onpeer{p}{\acts} = \onpeer{p}{\actsb}$.
  We write $\equivclass{\acts}{\resche}$ for the equivalence class of
$\acts$ wrt.\ $\resche$.
\end{definition}

\begin{definition}[Existential boundedness~\cite{kuske14}]\label{def:exist-bounded}
    We say that a valid execution $\acts$ is \emph{$k$-\mbounded{}} if, for every
  prefix $\actsb$ of $\acts$ the difference between the number of
  \emph{matched} events of type $\PSEND{pq}{}$ and those of type
  $\PRECEIVE{pq}{}$ is bounded by $k$, i.e.,
    $  
  \mathit{min}\{
  \lvert \esndproj{\actsb}{pq} \rvert 
  , 
  \lvert \ercvproj{\acts}{pq} \rvert 
  \}
  -
  \lvert \ercvproj{\actsb}{pq} \rvert 
  \leq k
  $.

  \noindent
  Write $\bounded{\ASetC}{k}$ for the set of $k$-\mbounded{} words.
      An execution $\acts$ is \emph{existentially} $k$-bounded if
  $\equivclass{\acts}{\resche} \cap \bounded{\ASetC}{k} \, \neq
  \varnothing$.
      A system $S$ is existentially $k$-bounded, written
  $\exists$-$k$-bounded, if each execution in
  $\{\acts \st \exists s \qst s_0 \smash{\TRANSS{\acts}} s \}$ is
  existentially $k$-bounded.
\end{definition}

\newcounter{ExampleCounter}

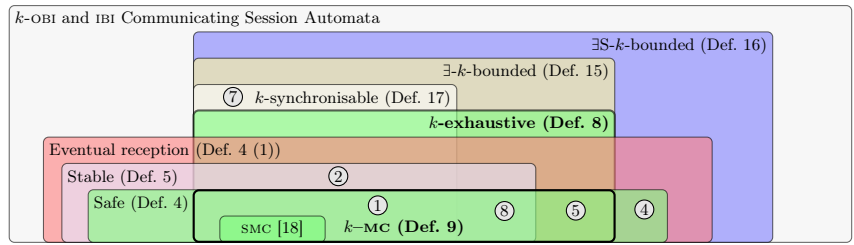
\begin{figure}[t]
  \centering
  \begin{tikzpicture}[scale=0.7, every node/.style={transform shape}]
        \node (cfsm) at (0,0) [draw=black,minimum width=16cm,minimum
    height=4.5cm,venn=gray!10,rounded corners=2pt] {}; \node[below
    right] at (cfsm.north west) {{$\OBI{k}$ and $\infIBI$ Communicating Session Automata}};
        \node (cexist) at (1,0) [draw=black,minimum width=11cm,minimum height=4cm,venn=blue!50,rounded corners=2pt] {};
    \node[below left] at (cexist.north east) {{{$\exists$S-$k$-bounded}} (Def.~\ref{def:classical-exist-bounded})};
        \node (exist) at (-0.5,0) [draw=black,minimum width=8cm,minimum height=3.5cm,venn=yellow!50,rounded corners=2pt] {};
    \node[below left] at (exist.north east) {{{$\exists$-$k$-bounded}}
      (Def.~\ref{def:exist-bounded})};
        \node (synk) at (-2,0) [draw=black,minimum width=5cm,minimum
    height=3cm,venn=white!50,rounded corners=2pt] {};
    \node[below left] at (synk.north east) {{{$k$\synk}}
      (Def.~\ref{def:synk})};
        \node (exh) at (-0.5,0) [draw=black,minimum width=8cm,minimum
    height=2.5cm,venn=green!50,text width=1cm,align=center,rounded
    corners=2pt,thick] {};
            \node[below left] at (exh.north east) {\textbf{$k$-exhaustive
          (Def.~\ref{def:exhaustive})}};
        \node (evtreception) at (-1,0) [draw=black, minimum
    width=12.7cm,minimum height=2cm,venn=red!50,rounded corners=2pt]
    {};
    \node[below right] at (evtreception.north west) {{{Eventual reception}}
      (Def.~\ref{def:k-safety} (1))};
        \node (safe) at (-2.5,0) [draw=black, minimum width=9cm,minimum height=1.5cm,venn=blue!10,rounded corners=2pt] {};
        \node[below right] at (safe.north west) (safetxt)
    {{{Stable}} (Def.~\ref{def:deadlockfree})};
        \node  at (-1,0) [draw=black, minimum width=11cm,minimum height=1cm,venn=green!50,rounded corners=2pt]
    (safetynode) {};
                        \node[below right] at (safetynode.north west) (safetytxt)
    {{{Safe}} (Def.~\ref{def:k-safety})};
        \node  at (-3,0) [draw=black, minimum width=2cm,minimum height=0.5cm,venn=green!50,rounded corners=2pt]
    (smcnode) {};
        \node at (smcnode) (smctxt) {{\color{black}\SMC{\footnotesize \cite{DY13}}}};
        \draw[black, thick,rounded corners=2pt]
    ([yshift=-0.5pt,xshift=-0.5pt]safetynode.north east-|exh.north east)
    --
    ([yshift=-0.5pt,xshift=0.5pt]safetynode.north west-|exh.north west)
    --
    ([yshift=0.5pt,xshift=0.5pt]safetynode.south west-|exh.south west)
    --
    ([yshift=0.5pt,xshift=-0.5pt]safetynode.south east-|exh.south east)
    -- cycle
    ;
        \node[above] at ($(safetynode.south east-|exh.south
    east)!0.5! (safetynode.south west-|exh.south west)$) 
    (kmctitle) {\textbf{$k$-\MC\ (Def.~\ref{def:compa})}};
        \node[exnode] at ([xshift=-2.5cm]$(exh.north)!0.5!(synk.north)$) {\exref{extbl:ksynk}};
    \node[exnode]  at ($(safetynode.north)!0.3!(safetynode.south)$)  (exrunning) {\exref{extbl:kmc}};     \node[exnode]  at ($(synk.east)!0.6!(safe.south east)$) (exnosynk)  {\exref{extbl:nosynk}};
            \node[exnode]  at ([yshift=-0.9cm]$(evtreception.east)!0.7!(exist.east)$) (exunboundedsafe) {\exref{extbl:unboundedsafe}};
    \node[exnode]  at ([yshift=-0.4cm]$(safe.east)!0.5!(exh.east)$) {\exref{extbl:safenotstable}};
        \node[exnode]  at ($(safe.north)!0.5!(safetynode.north)$)
    {\exref{extbl:dlfnotsafe}};
                  \end{tikzpicture} 
    \caption{Relations between $k$-exhaustivity, existential
    $k$-boundedness, and $k$-synchronisability in $\OBI{k}$ and
    $\infIBI$ \CSA\ (the circled numbers refer to
    Table~\ref{tbl:examples}).  }\label{fig:venn}
\end{figure}


\begin{example}
  Consider Figure~\ref{fig:ex-unbounded}.
    $(M_\p, M_\q)$ is \emph{not} existentially $k$-bounded, for any $k$:
  at least one of the queues must grow infinitely for the
  system to progress.
    Systems $(M_\p, N_\q)$ and $(M_\p, N'_\q)$ are existentially bounded
  since any of their executions can be scheduled to an
  $\resche$-equivalent execution which is $2$-\mbounded{}.
\end{example}

The relationship between $k$-exhaustivity and existential boundedness
is stated in Theorem~\ref{thm:existentially-iff-kmc-notreduced} and
illustrated in Figure~\ref{fig:venn} for $\OBI{k}$ and $\infIBI$ \CSA, where
\SMC\ refers to synchronous multiparty compatibility~\cite[Definition
4.2]{DY13}.
The circled numbers in the figure refer to key examples summarised in
Table~\ref{tbl:examples}.
The strict inclusion of $k$-exhaustivity in existential
$k$-boundedness is due to systems that do not have the eventual
reception property, see
Example~\ref{ex:exist-not-exh}.

\begin{example}\label{ex:exist-not-exh}
      The system below is $\exists$-1-bounded but is \emph{not}
  $k$-exhaustive for any $k$.
      \begin{center}
    $\begin{array}{c@{\qquad}c@{\qquad}c}
   M_\ptp{p}:
   \begin{tikzpicture}[mycfsm]
     \node[state, initial, initial where=left] (s0) {};
          \path
     (s0) edge [loop right,looseness=60] node [right] {$\PRECEIVE{sp}{{\color{red}c}}$} (s0)
          ;
   \end{tikzpicture}
   &
     M_\s:
     \begin{tikzpicture}[mycfsm]
       \node[state, initial, initial where=left] (s0) {};
       \node[state, right=of s0] (s1) {};
              \path
       (s0) edge [bend left] node [above] {$\PSEND{sr}{a}$} (s1)
       (s1) edge [bend left] node [below] {$\PSEND{sp}{b}$} (s0)
       ;
     \end{tikzpicture}
   &
     M_\ptp{r}:
     \begin{tikzpicture}[mycfsm]
       \node[state, initial, initial where=left] (s0) {};
              \path
       (s0) edge [loop right,looseness=60] node [right] {$\PRECEIVE{sr}{a}$} (s0)
       ;
     \end{tikzpicture}
 \end{array}
  $

  \end{center}
  \noindent
  For any $k$, the channel $\ptp{sp}$ eventually gets full
  and the send action $\PSEND{sp}{b}$ can no longer be fired;
  hence it does \emph{not} satisfy $k$-exhaustivity.
    Note that each execution can be reordered into a $1$-\mbounded{}
  execution (the $\msg{b}$'s are never matched).
\end{example}

\begin{restatable}{theorem}{thmexistentiallyiffkmcnotreduced}\label{thm:existentially-iff-kmc-notreduced}
  (1) If $S$ is $\OBI{k}$, $\infIBI$, and $k$-exhaustive,
  then it is $\exists$-$k$-bounded.
    (2) If $S$ is $\exists$-$k$-bounded and satisfies eventual reception,
  then it is $k$-exhaustive.
  \end{restatable}

\subsection{Existentially Stable Bounded Communicating
  Automata}\label{sub:classical-exist}
The ``classical'' definition of existentially bounded communicating
automata as found in~\cite{GenestKM07} differs slightly from
Definition~\ref{def:exist-bounded}, as it relies on a different notion
of accepting runs, see~\cite[page 4]{GenestKM07}.
Assuming that all local states are accepting, we adapt their
definition as follows: a \emph{stable accepting run} is
an execution $\acts$ starting from $s_0$ which terminates in a
\emph{stable} configuration.
\begin{definition}[Existential stable boundedness~\cite{GenestKM07}]
\label{def:classical-exist-bounded}
  A system $S$ is \emph{existentially stable $k$-bounded}, written
  $\exists$S-$k$-bounded, if for each execution $\acts$ in
  $\{\acts \st \exists \stablecsconf{q} \in \RS(S) \qst s_0
  \TRANSS{\acts} \stablecsconf{q} \}$ there is $\actsb$ such that
  $s_0 \kTRANSS{\actsb}$ with $\acts \resche \actsb$.
\end{definition}
A system is existentially stable $k$-bounded if each of its
executions leading to a \emph{stable} configuration can be re-ordered
into a $k$-bounded execution (from $s_0$).

\begin{restatable}{theorem}{thmkuskeimpclassical}\label{thm:kuske-imp-classical}
    (1) If $S$ is existentially $k$-bounded, then it is existentially \emph{stable}
  $k$-bounded.
    (2) If $S$ is existentially \emph{stable} $k$-bounded and has the stable
  property, then it is existentially $k$-bounded.
\end{restatable}

We illustrate the relationship between existentially stable bounded
communicating automata and the other classes in Figure~\ref{fig:venn}.
The example below further illustrates the strictness of the
inclusions, see Table~\ref{tbl:examples} for a summary.

\begin{example}
  Consider the systems in Figure~\ref{fig:ex-unbounded}.
    $(M_\p, M_\q)$ and $(M_\p, N'_\q)$ are (trivially) existentially
  stable $1$-bounded since none of their (non-empty) executions
  terminate in a stable configuration.
    The system $(M_\p, N_\q)$ is existentially stable $2$-bounded since
  each of its executions can be re-ordered into a $2$-bounded one.
        The system in Example~\ref{ex:exist-not-exh} is (trivially)
  $\exists$S-$1$-bounded:
    none of its (non-empty) executions terminate in a stable
  configuration (the $\msg{b}$'s are never received).
\end{example}

\begin{restatable}{theorem}{thmexistentiallydlfiffkexh}\label{thm:existentially-dlf-iff-kexh}
  Let $S$ be an $\exists$(S)-$k$-bounded system with the
  stable property, then it is $k$-exhaustive.
\end{restatable}

\newcommand{\sysstructure}[1]{#1}

\begin{table}[t]
  \centering
  \caption{Properties for key examples, where
    direct.\ stands for directed,
    $\infOBI$  for $\OBI{k}$, 
    $\infDIBI$ for $\DIBI{k}$,
    \ER{} for eventual reception property, 
    \DF\ for stable property, 
    exh.\ for $k$-exhaustive,
    $\exists$(S)-b for $\exists$(S)-bounded,
    and
    syn.\ for $n$\synk\ (for some $n \in \naturals_{>0}$).
  }\label{tbl:examples}
  {\scriptsize
    \setlength{\tabcolsep}{2pt}
    \renewcommand{\arraystretch}{0.9}
    \begin{tabular}{cll|c|ccc|ccc|c|cc|c}
      \toprule
      \# & System & Ref.\ 
      & $k$  & direct.\ & $\infOBI$ & $\infDIBI$ 
      & safe & \ER{} &\DF{} & exh. &
                                     $\exists$S-b & $\exists$-b & syn.
      \\
      \midrule 
      \refstepcounter{ExampleCounter}\arabic{ExampleCounter}\label{extbl:kmc}
         &  \sysstructure{$(M_\ptp{c},M_\ptp{s},M_\ptp{l}) $}
                  &  Fig.~\ref {fig:running-example}
      & 1
      & \okmark  & \okmark  & \okmark 
      & \okmark & \okmark & \okmark & \okmark & \okmark  & \okmark  & \okmark
      \\
      \refstepcounter{ExampleCounter}\arabic{ExampleCounter}\label{extbl:dlfnotsafe}
         &  \sysstructure{$(M_\s,M_\q,M_\rr) $ }
                  &  Ex.~\ref {ex:dlf-not-safe}
      & 1
      & \okmark  & \okmark   & \okmark 
              & \komark & \okmark & \okmark & \okmark & \okmark  &
                                                                   \okmark
                                                                & \okmark
      \\ 
      \refstepcounter{ExampleCounter}\arabic{ExampleCounter}\label{extbl:notcsa}
         &  \sysstructure{$(M_\ptp{p},M_\ptp{q},M_\ptp{r}) $ }
                  &  Fig.~\ref{fig:ex-nondirected}
      & $\geq 2$
      & \komark  & \okmark  & \komark
      & \komark & \komark & \komark & \komark & \okmark  & \okmark  & \komark
      \\
      \refstepcounter{ExampleCounter}\arabic{ExampleCounter}\label{extbl:unboundedsafe}
         &  \sysstructure{$(M_\p,M_\q) $ }
                  &  Fig.~\ref{fig:ex-unbounded}
      & any
      & \okmark   & \okmark  & \okmark 
      & \okmark & \okmark & \komark & \komark & \okmark  & \komark  & \komark
            \\
      \refstepcounter{ExampleCounter}\arabic{ExampleCounter}\label{extbl:safenotstable}
         &  \sysstructure{$(M_\p,N'_\q) $ }
                  &  Fig.~\ref{fig:ex-unbounded}
      & 2
      & \okmark   & \okmark   & \okmark 
              & \okmark & \okmark & \komark & \okmark & \okmark  & \okmark  & \komark
      \\
            \refstepcounter{ExampleCounter}\arabic{ExampleCounter}\label{extbl:not-obi}
         &  \sysstructure{$(M_\ptp{p},M_\ptp{q},M_\ptp{r},M_\ptp{s}) $ }
                  &  Fig.~\ref{fig:ex-obi-main}
      & 2
             & \komark   & \okmark   & \okmark 
      & \okmark & \okmark & \komark & \okmark & \okmark  & \okmark  & \komark
      \\
      \refstepcounter{ExampleCounter}\arabic{ExampleCounter}\label{extbl:ksynk}
         &  \sysstructure{$(M_\s,M_\rr,M_\p) $ }
                  &  Ex.~\ref{ex:exist-not-exh}
      & any
      & \okmark  & \okmark   & \okmark 
              & \komark & \komark & \komark & \komark & \okmark  &
                                                                   \okmark  & \okmark
      \\
                                                \refstepcounter{ExampleCounter}\arabic{ExampleCounter}\label{extbl:nosynk}
         &  \sysstructure{$(M_\p,M_\q) $ }
                  &  Ex.~\ref {ex:small-nosynk}
      & 1
      & \okmark  & \okmark   & \okmark 
              & \okmark & \okmark & \okmark & \okmark & \okmark  & \okmark  & \komark
      \\ 
                                                      \bottomrule
    \end{tabular}   }
\end{table}

 \subsection{Synchronisable Communicating Session Automata}\label{sec:cav-sync}
In this section, we study the relationship between
synchronisability~\cite{Bouajjani2018} and $k$-exhaustivity via
existential boundedness.
Informally, communicating automata
are synchronisable if each of their executions can be scheduled in
such a way that it consists of sequences of ``exchange phases'', where
each phase consists of a bounded number of send actions, followed by a
sequence of receive actions.
The original definition of $k$\synk\ systems~\cite[Definition
1]{Bouajjani2018} is based on communicating automata with
\emph{mailbox} semantics, i.e., each automaton has one input queue.
Here, we adapt the definition so that it matches our point-to-point
semantics.
We write $\ASetSend$ for 
$\ASet \cap (\CSet \times \{!\} \times \ASigma)$,
and
$\ASetRcv$ for 
$\ASet \cap (\CSet \times \{?\} \times \ASigma)$.
\begin{definition}[Synchronisability]\label{def:synk}
  A valid execution $\acts = \acts_1 \cdots \acts_n$ is a \emph{$k$-exchange}
  if and only if:
            (1)
        $\forall 1 \leq i \leq n \qst \acts_i \in \ASetSendC \concat
    \ASetRcvC \land \lvert \acts_i \rvert \leq 2k$; and 
            
    \noindent
    (2)
            $\forall \ptp{pq} \in \CSet \qst 
    \forall 1 \leq i \leq n \qst \esndproj{\acts_i}{pq} \neq
    \ercvproj{\acts_i}{pq} \implies
    \forall i < j \leq n \qst \ercvproj{\acts_j}{pq} =\emptyw$.  
          
      We write $\kexchange{\ASetC}{k}$ for the set of executions that are
  $k$-exchanges and say that an execution $\acts$ is \emph{$k$\synk}\ if
  $\equivclass{\acts}{\resche} \cap \kexchange{\ASetC}{k} \, \neq
  \varnothing$.
      A system $S$ is  \emph{$k$\synk}\ if each execution in
  $\{\acts \st \exists s \qst s_0 \smash{\TRANSS{\acts}} s \}$ is
  $k$\synk.
\end{definition}

Condition (1) says that execution $\acts$ should be
a sequence of an arbitrary number of send-receive phases, where
each phase consists of at most $2k$ actions.
Condition (2) says that if a message is not
received in the phase in which it is sent, then it cannot be
received in $\acts$.
Observe that the bound $k$ is on the number of actions (over possibly
different channels) in a phase rather than the number of pending
messages in a given channel.

\begin{example}\label{ex:small-nosynk}
  The system below (left) is $1$\MC\ and $\exists$(S)-1-bounded, but it is
  \emph{not} $k$-synchronisable for any $k$.
    The subsequences of send-receive actions in the $\resche$-equivalent
  executions below are highlighted (right).
    \[
  \begin{array}{c|c}
    \begin{array}{lcr}
      M_\p:
      \begin{tikzpicture}[mycfsm, node distance =0.5cm and 0.5cm]
        \node[state, initial, initial where=left] (s0) {};
        \node[state, right=of s0] (s1) {};
        \node[state, right=of s1] (s2) {};
        \node[state, right=of s2] (s3) {};
        \node[state, right=of s3] (s4) {};
                \path
        (s0) edge node [above] {$\PSEND{pq}{a}$} (s1)
        (s1) edge node [above] {$\PRECEIVE{qp}{c}$} (s2)
        (s2) edge node [above] {$\PSEND{pq}{b}$} (s3)
        (s3) edge node [above] {$\PRECEIVE{qp}{d}$} (s4)    
        ;
      \end{tikzpicture}
      \\
      M_\ptp{q}:
      \begin{tikzpicture}[mycfsm, node distance =0.5cm and 0.5cm]
        \node[state, initial, initial where=left] (s0) {};
        \node[state, right=of s0] (s1) {};
        \node[state, right=of s1] (s2) {};
        \node[state, right=of s2] (s3) {};
        \node[state, right=of s3] (s4) {};
                \path
        (s0) edge node [above] {$\PSEND{qp}{c}$} (s1)
        (s1) edge node [above] {$\PSEND{qp}{d}$} (s2)
        (s2) edge node [above] {$\PRECEIVE{pq}{a}$} (s3)
        (s3) edge node [above] {$\PRECEIVE{pq}{b}$} (s4)    
        ;
      \end{tikzpicture}
    \end{array}
    \quad
 &
   \quad
   {\small
   \begin{array}{cccccc}
     \acts_1 & = & 
                   \underbracket{
                   {\color{blue} \PSEND{pq}{a}}
                   \concat
                   \PSEND{qp}{c}
                   \concat
                   \PRECEIVE{qp}{c}
                   }
                   \concat
                   \underbracket{
                   \PSEND{qp}{d}
                   \concat
                   {\color{red}  \PRECEIVE{pq}{a}}
                   }
                                      \concat
                                      \underbracket{
                   \PSEND{pq}{b}
                   \concat
                   \PRECEIVE{qp}{d}
                   \concat
                   \PRECEIVE{pq}{b}
                   }
     \\
     \acts_2 & = & 
                   \underbracket{
                   \PSEND{pq}{a}
                   \concat
                   \PSEND{qp}{c}
                   \concat
                   {\color{blue} \PSEND{qp}{d}}
                   \concat
                   \PRECEIVE{qp}{c}
                   \concat
                   \PRECEIVE{pq}{a}
                   }
                                      \concat
                                      \underbracket{
                   \PSEND{pq}{b}
                   \concat
                   {\color{red}    \PRECEIVE{qp}{d}}
                   \concat
                   \PRECEIVE{pq}{b} 
                   }
   \end{array}
                   }
  \end{array}
  \]

  Execution $\acts_1$ is $1$-bounded for $s_0$, but it is not a
  $k$-exchange since, e.g., $\msg{a}$ is received outside of the
  phase where it is sent.
    In $\acts_2$, message $\msg{d}$ is received outside of its
  sending phase.
                    In the terminology of~\cite{Bouajjani2018}, this system is not
  $k$-synchronisable because there is a ``\emph{receive-send
    dependency}'' between the exchange of message $\msg{c}$ and
  $\msg{b}$, i.e., $\ptp{p}$ must receive $\msg{c}$ before it sends
  $\msg{b}$.
    Hence, there is no $k$-exchange that is $\resche$-equivalent to
  $\acts_1$ and $\acts_2$.
\end{example}

\begin{restatable}{theorem}{thmksynkrelkmc}\label{thm:ksynk-rel-kmc}
      (1) If $S$ is $k$\synk, then it is  $\exists$-$k$-bounded.
    (2) If $S$ is $k$\synk\ and has the eventual reception property,
  then it is $k$-exhaustive.
    \end{restatable}

Figure~\ref{fig:venn} and Table~\ref{tbl:examples} summarise the
results of \S~\ref{sec:exist-bounded} wrt.\ $\OBI{k}$ and $\infIBI$
\CSA.
We note that any finite-state system is $k$-exhaustive (and
$\exists$(S)-$k$-bounded) for sufficiently large $k$, while this does
not hold for synchronisability, see Example~\ref{ex:small-nosynk}.

\begin{table}[t]
  \centering
  \caption{Experimental evaluation.
        $\lvert \PSet \rvert$ is the number of participants,
        $k$ is the bound,
        $\lvert\RTSACRO\rvert$ is the number of transitions in 
    the \emph{reduced} $\kTS{S}$ (see \appendixref{app:por}),
            direct.\ stands for directed,
                                    Time is the time taken to check all the properties shown in this table,
    and
        \GMC\ is \okmark\ if the system is generalised multiparty compatible~\cite{LTY15}.
          }\label{tab:benchmarks} 
  {\scriptsize
    \setlength{\tabcolsep}{5pt}
    \renewcommand{\arraystretch}{0.9}
    \begin{tabular}{l|c|c|c|c|c|c|c|c|c}
            \toprule
      Example  
      & $\lvert \PSet \rvert$ \ 
      & $k$
      & $\lvert\RTSACRO\rvert$ 
      & \textit{direct.}
      & $\OBI{k}$
      & $\CIBI{k}$
      & $k$\MC\
            & Time & \GMC
      \\\midrule
      Client-Server-Logger & 3 
      & 1  
      & 11 & \okmark 
      & \okmark  & \okmark   & \okmark  
      & 0.04s &\komark
      \\
      4 Player game\modifmark~\cite{LTY15} & 4 
      & 1  
      & 20 & \komark 
      & \okmark  & \okmark   & \okmark  
      & 0.05s & \okmark
      \\
      Bargain~\cite{LTY15} & 3 
      & 1 
      & 8 & \okmark
      & \okmark  & \okmark   & \okmark  
      & 0.03s
             & \okmark
      \\
      Filter collaboration~\cite{YellinS97}& 2 
      & 1                             
      &10 & \okmark 
      & \okmark  & \okmark   & \okmark  
      & 0.03s  & \okmark
      \\
      Alternating bit\modifmark~\cite{PEngineering}& 2 
      & 1  
      & 8 & \okmark
      & \okmark  & \okmark   & \okmark  
      & 0.04s & \komark
      \\
      TPMContract v2\modifmark~\cite{HalleB10}& 2 
      & 1  
      & 14 &  \okmark 
      & \okmark  & \okmark   & \okmark  
      & 0.04s  & \okmark
      \\
      Sanitary agency\modifmark~\cite{SalaunBS06}& 4 
      & 1  
      & 34 &  \okmark 
      & \okmark  & \okmark   & \okmark  
      & 0.07s  & \okmark
      \\
      Logistic\modifmark~\cite{BPMNcoreography}& 4 
      & 1  
      & 26 &  \okmark 
      & \okmark  & \okmark   & \okmark  
      & 0.05s  & \okmark
      \\
      Cloud system v4~\cite{GudemannSO12}& 4 
      & 2  
      & 16 &  \komark 
      & \okmark  & \okmark   & \okmark 
      & 0.04s  & \okmark 
                                                                              \\
      Commit protocol~\cite{Bouajjani2018} & 4 
      & 1  
      & 12 & \okmark 
      & \okmark  & \okmark   & \okmark  
      & 0.03s & \okmark
      \\
      Elevator\modifmark~\cite{Bouajjani2018} & 5 
      & 1  
      & 72 & \komark 
      & \okmark  & \komark   & \okmark  
      & 0.14s  & \komark
      \\
      Elevator-dashed\modifmark~\cite{Bouajjani2018} & 5 
      & 1  
      & 80 & \komark 
      & \okmark  & \komark   & \okmark  
      & 0.16s  & \komark
      \\
      Elevator-directed\modifmark~\cite{Bouajjani2018} & 3 
      & 1  
      & 41 & \okmark
      & \okmark  & \okmark   & \okmark  
      & 0.07s  & \okmark
      \\
      Dev system~\cite{PereraLG16}& 4 
      & 1  
      & 20 & \okmark 
      & \okmark  & \okmark   & \okmark  
      & 0.05s  & \komark
      \\
      Fibonacci~\cite{NHYA2018} & 2 
      & 1  
      & 6 & \okmark 
      & \okmark  & \okmark   & \okmark  
      & 0.03s & \okmark
      \\ 
      \textsc{Sap}-Negot.~\cite{NHYA2018,ocean} & 2 
      & 1  
      & 18 & \okmark 
      & \okmark  & \okmark   & \okmark  
      & 0.04s & \okmark
      \\
      \textsc{sh}~\cite{NHYA2018} & 3 
      & 1  
      & 30 & \okmark 
      & \okmark  & \okmark   & \okmark  
      & 0.06s & \okmark
      \\
      Travel agency~\cite{NHYA2018,scribble} & 3 
      & 1  
      & 21 & \okmark 
      & \okmark  & \okmark   & \okmark  
      & 0.05s & \okmark
      \\
      \textsc{http}~\cite{NHYA2018,HuBook17} & 2 
      & 1  
      & 48 & \okmark 
      & \okmark  & \okmark   & \okmark  
      & 0.07s & \okmark
      \\
      \textsc{smtp}~\cite{NHYA2018,HuY16} & 2 
      & 1  
      & 108 & \okmark 
      & \okmark  & \okmark   & \okmark  
      & 0.08s &  \okmark
      \\ 
      gen\_server (buggy)~\cite{TaylorTWD16} & 3
      & 1  
      & 56 & \komark 
      & \komark  & \okmark   & \komark  
      & 0.03s &  \komark 
      \\
      gen\_server (fixed)~\cite{TaylorTWD16} & 3
      & 1  
      & 45 & \komark 
      & \okmark  & \okmark   & \okmark  
      & 0.03s &  \okmark
      \\
      double buffering~\cite{MostrousESOP09} & 3
      & 2  
      & 16 & \okmark 
      & \okmark  & \okmark   & \okmark
      & 0.01s &  \komark
      \\
      \bottomrule
    \end{tabular}
  }
\end{table}

\section{Experimental Evaluation}\label{sec:implementation}
We have implemented our theory in a tool~\cite{kmc}
which takes two inputs: ($i$) a system of communicating automata and
($ii$) a bound $\textsc{max}$.
The tool iteratively checks whether the system validates the premises
of Theorem~\ref{thm:soundness}, until it succeeds or reaches
$k=\textsc{max}$.
We note that the $\OBI{k}$ and $\infIBI$ conditions are required for
our soundness result (Theorem~\ref{thm:soundness}), but are orthogonal
for checking $k$\MC.
Each condition is checked on a \emph{reduced bounded transition
  system}, called $\kRTS{S}$.
Each verification procedure for these conditions is implemented in
Haskell using a simple (depth-first-search based) reachability check
on the paths of $\kRTS{S}$.
We give an (optimal) partial order reduction algorithm to construct
$\kRTS{S}$ in \appendixref{app:por} and show that it preserves our
properties.

We have tested our tool on $20$ examples taken from the literature,
which are reported in Table~\ref{tab:benchmarks}. The table shows that
the tool terminates virtually instantaneously on all examples.
The table suggests that many systems are indeed $k$\MC\ and most can
be easily adapted to validate bound independence.
The last column refers to the \textsc{gmc} condition, a form of
\emph{synchronous} multiparty compatibility (\textsc{smc}) introduced
in~\cite{LTY15}.
The examples marked with \modifmark\ have been slightly modified to
make them \CSA\ that validate $\OBI{k}$ and $\infIBI$.
For instance, we take only one of the possible interleavings between
mixed actions to remove mixed states (taking send action before
receive action to preserve safety), see
\appendixref{app:evaluation-details}.

We have assessed the scalability of our approach with automatically
generated examples, which we report in Figure~\ref{fig:plots}.
Each system considered in these benchmarks consists of $2m$ (directed)
\CSA\ for some $m \geq 1$ such that
$S = (M_{\ptp{p_i}})_{1 \leq \ptp{i} \leq 2m}$, and each automaton
$M_{\ptp{p_i}}$  is of the form (when $i$ is \emph{odd}):
\[ 
  \begin{array}{c}
    M_{\ptp{p_i}}: 
     \begin{tikzpicture}[baseline=(current  bounding  box.center),
       font=\scriptsize, node distance = 0.5cm and 1.4cm,  
       initial distance=0.25cm,>=stealth]
             \node[state, initial, initial where=left] (s0) {};
      \node[state, right=of s0] (s1) {};
      \node[state, right =of s1,xshift=0pt] (s2) {};
      \node[state, right=of s2] (s3) {};
      \node[state, right =of s3] (s4) {};
      \node[state, right=of s4,xshift=0pt] (s5) {};
      \node[state, right=of s5] (s6) {};
            \path[->]
      (s0) edge [bend left=20] node (p1t) [above]  {$\PSEND{p_{i}p_{i+1}}{a_1}$} (s1)
      (s0) edge [bend right=20] node (p1b) [below] {$\PSEND{p_{i}p_{i+1}}{a_n}$} (s1)
            (s2) edge [bend left=20] node (p2t) [above]  {$\PSEND{p_{i}p_{i+1}}{a_1}$} (s3)
      (s2) edge [bend right=20] node (p2b) [below]  {$\PSEND{p_{i}p_{i+1}}{a_n}$} (s3)
           (s3) edge [bend left=20] node (p3t) [above]  {$\PRECEIVE{p_{i+1}p_{i}}{a_1}$} (s4)
      (s3) edge [bend right=20] node (p3b) [below] {$\PRECEIVE{p_{i+1}p_{i}}{a_n}$} (s4)
            (s5) edge [bend left=20] node (p5t) [above]  {$\PRECEIVE{p_{i+1}p_{i}}{a_1}$} (s6)
      (s5) edge [bend right=20] node (p5b) [below] {$\PRECEIVE{p_{i+1}p_{i}}{a_n}$} (s6)      
      ;
            \draw[dotted,shorten <=2pt,shorten >=1pt] (p1t) -- (p1b);
      \draw[dotted,shorten <=2pt,shorten >=1pt] (p2t) -- (p2b);
      \draw[dotted,shorten <=2pt,shorten >=1pt] (p3t) -- (p3b);
      \draw[dotted,shorten <=2pt,shorten >=1pt] (p5t) -- (p5b);
            \draw[dotted,shorten <=2pt,shorten >=1pt] (s1) -- (s2);
      \draw[dotted,shorten <=2pt,shorten >=1pt] (s4) -- (s5);
      \node[below=of s0] (a) {};
      \node[below=of s3] (b) {};
      \draw [decoration={brace,mirror,raise=-0.2cm}, decorate] (a)
      --  (b) node [pos=0.5,anchor=north,yshift=0.2cm] {$k$ times}; 
      \node[below=of s6] (c) {};
      \draw [decoration={brace,mirror,raise=-0.2cm}, decorate] (b)
      --  (c) node [pos=0.5,anchor=north,yshift=0.2cm] {$k$ times}; 
    \end{tikzpicture}
  \end{array}
\]
Each $M_{\ptp{p_i}}$ sends $k$ messages to participant
$\ptp{p_{i{+}1}}$, then receives $k$ messages from
$\ptp{p_{i+1}}$. Each message is taken from an alphabet
$\{\msg{a_1}, \ldots, \msg{a_n}\}$ ($n \geq 1$).
$M_{\ptp{p_i}}$ has the same structure when $i$ is \emph{even}, but
interacts with ${\ptp{p_{i-1}}}$ instead. Observe that any system constructed in this way is $k$\MC\ for any
$k \geq 1$, $n \geq 1$, and $m \geq 1$.
The shape of these systems allows us to assess how our approach fares
in the worst case, i.e., large number of paths in $\kRTS{S}$.
Figure~\ref{fig:plots} gives the time taken for our tool to terminate
($y$ axis) wrt.\ the number of transitions in $\kRTS{S}$ where $k$ is
the least natural number for which the system is $k$\MC.
The plot on the left in Figure~\ref{fig:plots} gives the timings when
$k$ is increasing (every increment from $k{=}2$ to $k{=}100$) with the
other parameters fixed ($n{=}1$ and $m{=}5$).
The middle plot gives the timings when $m$ is increasing (every
increment from $m{=}1$ to $m{=}26$) with $k{=}10$ and $n{=}1$.
The right-hand side plot gives the timings when $n$ is increasing
(every increment from $n{=}1$ to $n{=}10$) with $k{=}2$ and $m{=}1$.
The largest $\kRTS{S}$ on which we have tested our tool has $12222$
states and $22220$ transitions, and the verification took under 17
minutes.\footnote{All the benchmarks in this paper were run on an
  8-core Intel i7-7700 machine with 16GB RAM running a 64-bit Linux.}
Observe that partial order reduction mitigates the increasing size of
the transition system on which $k$\MC\ is checked, e.g.,
these experiments show that parameters $k$ and $m$ have only a linear
effect on the number of transitions (see horizontal distances between
data points). However the number of transitions increases
exponentially with $n$ (since the number of paths in each automaton
increases exponentially with $n$).

\begin{figure*}[t]
  \centering
    \begin{tikzpicture}
    [node distance = 0cm and -0.3cm]
            \node (plota) 
    {\includegraphics[width=0.33\textwidth]{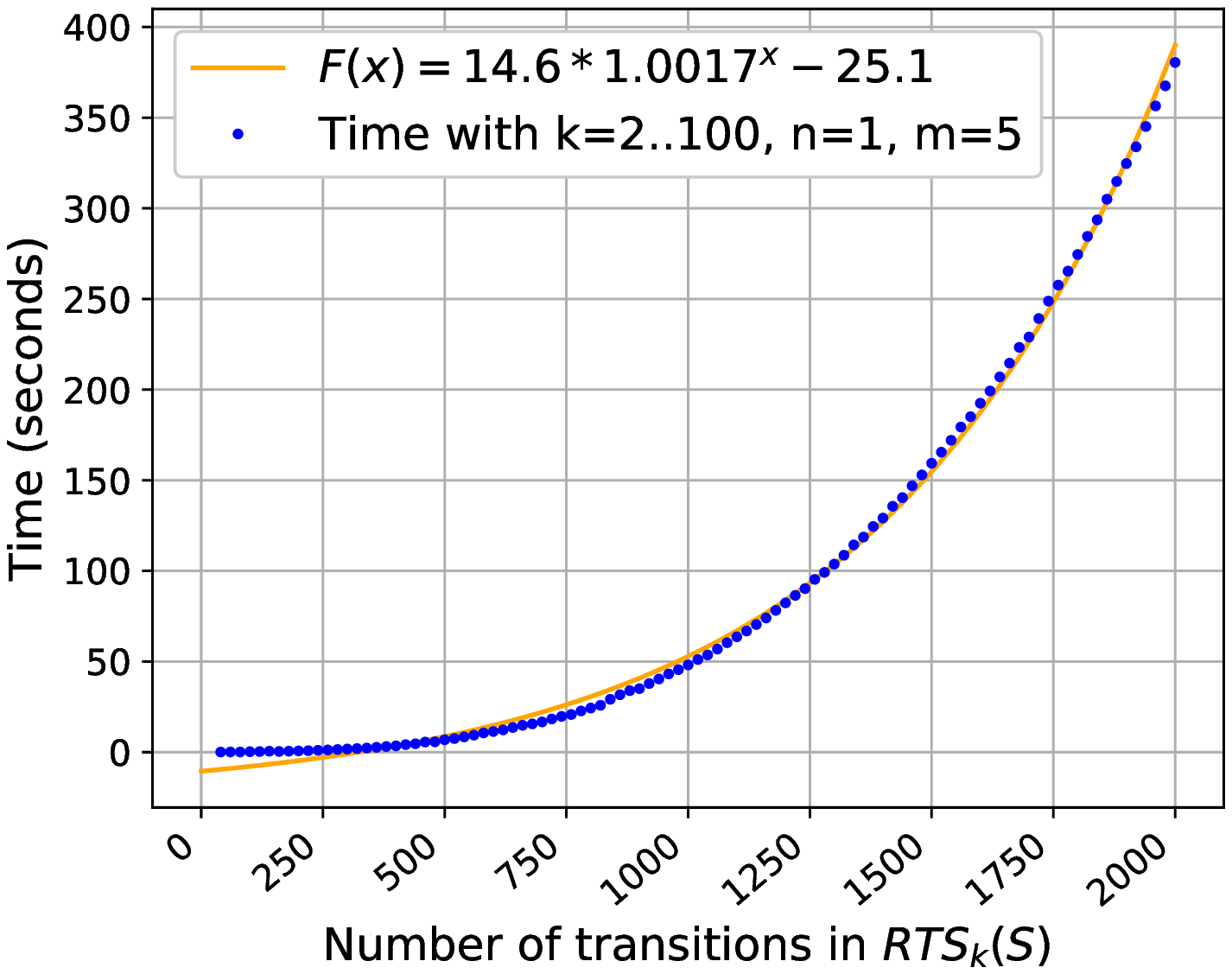}};
        \node[right=of plota] (plotb) {\includegraphics[width=0.33\textwidth]{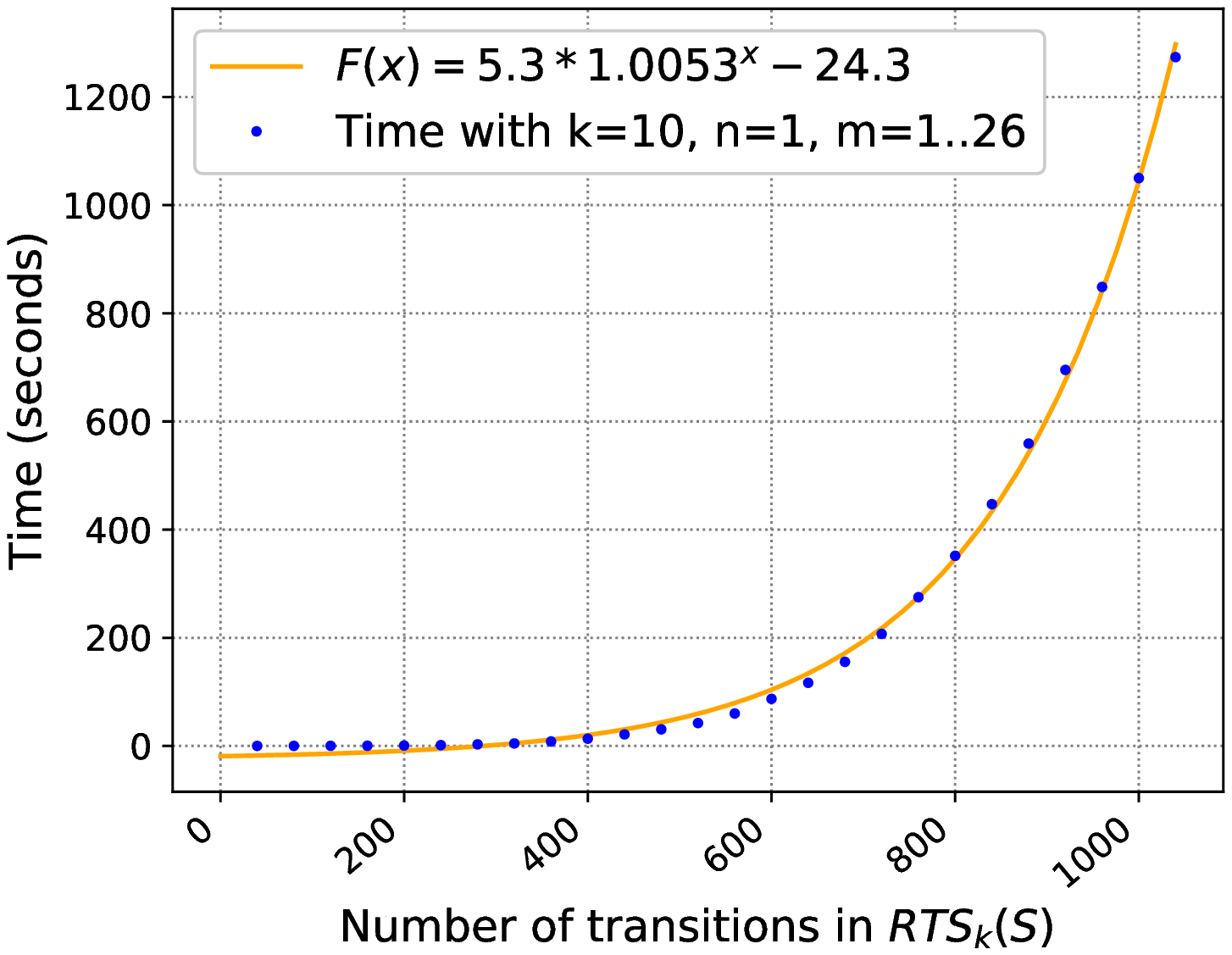}};
        \node[right=of plotb] (plotb) {\includegraphics[width=0.33\textwidth]{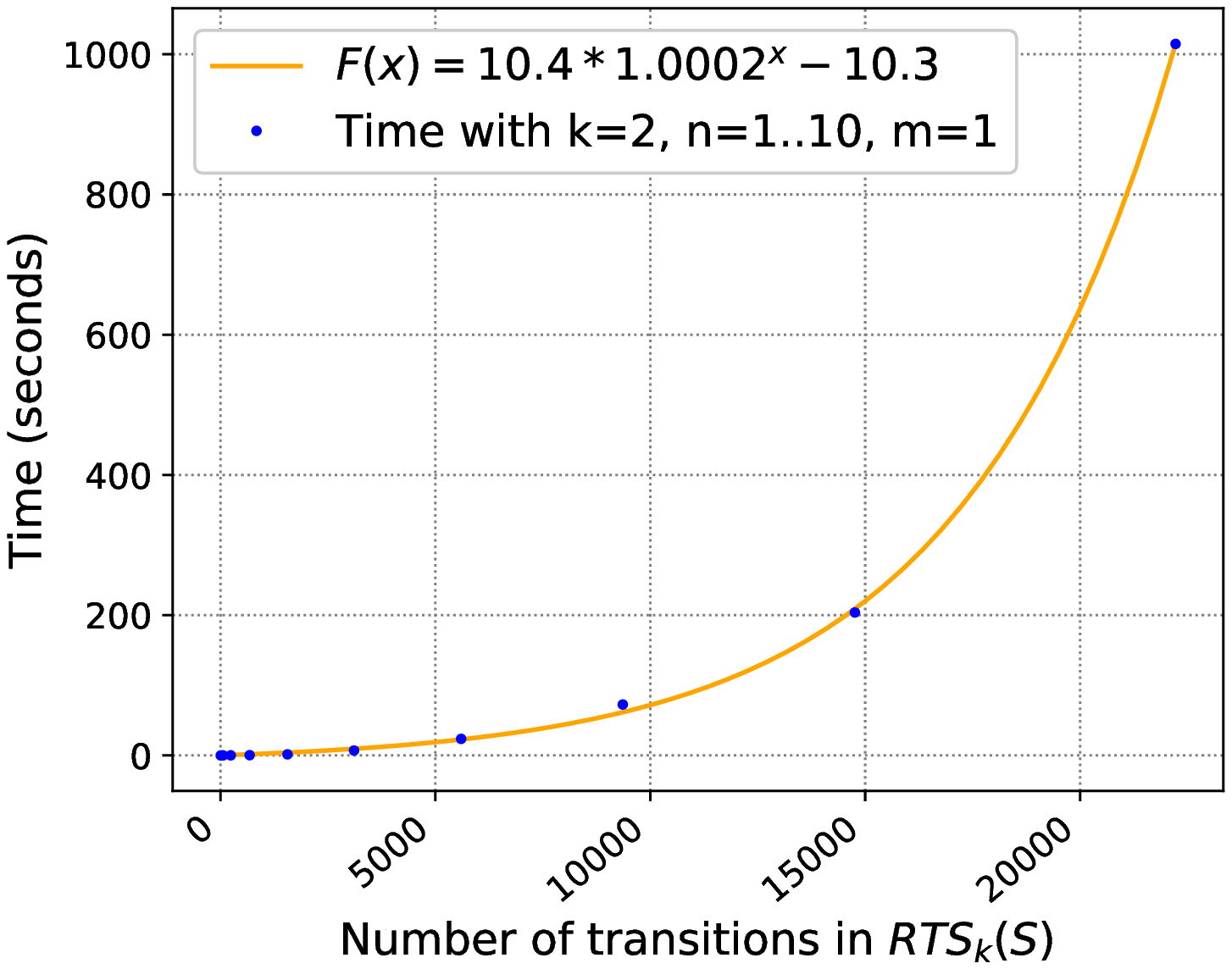}};
      \end{tikzpicture}
    \caption{Benchmarks: increasing $k$ (left), increasing $m$ (middle), and increasing $n$ (right).}
  \label{fig:plots}
\end{figure*}

 \section{Related Work}\label{sec:related}
\paragraph{Theory of communicating automata}
Communicating automata were introduced, and shown to be Turing
powerful, in the 1980s~\cite{cfsm83} and have since then been studied
extensively, namely through their connection with message sequence
charts (MSC)~\cite{Muscholl10}.
Several works achieved decidability results by using bag or lossy
channels~\cite{ClementeHS14, CeceFI96,AbdullaJ93, AbdullaBJ98} or by
restricting the topology of the network~\cite{TorreMP08,PengP92}.

Existentially bounded communicating automata stand out because they
preserve the \FIFO\ semantics of communicating automata, do not
restrict the topology of the network, and include infinite state
systems.
Given a bound $k$ and an arbitrary system of (deterministic)
communicating automata $S$, it is generally \emph{undecidable} whether
$S$ is existentially $k$-bounded. However, the question becomes
decidable (\PSPACE{}-complete) when $S$ has the stable property.
The stable property is
itself generally \emph{undecidable} (it is called deadlock-freedom
in~\cite{GenestKM07,kuske14}).
Hence this class is \emph{not} directly applicable to the verification of
message passing programs since its membership is overall undecidable.
We have shown that $\OBI{k}$, $\infIBI$, and $k$-exhaustive \CSA\
systems are (strictly) included in the class of existentially bounded
systems.
Hence, our work gives a sound \emph{practical} procedure to check
whether \CSA\ are existentially $k$-bounded.
To the best of our knowledge, the only tools dedicated to the
verification of (unbounded) communicating automata are
McScM~\cite{HeussnerGS12} and Chorgram~\cite{LTY17}.
Bouajjani et al.~\cite{Bouajjani2018} study a variation of
communicating automata with \emph{mailboxes} (one input queue per
automaton).
They introduce the class of synchronisable systems and a procedure to
check whether a system is $k$-synchronisable; it relies on executions
consisting of $k$-bounded exchange phases.
Given a system and a bound $k$, it is decidable (\PSPACE{}-complete)
whether its executions are equivalent to $k$-synchronous executions.
Section~\ref{sec:cav-sync} states that any
$k$-synchronisable system which satisfies eventual reception is also
$k$-exhaustive, see Theorem~\ref{thm:ksynk-rel-kmc}.
In contrast to existential boundedness, synchronisability does not
include all finite-state systems.
Our characterisation result, based on local bound-agnosticity
(Theorem~\ref{thm:completeness}), is \emph{unique} to
$k$-exhaustivity. It does not apply to existential boundedness nor
synchronisability, see, e.g., Example~\ref{ex:exist-not-exh}.
The term ``synchronizability'' is used by Basu et
al.~\cite{BasuBO12,BasuB16b} to refer to another verification
procedure for communicating automata with mailboxes.
Finkel and Lozes~\cite{FinkelL17} have shown that this notion of
synchronizability is undecidable.
We note that a system that is safe with a point-to-point semantics,
may not be safe with a mailbox semantics (due to independent send
actions), and vice-versa.
For instance, the system in Figure~\ref{fig:ex-nondirected} is safe
when executed with mailbox semantics.

\paragraph{Multiparty compatibility and programming
  languages}
The first definition of multiparty compatibility appeared
in~\cite[Definition 4.2]{DY13}, inspired by the work in~\cite{GoudaMY84}, to
characterise the relationship between global types and communicating
automata.
This definition was later adapted to the setting of communicating timed
automata in~\cite{BLY15}.
Lange et al.~\cite{LTY15} introduced a generalised version of
multiparty compatibility (\GMC) to support communicating automata that
feature mixed or non-directed states. 
Because our results apply to automata without mixed states, $k$\MC\ is
not a strict extension of \GMC, and \GMC\ is not a strict extension of
$k$\MC\ either, as it requires the existence of \emph{synchronous}
executions.
In future work, we plan to develop an algorithm to synthesise
representative choreographies from $k$\MC\ systems, using the
algorithm in~\cite{LTY15}.

The notion of multiparty compatibility is at the core of recent works
that apply session types techniques to programming languages.
Multiparty compatibility is used in~\cite{NgY16} to detect deadlocks
in Go programs, and in~\cite{HuY16} to study the well-formedness of
Scribble protocols~\cite{scribble} through the compatibility of their
projections.
These protocols are used to generate various endpoint APIs that
implement a Scribble specification~\cite{HuY16,HY2017,NHYA2018}, and
to produce runtime monitoring tools~\cite{NY2017,NY2017b,NBY2017}.
Taylor et al.~\cite{TaylorTWD16} use multiparty compatibility and
choreography synthesis~\cite{LTY15} to automate the analysis of the
\texttt{gen\_server} library of Erlang/OTP.
We can transparently widen the set of safe programs captured by these
tools by using $k$\MC\ instead of synchronous multiparty
compatibility (\SMC).
The $k$\MC\ condition corresponds to a much wider instance of the
\emph{abstract} safety invariant $\asphi$ for session
types defined in~\cite{ScalasY19}.
Indeed $k$\MC\ includes \SMC\ (see \appendixref{app:smc}) and all
finite-state systems (for $k$ sufficiently large).

 \section{Conclusions}\label{sec:conc}
We have studied \CSA\ via a new condition called $k$-exhaustivity.
The $k$-exhaustivity condition is ($i$) the basis for a wider notion
of multiparty compatibility, $k$\MC, which captures asynchronous
interactions and ($ii$) the first practical, empirically validated,
sufficient condition for existential $k$-boundedness.
We have shown that $k$-exhaustive systems are fully characterised by
local bound-agnosticity (each automaton behaves equivalently for any
bound greater than or equal to $k$).
This is a key requirement for asynchronous message passing programming
languages where the possibility of having infinitely many orphan
messages is undesirable, in particular for Go and Rust which provide
\emph{bounded} communication channels.

For future work, we plan to extend our theory beyond \CSA.
We believe that it is possible to support mixed states and states
which do not satisfy $\infIBI$, as long as their outgoing transitions
are independent (i.e., if they commute).
Additionally, to make $k$\MC\ checking more efficient, we will
elaborate heuristics to find optimal bounds and off-load the
verification of $k$\MC\ to an off-the-shelf model checker.

 \paragraph*{\bf Acknowledgements}
We thank Laura Bocchi and Alceste Scalas for their comments, and David
Castro and Nicolas Dilley for testing the artifact.  This work is
partially supported by EPSRC EP/K034413/1, EP/K011715/1, 
EP/L00058X/1, EP/N027833/1, and EP/N028201/1.

\bibliographystyle{abbrv} \bibliography{kmc}

\iflong

\newpage
\appendix

\tableofcontents

\newpage

\section{Partial order reduction for \CSA}\label{sec:por}\label{app:por}
In this section, we give a partial order reduction algorithm that
allow us to mitigate the exponential cost of checking $k$\MC\ (wrt.\
the bound $k$) by exploiting the commutativity of independent actions.

Next, we define function $\spartition{s}$ which partitions the
transitions enabled at $s$, grouping them by subject and arranging
them into a sorted list.
\begin{definition}[Partition]\label{def:partition}
  Let $S$, $s \in \RS_k(S)$, and $\kTS{S} = (N, s_0, \Delta)$.
    The partition of the enabled transitions at $s$ is
  $\spartition{s} \defi L_1 \cdots L_n$ such that
        \begin{enumerate}
  \item \label{en:parti-whole}
    $\{ \action \st s \kTRANSS{\action} s' \} = \bigcup_{1 \leq i \leq n} L_i$
  \item $\forall 1 \leq i \neq j \leq n \qst $
        $L_i \cap L_j = \varnothing$ and
        $\action_i \in L_i, \action_j \in L_j \implies \subj{\action_i} \neq 
    \subj{\action_j}$.
                                  \item \label{en:parti-group}
    $\forall 1 \leq i \leq n \qst \action , \action' \in L_i \implies
    \subj{\action} = \subj{\action'}$
  \item \label{en:sortlist}
    $\forall 1 \leq i < j \leq n \qst \lvert L_i \rvert \leq \lvert L_j \rvert$
  \end{enumerate}
\end{definition}
Definition~\ref{def:partition}
specifies (1) that the family of sets
$\{L_i\}_{1 \leq i \leq n}$
is a partition of the transitions enabled at $s$ and
(2) that
the function groups transitions executed by the same participant
together. 
The last condition guarantees that the list is sorted by increasing
order of cardinality, to decrease the state space generated by
Algorithm~\ref{algo:reduction}.
Definition~\ref{def:partition} is used in
Algorithm~\ref{algo:reduction} which generates the transition relation
$\hat\Delta$ of a reduced transition system (the states are implicit
from $\hat\Delta$).

\begin{definition}[Reduced transition system]\label{def:krts}
  The reduced $k$-bounded transition system of $S$ is a labelled
  transition system $\kRTS{S} = (\hat{N}, s_0, \hat\Delta)$ which is a
  sub-graph of $\kTS{S}$ such that $\hat\Delta$ is obtained from
  Algorithm~\ref{algo:reduction} and $\hat{N}$ is the smallest set
    such that $s_0 \in \hat{N}$ and
    $s \in \hat{N} \implies \exists (s_1, \action, s_2) \in \hat{\Delta}
  \qst s \in \{s_1,s_2\}$.
      We write $s \rkTRANSS{\action} s'$ iff
  $(s, \action , s') \in \hat\Delta$.
\end{definition}
Algorithm~\ref{algo:reduction} is adapted from the
\emph{persistent-set selective search} algorithm from~\cite[Chapter
4]{Godefroid96}, where instead of computing a persistent state for
each explored state, we use a partition of enabled transitions.
Each $L_i$ in $\spartition{s}$ can be seen as a persistent set since
no transition outside of $L_i$ can affect the ability of transitions
in $L_i$ to fire. Storing all enabled transitions in a list that is progressively
consumed guarantees that no transition is forever deferred, hence the
cycle proviso~\cite[Condition {C3ii}]{Peled98} is satisfied.

\newcommand{\myaction}[1]{  \mathrel{\raisebox{#1}{$\action$}}}
\begin{figure}[t]
  \centering
  \begin{minipage}[b]{.53\textwidth}
    \begin{algorithm}[H]
                  \footnotesize 
      \DontPrintSemicolon
      \SetKwComment{tcp}{// }{}
            $\svisited$ \algoAssign $\emptyset$  \tcp*{visited states}    \label{li:init-start}
      $\sacc$ \algoAssign $\emptyset$   \tcp*{transitions} 
      $\sstack$  \algoAssign $[{\pair{s_0}{[]}}]$  \tcp*{todo} 
      \label{li:init-end}
            \While{$\sstack \neq []$}{
        $\pair{s}{E}$ \algoAssign $\pop{\sstack}$ \;
        \If{$s \notin \svisited$}{ 
          $\svisited$ \algoAssign $\svisited \cup \{ s \}$ \;
          \If{$E = []$}{ \label{li:zero-begin}
            $E$ \algoAssign $\spartition{s}$\; \label{li:z-new-partition} 
          }
          
          \ForEach
          {$\action \in \head{E}$}{ \label{li:nzero-begin} $s'$
                        \algoAssign $\successor{s}{\action}$\label{li:succ-action} \;
                                    $\push{\sstack}{\pair{s'}{\tail{E}}}$  \label{li:tail} \;
            $\sacc$ \algoAssign $\sacc \cup \{(s, \action, s'\}$\label{li:nzero-end}
          } 
        }
      }
      \KwRet{$\sacc$}  \label{li:return}
            
      \vspace{0.5em}
      
      \caption{Computing $\kRTS{S}$.}\label{algo:reduction}
    \end{algorithm}
  \end{minipage}
    \vrule{}
    \begin{minipage}[c]{.45\textwidth}
        \begin{tabular}{c}
      $
      \appfun{S}{\lts} \! = \!
      \begin{cases}
        (\snddir{S} \lor \OBI{k}(S,\lts))
        \\ \land \\
                (\rcvdir{S} \lor \DIBI{k}(S,\lts)
        \\
        \qquad \qquad\quad
        \lor \CIBI{k}(S,\lts))
        \\ \land \\
                (k\textit{-exhaustive}(S,\lts))
      \end{cases}
      $
      \\
    \end{tabular}
    
    \vspace{1cm}
    
    \begin{algorithm}[H]
  \footnotesize 
    \DontPrintSemicolon  \For{$1 \leq k \leq \textsc{max}$}{
    $\lts$ \algoAssign $\kRTS{S}$ \;
        \If{$\appfun{S}{\lts}$}{
      \KwRet{$S$ is $k$-safe on $\lts$}
    }
  }
  \KwRet{\textbf{failed}}
  
  \vspace{0.5em}
  
  \caption{$k$\MC\ check.}\label{algo:kmc}
\end{algorithm}

   \end{minipage}
\end{figure}

Algorithm~\ref{algo:reduction} starts by initialising the required
data structures in Lines~\ref{li:init-start}-\ref{li:init-end}, i.e.,
the set of visited states ($\svisited$) and the set of accumulated
transitions ($\sacc$) are initialised to the empty set, while the
$\sstack$ contains only the pair ${\pair{s_0}{[]}}$ consisting of the
initial state of $\kTS{S}$ and the empty list. We overload $[]$ so
that it denotes the empty list and the empty stack.
The algorithm iterates on the content of $\sstack$ until it is empty.
Each element of the stack is a pair containing a state $s$ and a list
of sets of transitions.
For each pair $\pair{s}{E}$, if $E$ is empty, then we compute a new
partition (Line~\ref{li:z-new-partition}).
Then, we iterate over the first set of transitions in $E$ (we assume
$\head{E} = \emptyset$ when $E = []$), so to generate the successors
of $s$ according to $\head{E}$, see
Lines~\ref{li:nzero-begin}-\ref{li:nzero-end}.
In Line~\ref{li:succ-action}, we write $\successor{s}{\action}$
for the (unique) configuration $s'$ such that
$s \smash{\kTRANSS{\action}} \, s'$.
In Line~\ref{li:tail}, the tail of the list $E$ is pushed on
the stack along with the successors $s'$.
Finally, the algorithm returns a new set of transitions
(Line~\ref{li:return}).

We adapt the definitions of $\OBI{k}$ and $\DIBI{k}$ to reduced
transition systems, the definition of reduced $\CIBI{k}$ is similar
(see Definition~\ref{def:non-csa-mc-dep-reduced}).

\begin{definition}[Reduced $\OBI{k}$]\label{def:reduced-obi}
  Posing 
  $\kRTS{S} = (\hat{N}, s_0, \hat\Delta)$.
    System $S$ is \emph{reduced} $\OBI{k}$ if for all
  $s = \csconf{q}{w} \in \hat{N}$ and  $\p \in \PSet$,
    if $s \kTRANSS{\PSEND{pq}{a}}$, then
  $\forall (q_\p, \PSEND{pr}{b}, q'_\p) \in \delta_\p \qst
  s\kTRANSS{\PSEND{pr}{b}}$.
  \end{definition}

\begin{definition}[Reduced $\DIBI{k}$]\label{def:reduced-dibi}
  Posing 
  $\kRTS{S} = (\hat{N}, s_0, \hat\Delta)$.
    System $S$ is \emph{reduced} $\DIBI{k}$ if for all
  $s = \csconf{q}{w} \in \hat{N}$ and $\p \in \PSet$,
    if $s \rkTRANSS{\PRECEIVE{qp}{a}}$, then
  $ \forall (q_\p, \PRECEIVE{sp}{b}, q'_\p) \in \delta_\p \qst \s \neq
  \q \implies \!  \neg ( s \rkTRANSS{\PRECEIVE{sp}{b}} \, \lor \, s
  \rkTRANSR\rkTRANSS{\PSEND{sp}{b}})$.
    \end{definition}

The $\DIBI{k}$ and $\CIBI{k}$ properties (used to approximate
$\infIBI$) can be decided on the reduced transition system
(Theorem~\ref{thm:reduced-dibi-iff-dibi-both}).
The reduced $\OBI{k}$ property is strictly weaker than the $\OBI{k}$
property, see Example~\ref{ex:kobi-not-por}.
However, the \emph{reduced} $\OBI{k}$ property can replace $\OBI{k}$
in Theorem~\ref{thm:soundness} while preserving safety, see
Theorem~\ref{thm:reduced-soundness}.
Figure~\ref{fig:venn-iobi} gives an overview of the relationships
between the different variations of $\OBI{k}$, $\IBI{k}$, and
directedness.
The inclusions between $\infIBI$, $\CIBI{k}$, and $\DIBI{k}$ hold only
for (reduced) $\OBI{k}$ and $k$-exhaustive systems, see
Lemma~\ref{lem:kri-kexh-imp-infrip-both}.

\begin{figure}[t]
  \centering
  \begin{tabular}{c@{\quad}c}
    \begin{minipage}{0.45\textwidth}
      \centering
      \begin{tikzpicture}[scale=0.7, every node/.style={transform shape}]
        \node (cibi) at (0,0) [draw=black, minimum
        width=5.8cm,minimum height=1.6cm,venn=blue!50,rounded corners=2pt]
        {};
                \node (sibi) at (0,0) [draw=black, minimum
        width=5.5cm,minimum height=0.8cm,venn=green!50,rounded corners=2pt]
        {};
                \node (directed) at (0,0) [draw=black, minimum
        width=2cm,minimum height=0.2cm,venn=yellow!50,rounded corners=2pt]
        {Send directed};
                \node[below left,yshift=1pt] at (sibi.north east) {$\OBI{k}$};
                \node[below left,yshift=1pt] at (cibi.north east) {reduced $\OBI{k}$};
                      \end{tikzpicture}
    \end{minipage}
    &
      \begin{minipage}{0.45\textwidth}
        \centering
        \begin{tikzpicture}[scale=0.7, every node/.style={transform shape}]
          \node (ibi) at (0,0) [draw=black, minimum
          width=6.8cm,minimum height=1.9cm,venn=red!50,rounded corners=2pt]
          {};
                    \node (cibi) at (0,0) [draw=black, minimum
          width=5.8cm,minimum height=1.6cm,venn=blue!50,rounded corners=2pt]
          {};
                    \node (sibi) at (0,0) [draw=black, minimum
          width=5.5cm,minimum height=0.8cm,venn=green!50,rounded corners=2pt]
          {};
                    \node (directed) at (0,0) [draw=black, minimum
          width=2cm,minimum height=0.2cm,venn=yellow!50,rounded corners=2pt]
          {Receive directed};
                    \node[below left,yshift=2pt, align=right,text width=2cm] at (sibi.north east) {{(reduced)} $\DIBI{k}$};
                    \node[below left,yshift=2pt] at (cibi.north east) {{(reduced)} $\CIBI{k}$};
                    \node[below left,xshift=1pt] at (ibi.north east) {$\infIBI$};
        \end{tikzpicture}
      \end{minipage}
  \end{tabular}
  \caption{Overview of output and input bounded independence variations.}\label{fig:venn-iobi}
\end{figure}
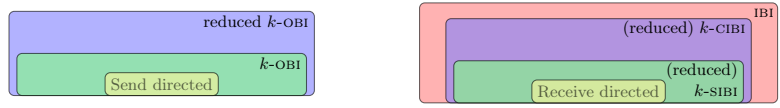
 
\begin{restatable}{theorem}{thmreduceddibiiffdibiBOTH}\label{thm:reduced-dibi-iff-dibi-both}
  Let $S$ be reduced $\OBI{k}$. $S$ is reduced $\CIBI{k}$ (resp.\ $\DIBI{k}$) iff $S$ is
  $\CIBI{k}$ (resp.\ $\DIBI{k}$).
\end{restatable}

\begin{restatable}{lemma}{lemobiimpreducedobi}\label{lem:obi-imp-reduced-obi}
  Let $S$ be a system, if $S$ is $\OBI{k}$, then $S$ is
  also {reduced} $\OBI{k}$.
\end{restatable}

\begin{restatable}{theorem}{thmreducedsoundness}\label{thm:reduced-soundness}
  If $S$ is reduced $\OBI{k}$, $\infIBI$, and $k$\MC, then it is safe.
\end{restatable}

\begin{example}\label{ex:kobi-not-por}
  The system below is reduced $\OBI{1}$, but not $\OBI{1}$.
      There is a configuration in $\TS{1}{S}$ from which $M_\p$ can fire
  $\PSEND{pr}{d}$ but not $\PSEND{pq}{b}$.
      Depending on the ordering chosen to sort the list of sets of
  transitions in $\spartition{\_}$, $\PRECEIVE{pq}{a}$ may always be
  executed before $M_\p$ reaches the violated state in $\RTS{1}{S}$,
  hence hiding the violation of $\OBI{k}$ in the reduced transition
  system.
      \[
  \begin{array}{cccc}
    \p:
    \begin{tikzpicture}[mycfsm, node distance = 0.6cm and 0.6cm]
      \node[state, initial, initial where=left] (s0) {};
      \node[state, right=of s0] (s1) {};
      \node[state, right=of s1] (s2) {};
      \node[state, above=of s2] (s3) {};
      \node[state, right=of s2] (s4) {};
            \path
      (s0) edge node [above] {$\PSEND{pq}{a}$} (s1)
      (s1) edge node [below] {$\PRECEIVE{rp}{c}$} (s2)
      (s2) edge node [left] {$\PSEND{pq}{b}$} (s3)
      (s2) edge node [above] {$\PSEND{pr}{d}$} (s4)
      ; 
    \end{tikzpicture}
    \qquad
    &
      \ptp{q}:
      \begin{tikzpicture}[mycfsm, node distance = 0.6cm and 0.6cm]
        \node[state, initial, initial where=left] (s0) {};
        \node[state, right=of s0] (s1) {};
        \node[state, right=of s1] (s2) {};
                \path
        (s0) edge node [above] {$\PRECEIVE{pq}{a}$} (s1)
        (s1) edge node [above] {$\PRECEIVE{pq}{b}$} (s2)
        ;
      \end{tikzpicture} 
      \qquad
    &
      \ptp{s}:
      \begin{tikzpicture}[mycfsm, node distance = 0.6cm and 0.6cm]
        \node[state, initial, initial where=left] (s0) {};
        \node[state, right=of s0] (s1) {};
                \path
        (s0) edge node [above] {$\PSEND{rp}{c}$} (s1)
        ;
      \end{tikzpicture}
   &
      \ptp{r}:
      \begin{tikzpicture}[mycfsm, node distance = 0.6cm and 0.6cm]
        \node[state, initial, initial where=left] (s0) {};
        \node[state, right=of s0] (s1) {};
                \path
        (s0) edge node [above] {$\PRECEIVE{pr}{d}$} (s1)
        ;
      \end{tikzpicture}
  \end{array}
  \]
    This system is $k$-exhaustive for any $k \geq 1$ and (reduced)
  $\OBI{k}$ for any $k \geq 2$.
\end{example}

Below we adapt the definitions of safety
(Definition~\ref{def:k-safety}) and $k$-exhaustivity
(Definition~\ref{def:exhaustive}) to reduced transition systems.
\begin{definition}[Reduced $k$-safety]\label{def:reduced-k-safety}
  Posing $\kRTS{S} = (\hat{N}, s_0, \hat\Delta)$.  System $S$ is
  \emph{reduced} $k$-\emph{safe} if the following conditions hold for
  all $ s = \csconf q w \in \hat{N}$,
  \begin{enumerate}
  \item \label{en:safety-er}
    $\forall\p\q \in \CSet$, if $w_{\p\q} = \msg{a} \cdot w'$,
    then $s \rkTRANSR \rkTRANSS{\PRECEIVE{pq}{a}}$.
  \item \label{en:safety-pro}
    $\forall\p \in \PSet$, if $q_\p$ is a \emph{receiving} state,
    then $s \rkTRANSR \rkTRANSS{\PRECEIVE{qp}{a}}$ for $\q \in
    \PSet$ and $\msg{a} \in \ASigma$.
  \end{enumerate}
  \end{definition}

\begin{definition}[Reduced
  $k$-exhaustivity]\label{def:reduced-exhaustive}
  Posing 
  $\kRTS{S} = (\hat{N}, s_0, \hat\Delta)$.
    System $S$ is \emph{reduced} $k$-\emph{exhaustive} if
  for all $s = \csconf{q}{w} \in \hat{N}$ and $\p \in \PSet$,
    if $q_\p$ is a sending state, then
    $
  \forall (q_\p, \action, q'_\p) \in \delta_\p \qst \exists
  \acts \in \ASetC
  \qst
  s  \rkTRANSS{\acts}\rkTRANSS{\action}  \text{ and } \p \notin
  \acts.
  $
  \end{definition}

Next, we state that checking $k$-safety (resp.\ $k$-exhaustivity) is
equivalent to checking \emph{reduced} $k$-safety (resp.\
$k$-exhaustivity), which implies that checking $k$\MC\ can be done on
$\kRTS{S}$ instead of $\kTS{S}$, the former being generally much
smaller than the latter.
We note that the reduction requires (reduced) $\OBI{k}$ and $\IBI{k}$
to hold as they imply that if a transition $(q_\p, \action , q'_\p)$ is
enabled at $s = \csconf{q}{w}$, then we have that ($i$) \emph{all}
send actions outgoing from local state $q_\p$ are enabled at $s$ (and
they will stay enabled until one is fired) or ($ii$) \emph{exactly
  one} receive action is enabled from $q_\p$ (and it will stay enabled
until it is fired).

\begin{restatable}{theorem}{themsafeiffredsafekexhau}\label{them:safe-iff-red-safe-kexhau}
  Let $S$ be reduced $\OBI{k}$ and reduced $\IBI{k}$.
    (1) $S$ is \emph{reduced} $k$-safe iff $S$ is $k$-safe.
    (2) $S$ is \emph{reduced} $k$-exhaustive iff
  $S$ is $k$-exhaustive.
\end{restatable}

Algorithm~\ref{algo:kmc} checks whether a system $S$ is $k$\MC\, for
some $k \leq \textsc{max}$, where $\textsc{max}$ is a user-provided
constant.
At each iteration, it constructs the $\kRTS{S}$ of the input system
$S$.
If $k$ is a sufficient bound to make a sound decision (function
$\appfun{S}{\lts}$), then it tests for $k$-safety,
otherwise proceeds to the next iteration with $k{+}1$.
Function $\appfun{S}{\lts}$ checks whether the premises of
Theorem~\ref{thm:reduced-soundness} hold, i.e.,
if $S$ is not send directed, written $\snddir{S}$, then it checks for
$\OBI{k}$;
$S$ is not receive directed, written $\rcvdir{S}$, then it checks for
$\DIBI{S}$ or $\CIBI{k}$; then checks whether $k$-exhaustivity holds
(all conditions are checked on $\RTS{k}{S}$).

The equivalence relation defined below relates executions which only
differ by re-ordering of independent actions, it is used in several
results below.
\begin{definition}[Projected equivalence]
    Let $\acts, \actsb \in \ASetC$, we define:
  $\eqpeer{\acts}{\actsb}$ \emph{if}
    $\forall \p \in \PSet \qst
  \onpeer{p}{\acts} = \onpeer{p}{\actsb}$.
\end{definition}

Finally, we state the optimality of Algorithm~\ref{algo:reduction}: it
never explores two executions which are $\eqpeerop$-equivalent more
than once. 
Our notion of optimality is slightly different from that
of~\citeout{AbdullaAJS14} since Algorithm~\ref{algo:reduction} does not
use sleep sets.

\begin{restatable}{lemma}{lemporoptimal}\label{lem:por-optimal}
  Let $S$ be a system such that
  $\kRTS{S} = (\hat{N}, s_0, \hat\Delta)$,
  for all $\acts$ and $\acts'$ such that $s_0 \rkTRANSS{\acts}$ and
  $s_0 \rkTRANSS{\acts'}$, we have that:
  $\acts \eqpeerop \acts' \implies \acts = \acts'$.
  \end{restatable}

 \section{Overview of the proofs of Theorems~\ref{thm:soundness}
  and~\ref{thm:reduced-soundness}}

The properties $\OBI{k}$ and $\infIBI$, and $k$-exhaustivity
\emph{together} guarantee that any choice made by an automaton is not
constrained nor influenced by the channel bounds.
The proof that $k$\MC\ guarantees safety for such systems crucially
relies on this.
The independence of choice wrt.\ the channel bounds for these \CSA\
allows us to construct sets of executions that include all possible
individual choices.
We characterise this form of closure with the definition below,
which is crucial for the further developments of this section.
\begin{definition}[$k$-Closed]\label{def:kclosed}
  Given a system $S$,
  $\paset \subseteq \ASetC$, and $s \in \RS_{k}(S)$,
    we say that \emph{$\kclosed{\paset}{s}$}, if the following two
  conditions hold:
    \begin{enumerate}
  \item $\forall \acts \in \paset \qst \exists s' \in \RS_k(S) \qst s
    \bTRANSS{\acts}{k} s'$
      \item
    $\forall \acts_0 \concat \PSEND{pq}{a} \concat \acts_1 \in \paset$
    s.t.\ $s \TRANSS{\acts_0} \csconf q w$ and
    $\forall (q_\p, \action, q'_\p) \in \delta_\p$ there is
    $\acts_0 \concat \acts_2 \concat \action \concat \acts_3 \in
    \paset$
    with $\acts_2 \concat \acts_3 \in \ASetC$ and $\p \notin \acts_2$.
  \end{enumerate}
\end{definition}
In other words, $\kclosed{\paset}{s}$ if (1) all executions in
$\paset$, starting from $s$, lead to a configuration in $\RS_k(S)$ and
(2) whenever an automaton $\p$ fires a send action in an execution in
$\paset$, then all possible choices that $\p$ can make are also
represented in $\paset$.

\begin{example}
  Consider the $1$\MC\ system $(M_\p, M_\q)$ below.
  \begin{center}
    \begin{tabular}{lcr}
      $\p:$
      \begin{tikzpicture}[mycfsm]
        \node[state, initial, initial where=above,font=\tiny] (s0) {0};
        \node[state, right=of s0,font=\tiny] (s1) {1};
        \node[state, right=of s1,font=\tiny] (s2) {2};
        \node[state, right=of s2,font=\tiny] (s3) {3};
                \path
        (s0) edge node [above] {$\PSEND{pq}{a}$} (s1)
        (s1) edge node [below] {$\PSEND{pq}{b}$} (s2)
        (s2) edge [bend right] node [below] {$\PRECEIVE{qp}{c}$} (s3)
        (s2) edge [bend left] node [above] {$\PRECEIVE{qp}{d}$} (s3)
        ;
      \end{tikzpicture}
      \qquad\qquad\qquad
      &
        $\ptp{q}:$
        \begin{tikzpicture}[mycfsm]
          \node[state, initial, initial where=above,font=\tiny] (s0) {0};
          \node[state, right=of s0,font=\tiny] (s1) {1};
          \node[state, right=of s1,font=\tiny] (s2) {2};
          \node[state, right=of s2,font=\tiny] (s3) {3};
                    \path
          (s0) edge [bend right] node [below] {$\PSEND{qp}{c}$} (s1)
          (s0) edge [bend left] node [above] {$\PSEND{qp}{d}$} (s1)
          (s1) edge node [above] {$\PRECEIVE{pq}{a}$} (s2)
          (s2) edge node [below] {$\PRECEIVE{pq}{b}$} (s3)
          ;
        \end{tikzpicture}
    \end{tabular}
  \end{center}
  The sets $\{ \emptyw \}$ and
  $ \{ \PSEND{qp}{c} , \PSEND{qp}{d}, \emptyw \}$ are both $1$-closed for
  $s_0 = (0,0;\emptyw,\emptyw)$.
    Instead, the set $\{ \PSEND{qp}{c} , \emptyw \}$ is not $1$-closed for $s_0$ since
  there is a branching in participant $\q$ that is not represented.
  \end{example}

Lemma~\ref{lem:k-mc-kclosed} follows from the facts that ($i$) $S$ is
(reduced) $\OBI{k}$ and ($ii$) $S$ is $k$-exhaustive, i.e., all send
actions are eventually enabled within the $k$-bounded executions.

\begin{restatable}{lemma}{lemkmckclosed}\label{lem:k-mc-kclosed}
    Let $S$ be \emph{reduced} $\OBI{k}$ and $k$-exhaustive. For all
  $s \in \RS_k(S)$, if $s \TRANSS{\PSEND{pq}{a}}$ and
  $\paset = \{ \acts \st s \kTRANSS{\acts}\kTRANSS{\PSEND{pq}{a}}
  \land \p \notin \acts \}$,
  then $\kclosed{\paset \neq \emptyset}{s}$.
  \end{restatable}

Note that if $\PSEND{pq}{a}$ is the only action enabled at $s$, then
$\paset = \{ \emptyw \}$. In general, we do not have $\emptyw \in
\paset$, as shown in the example below.
\begin{example}
  Consider the $1$\MC\ system $(M_\p, M_\q)$ below.
  \begin{center}
    \begin{tabular}{lcr}
      $\p:$
      \begin{tikzpicture}[mycfsm]
        \node[state, initial, initial where=above,font=\tiny] (s0) {0};
        \node[state, right=of s0,font=\tiny] (s1) {1};
        \node[state, right=of s1,font=\tiny] (s2) {2};
                \path
        (s0) edge node [above] {$\PSEND{pq}{a}$} (s1)
        (s1) edge node [above] {$\PSEND{pq}{b}$} (s2)
        ;
      \end{tikzpicture}
      \qquad\qquad\qquad
      &
        $\ptp{q}:$
        \begin{tikzpicture}[mycfsm]
          \node[state, initial, initial where=above,font=\tiny] (s0) {0};
          \node[state, right=of s0,font=\tiny] (s1) {1};
          \node[state, right=of s1,font=\tiny] (s2) {2};
                    \path
           (s0) edge node [above] {$\PRECEIVE{pq}{a}$} (s1)
           (s1) edge node [above] {$\PRECEIVE{pq}{b}$} (s2)
          ;
        \end{tikzpicture}
    \end{tabular}
  \end{center}
    Pose $s = (1,0;\msg{a}, \emptyw)$, we have that the set 
  $\{ \acts \st s \bTRANSS{\acts}{1}\bTRANSS{\PSEND{pq}{b}}{1} \land
  \p \notin \acts \} = \{ \PRECEIVE{pq}{a} \}$ is $1$-closed for $s$.
    Indeed, for the action $\PSEND{pq}{b}$ to be fired in a $1$-bounded
  execution, message $\msg{a}$ must be consumed first.
    \end{example}

Lemma~\ref{lem:new-kclosed-set} below states that if there is a
$k$-closed set of executions for a configuration $s$, we can
construct another $k$-closed set for any successor of $s$.
\begin{restatable}{lemma}{lemnewkclosedset}\label{lem:new-kclosed-set}
  Let $S$ be a $\IBI{k}$ system, $s, s'\in \RS_{k}(S)$
  and $\paset \subseteq \ASetC$ such that $\kclosed{\paset}{s}$,
    $s \bTRANSS{\action}{k} s'$, and
    $\hat\paset = \hat\paset_1 \cup \hat\paset_2$, where
  \[
  \hat\paset_1 = 
  \left\{
    \acts \st \!
    \acts \in \paset 
    \land 
    \subj{\action}\notin \acts
  \right\}
    \, \text{and} \;
  \hat\paset_2 = 
  \left\{
    \acts_1 \concat \acts_2 \st \!
    \acts_1 \concat \action \concat \acts_2 \in \paset 
    \land
    \subj{\action} \notin \acts_1
  \right\}
  \]
    Then the following holds:
    \begin{enumerate}
  \item \label{it:new-kclosed-set-part-a}
    The set 
$\kclosed{\hat\paset}{s'}$
    
  \item For all $\actsb \in \hat\paset$, there is $\acts \in \paset$
    such that either:
    \begin{itemize}
    \item \label{it:new-kclosed-set-part-b} 
            $\actsb \in \hat\paset_1$,
            $\actsb = \acts$, $\subj{\action} \notin \actsb$, and there
      are $t, t' \in \RS_k(S)$ such that $s \kTRANSS{\actsb} t$,
      $s' \kTRANSS{\actsb} t'$, and $t \kTRANSS{\action} t'$,
      and 
            $\eqpeer{\acts \concat \action}{\action \concat
        \actsb}$;
            or
          \item 
            $\actsb \in \hat\paset_2$,
            there is $t \in \RS_k(S)$ such that $s \kTRANSS{\acts} t$,
      $s' \kTRANSS{\actsb} t$,
      and 
            $\eqpeer{\acts}{\action \concat
        \actsb}$.
    \end{itemize}
  \item \label{it:new-kclosed-set-part-c} 
    $\paset  = \emptyset \iff \hat\paset = \emptyset$.
  \end{enumerate}
  \end{restatable}
Figure~\ref{fig:kclosed-lems} (left and middle) illustrates the
construction of the executions in $\hat\paset$.
The crucial part of the proof is to show that $\hat\paset$ is indeed
$k$-closed, this is done by case analysis on the structure of an arbitrary 
execution in $\hat\paset$.
The assumption that $S$ is a $\IBI{k}$ system is key here: we can rely
on the fact that if $\action$ is a receive action, then it is the
unique receive action that $\subj{\action}$ can execute from $s$.

Next, Lemma~\ref{lem:closed-set-paths} states that given the existence
of a $k$-closed set of executions, one can find an alternative but
equivalent path to a common configuration.
We show the result below by induction on $n$, using
Lemma~\ref{lem:new-kclosed-set}.
\begin{restatable}{lemma}{lemclosedsetpaths}\label{lem:closed-set-paths}
  Let $S$ be a reduced $\OBI{k}$ and $\IBI{k}$ system, then
    for all $s_1, \ldots , s_n \in \RS_k(S)$, such that $s_1
  \kTRANSS{\action_1} s_2 \cdots s_{n-1} \kTRANSS{\action_{n-1}} s_n$ (with $n > 1$).
    If there is $\emptyset \neq \paset \subseteq \ASetC$ such that
  $\kclosed{\paset}{s_1}$, then 
    there is $\acts_1 \in \paset$ and $\actsb,\acts_n \in \ASetC$ such
  that $s_1 \kTRANSS{\acts_1} t_1 \kTRANSS{\actsb} t_n$ and
  $s_n \kTRANSS{\acts_n} t_n$, for some $t_1, t_n \in \RS_k(S)$ with
    $\lvert \actsb \rvert < n$ and
    $\eqpeer{\acts_1 \concat \actsb}{\action_1 \cdots \action_n \concat \acts_n}$.
  \end{restatable}
Figure~\ref{fig:kclosed-lems} (right) illustrates
Lemma~\ref{lem:closed-set-paths}.
A key consequence of Lemma~\ref{lem:k-mc-kclosed} and
Lemma~\ref{lem:closed-set-paths} is that if $s_1 \in \RS_k(S)$, then
we have $s_1 \kTRANSS{\acts_1} t_1 \kTRANSS{\action_1}$, i.e., $t_1
\in \RS_k(S)$;
we use this result to show Lemma~\ref{lem:exist-path-eqpeer}.

\begin{restatable}{lemma}{lemexistpatheqpeer}\label{lem:exist-path-eqpeer}
  Let $S$ be reduced $\OBI{k}$, $\IBI{k{+}1}$, and $k$-exhaustive,
  then for all $s \in RS_k(S)$ and $s' \in \RS_{k{+}1}(S)$ such that
  $s \bTRANSS{\acts}{k+1} s'$, there is $t \in \RS_k(S)$ and
  $\actsb, \, \actsb' \in \ASetC$, such that $s \kTRANSS{\actsb} t$,
  $s' \bTRANSS{\actsb'}{k+1} t$, and
  $\eqpeer{\actsb}{\acts \concat \actsb'}$.
  \end{restatable}

Lemma~\ref{lem:exist-path-eqpeer} states that if $S$ is (reduced)
$\OBI{k}$, $\IBI{k{+}1}$, and $k$-exhaustive then there is a path from
any $k{+}1$-reachable configuration to a $k$-reachable configuration.
The proof is by induction on the length of $\acts$ using
Lemma~\ref{lem:k-mc-kclosed} as a starting assumption, then applying
Lemma~\ref{lem:closed-set-paths} repeatedly.

\begin{remark}
  The assumption that $S$ is $\IBI{k{+}1}$ is required, see
  Figure~\ref{fig:ex-nondirected} for an example that is $\OBI{1}$,
  $\IBI{1}$, and $1$-exhaustive but for which the conclusions of
  Lemma~\ref{lem:exist-path-eqpeer} do not hold.
\end{remark}

\begin{figure}[t]
  \centering
  \begin{tabular}{c|c|c}
    \begin{tikzpicture}[baseline=(current bounding box.center)]
      \node (s1) {$s$};
      \node[right=of s1] (s2) {$s'$};
      \node[below=of s1] (s3) {$t$};
      \node at (s2|-s3) (s4) {$t'$};
            \path[->] (s1) edge [above] node {$\action$} (s2);
      \path[->] (s3) edge [above] node {$\action$} (s4);
            \path[->] (s1) edge [left] node {$\actsb = \acts$} (s3);
      \path[->] (s2) edge [right] node {$\actsb = \acts$} (s4);
    \end{tikzpicture} 
    &
      \begin{tikzpicture}[baseline=(current bounding box.center),node distance = 0.4 and 1]
        \node (s1) {$s$};
        \node[right=of s1] (s2) {$s'$};
        \node[below=of s1] (s3) {$.$};
        \node[below=of s3] (s5) {$.$};
        \node[below=of s5] (s7) {$t$};
                        \node[right=of s7,xshift=-0.7cm,yshift=0.2cm] (txtt) {\begin{tabular}{l@{$\, = \,$}l}
                                                                $\actsb$ & $\acts_1 \concat \acts_2$
                                                                \\
                                                                $\acts$ & $\acts_1 \concat \action \concat \acts_2$
                                                              \end{tabular}
                                                            };
                                                                                                                        \path[->] (s1) edge [below] node {$\action$} (s2);
                                                                                                                        \path[->] (s1) edge [left] node {$\acts_1$} (s3);
                                                            \path[->] (s2) edge [right, bend left] node {$\acts_1$} (s5);
                                                            \path[->] (s3) edge [left] node {$\action$} (s5);
                                                            \path[->] (s5) edge [left] node {$\acts_2$} (s7);
                                                          \end{tikzpicture}
    &
            \begin{tikzpicture}[baseline=(current bounding box.center),node distance = 1 and 0.5]
        \node (s1) {$s_1$};
        \node[right= of s1] (s2) {$s_2$};
        \node[right= of s2] (smn) {$s_{n-1}$};
        \node[right= of smn] (sn) {$s_n$};
                \node[below=of s1] (t1) {$t_1$};
        \node at (sn|-t1) (tn) {$t_n$};
                \path[->] (s1) edge [above] node {$\action_1$} (s2);
        \path[dotted] (s2) edge [above] node {} (smn);
        \path[->] (smn) edge [above] node {$\action_{n-1}$} (sn);
                \path[->] (s1) edge [left] node {$\acts_1$} (t1);
        \path[->] (sn) edge [left] node {$\acts_n$} (tn);
        \path[->] (t1) edge [above] node {$\actsb$} (tn);
      \end{tikzpicture}
  \end{tabular}
  \caption{Illustrations for Lemma~\ref{lem:new-kclosed-set} and Lemma~\ref{lem:closed-set-paths}.}
  \label{fig:kclosed-lems}
\end{figure}
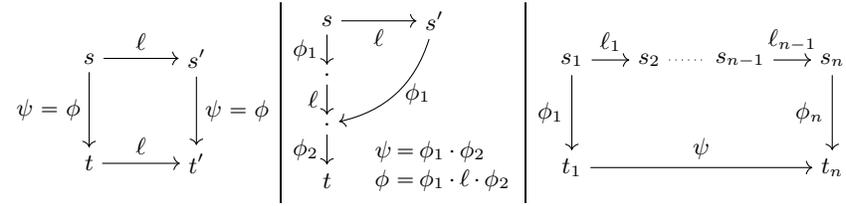

Since the $\infIBI$ property is undecidable in general, we have
introduced the $\CIBI{k}$ and $\DIBI{k}$ properties as sound
approximations of $\infIBI$, for $\OBI{k}$ and $k$-exhaustive systems.
We give a brief overview of the proof of
Lemma~\ref{lem:kri-kexh-imp-krip-reduced} (part of which implies
Lemma~\ref{lem:kri-kexh-imp-infrip-both}). The proof that $\CIBI{k}$
implies $\infIBI$ is similar, see
Lemma~\ref{lem:kri-kexh-imp-cibi-reduced} for the key result.

\begin{restatable}{lemma}{lemkrikexhimpkripreduced}\label{lem:kri-kexh-imp-krip-reduced}
  If $S$ is reduced $\OBI{k}$, $\DIBI{k}$, and $k$-exhaustive, then it is $\DIBI{k{+}1}$.
\end{restatable}

To show Lemma~\ref{lem:kri-kexh-imp-krip-reduced}, we show that for
any system that is reduced $\OBI{k}$, $\DIBI{k}$, and $k$-exhaustive,
the $\IBI{k{+}1}$ property holds, i.e.,
Lemma~\ref{lem:kri-kexh-imp-kba}.
The proof of Lemma~\ref{lem:kri-kexh-imp-kba} is by induction on the
length of an execution from $s_0$.
Then we show the final result by contradiction, using
Lemma~\ref{lem:exist-path-eqpeer} to find an execution that leads to a
$k$-reachable configuration.

 \section{Experimental evaluation: modified examples}\label{app:evaluation-details}
The examples marked with \modifmark\ have been slightly modified to
make them \CSA\ that validate $\OBI{k}$ and $\infIBI$.
To remove mixed states, we take only one of the possible interleavings
between mixed actions (we take the send action before receive action
to preserve safety).
The 4 Player game from~\cite{LTY15} has been modified so that
interleavings of mixed actions are removed (it is the only example of
Table~\ref{tab:benchmarks} that is $\CIBI{k}$ but not $\DIBI{k}$).
The Logistic example from~\cite[Figure 11.4]{BPMNcoreography} has been
modified so that the Supplier interacts sequentially (instead
of concurrently) with the Shipper then the Consignee.
We have added two dummy automata to the Elevator example
from~\cite{Bouajjani2018} which send (resp.\ receive) messages to
(resp.\ from) the Door so that a mixed state can be removed.
The Elevator-dashed example is a variant of the Elevator which is not
synchronisable.
These examples are not $\IBI{k}$ (for any $k$) because the Elevator
automaton can reach a state where it can consume messages sent by
different participants (messages \texttt{doorClosed} and
\texttt{openDoor}).
This situation cannot occur with a mailbox semantics, as
in~\cite{Bouajjani2018}, since each automaton has only one input
queue.
The Elevator-directed example is another variation where all the
automata are directed.

 \section{Extended related work}\label{sec:extended-related}
\paragraph{Theory of communicating automata}
Communicating automata were introduced in the 1980s~\citeout{cfsm83} and
have since then been studied extensively, namely through their
connection with message sequence charts (MSC)~\citeout{Muscholl10}. We
focus on closely related works.
Several works achieved decidability results by restricting the model.
For instance, some of these works substitute reliable and ordered
channels with bag or lossy channels~\citeout{ClementeHS14,
  CeceFI96,AbdullaJ93, AbdullaBJ98}.
La Torre et al.~\citeout{TorreMP08} restrict the topology of the network
so that each automaton can only consume messages from one queue (but
can send messages to all other queues).
Peng and Purushothaman~\citeout{PengP92} show that reachability, deadlock
detection, and un-boundedness detection are decidable for the class of
systems where each pair of automata can only exchange one type of
message and the topology of the network is a simple cycle.
DeYoung and Pfenning~\citeout{DeYoungP16} investigate a relationship
between proofs in a fragment of linear logic and communicating
automata that interact via a pipeline topology.

Out of these several variations, existentially bounded communicating
automata stand out because they preserve the \FIFO\ semantics of
communicating automata, do not restrict the topology of the network,
and include systems with an infinite state-space.
Existential bounds for MSCs first appeared in~\citeout{LohreyM04} and
were later applied to the study of communicating automata through MSCs
and monadic second order logic in~\citeout{GenestKM06,GenestKM07}.
Given a bound $k$ and an arbitrary system of (deterministic)
communicating automata $S$, it is generally \emph{undecidable} whether
$S$ is existentially $k$-bounded. However, the question becomes
decidable when $S$ has the stable property (a property called
deadlock-freedom in~\citeout{GenestKM07,kuske14}), the problem is
\PSPACE{}-complete.
The stable property is generally a desirable characteristic, but it is
generally \emph{undecidable}. Hence the bounded class is \emph{not}
directly applicable to verifying properties of message passing
programs since its membership is undecidable overall.
We have shown that ($i$) $\OBI{k}$, $\infIBI$, and $k$-exhaustive \CSA\
systems are (strictly) included in the class of existentially bounded
systems, ($ii$)
systems that are existentially bounded (in the sense
of~\citeout{kuske14}) and have the eventual reception property are
$k$-exhaustive; and ($iii$) systems that are existentially stable
bounded~\citeout{GenestKM07} and have the stable property are
$k$-exhaustive.
Hence, our work gives a sound \emph{practical} procedure to check
whether \CSA\ are existentially bounded.
Inspired by the work in~\citeout{GenestKM07}, Darondeau et
al.~\citeout{DarondeauGTY10} give decidability results for
``data-branching'' task systems, which are communicating
automata with internal transitions whose only branching states are
those where an \emph{internal} choice takes place.
The relationship between communicating automata and monadic second
order logic was further studied in~\citeout{Bollig14,BolligL06}.
To the best of our knowledge, the only tools dedicated to the
verification of (unbounded) communicating automata are
McScM~\citeout{HeussnerGS12} and Chorgram~\citeout{LTY17}.
Bouajjani et al.~\citeout{Bouajjani2018} study a variation of
communicating automata with \emph{mailboxes} (one input queue per
automaton).
They introduce the class of synchronisable systems and a procedure to
check whether a system is $k$-synchronisable; it relies on executions
consisting of $k$-bounded exchange phases.
Given a system and a bound $k$, it is decidable (\PSPACE{}-complete)
whether its executions are equivalent to $k$-synchronous executions.
In Section~\ref{sec:cav-sync}, we have shown that any
$k$-synchronisable system which satisfies eventual reception is also
$k$-exhaustive, see Theorem~\ref{thm:ksynk-rel-kmc}.
Our characterisation result, based on local bound-agnosticity
(Theorem~\ref{thm:completeness}), is \emph{unique} to
$k$-exhaustivity. It does not apply to existentially boundedness nor
synchronisability, see, e.g., Example~\ref{ex:exist-not-exh}.
The term ``synchronizability'' has been used by Basu et
al.~\citeout{BasuBO12,BasuB16b} to refer to another procedure for
checking properties of communicating automata with mailboxes.
Their notion of synchronizability requires that, for a given system,
its synchronous executions are equivalent to its asynchronous
executions when considering send actions only.
Finkel and Lozes~\citeout{FinkelL17} have later shown that this notion of
synchronizability is in fact undecidable.

In future work, we would like to study whether our results can be
adapted to automata which communicate via mailboxes.
We note that a system that is safe with a point-to-point semantics,
may not be safe with a mailbox semantics, and vice-versa.
For instance, the system in Figure~\ref{fig:ex-nondirected} is safe
when executed with mailbox semantics.
However, the system below is safe in the point-to-point semantics, but
\emph{unsafe} with mailbox semantics due to the fact that $\rr$ may
receive $\msg{b}$ before $\msg{a}$.
To the best of our knowledge, precise relationships and translations
between mailbox and point-to-point semantics have yet to be studied.
\begin{center}
    \begin{tabular}{c@{\qquad}c@{\qquad}c}
    $\ptp{p}:$
    \begin{tikzpicture}[mycfsm, node distance = 0.6cm and 1cm]
      \node[state, initial, initial where=above] (s0) {};
      \node[state, right=of s0] (s1) {};
            \path
      (s0) edge node [] {$\PSEND{pr}{a}$} (s1)
      ;
    \end{tikzpicture}
    &
      $\ptp{r}:$
      \begin{tikzpicture}[mycfsm, node distance = 0.6cm and 1cm]
      \node[state, initial, initial where=above] (s0) {};
      \node[state, right=of s0] (s1) {};
      \node[state, right=of s1] (s2) {};
            
            \path
      (s0) edge node [] {$\PRECEIVE{pr}{a}$} (s1)
      (s1) edge node [] {$\PRECEIVE{qr}{b}$} (s2)
      ;
    \end{tikzpicture}
    &
      $\ptp{r}:$
    \begin{tikzpicture}[mycfsm, node distance = 0.6cm and 1cm]
      \node[state, initial, initial where=above] (s0) {};
      \node[state, right=of s0] (s1) {};
            \path
      (s0) edge node [] {$\PSEND{qr}{b}$} (s1)
      ;
    \end{tikzpicture}
  \end{tabular}
\end{center}

\paragraph{Multiparty compatibility } The first definition of multiparty compatibility appeared
in~\citeout[Definition 4.2]{DY13}, inspired by the work in~\citeout{GoudaMY84}, to
characterise the relationship between global types and communicating
automata.
This definition was later adapted to the setting of communicating timed
automata in~\citeout{BLY15}.
Lange et al.~\citeout{LTY15} introduced a generalised version of
multiparty compatibility (\GMC) to support communicating automata that
feature mixed or non-directed states. 
Because our results apply to automata without mixed states, $k$\MC\ is
not a strict extension of \GMC, and \GMC\ is not a strict extension of
$k$\MC\ either, as it requires the existence of \emph{synchronous}
executions.
We discuss how our results may be extended to support communicating
automata with mixed states in Section~\ref{sec:conc}.
In future work, we will develop an algorithm to synthesise
representative choreographies from $k$\MC\ systems, using the
algorithm in~\citeout{LTY15}.

\paragraph{Communicating automata and programming languages}
The notion of multiparty compatibility is at the core of recent works
that apply session types techniques to mainstream programming
languages.
Ng and Yoshida~\citeout{NgY16} use the multiparty
compatibility defined in~\citeout{LTY15} to detect deadlocks in Go
programs. 
Hu and Yoshida~\citeout{HuY16} study the well-formedness of Scribble
protocols~\citeout{scribble} through the multiparty compatibility of
their projections.
These protocols are used to generate various endpoint APIs
implementing a Scribble specification~\citeout{HuY16,HY2017,NHYA2018} and
to produce runtime monitoring tools~\citeout{NY2017,NY2017b,NBY2017}.
Taylor et al.~\citeout{TaylorTWD16} use multiparty compatibility and
choreography synthesis~\citeout{LTY15} to automate the analysis of the
\texttt{gen\_server} library of Erlang/OTP.
We believe that we can transparently widen the set of safe programs
captured by these tools by using $k$\MC\ instead of synchronous
multiparty compatibility.

Desai et al.~\citeout{Desai14} propose a communicating automata-based
approach to verify safety properties of programs written in
P~\citeout{DesaiGJQRZ13}.
Their approach is based on exploring a subset of the (possibly
infinite) set of reachable configurations by prioritising certain
transitions in order to minimise the size of the queues.
Although the approach may not always terminate, they show that it is
sound and complete wrt.\ reachability of error configurations.
For instance the system in Figure~\ref{fig:desai}, adapted
from~\citeout[Section 9]{Desai14}, shows a system for which their
approach does \emph{not} terminate.
Note that this system is not existentially bounded and therefore it is
not $k$\MC\ for any $k$. It is however trivially 
existentially stable bounded since no stable configuration is reachable
except for the initial one.
An interesting area of future work is to investigate similar
priority-based executions of \CSA\ systems in order to check the
$k$\MC\ property more efficiently.

D'Osualdo et al.~\citeout{DOsualdoKO13} verify safety properties of
Erlang programs by inferring a model which abstracts away from message
ordering in mailboxes.
Their model is based on vector addition systems, for which the
reachability problem is decidable.
It would be interesting to adapt their approach to infer (mailbox)
communicating automata from Erlang programs.
Several approaches rely on \emph{sequentialization} of concurrent
programs~\citeout{QadeerW04,BouajjaniE14,BouajjaniEP11,InversoT0TP14,EmmiLQ12,Bakst17},
sometimes using bounded executions.
For instance, Bouajjani and Emmi~\citeout{BouajjaniE14} verify programs
that (asynchronously) send tasks to each other by considering
executions bounded by the number of times a sequence of tasks visits
the same process.
Bakst et al.~\citeout{Bakst17} address the verification of an
actor-oriented language (modelled on Erlang and Cloud Haskell) using
canonical sequentializations, which over-approximate a program.
They show that properties such as deadlock-freedom can be checked
efficiently.
Their approach requires the program to validate several structural
properties, one of which, \emph{symmetric non-determinism}, is
reminiscent of receive directedness as it requires
every receive action to only receive messages from a single process
(or a set of processes running the same code).
It would be interesting to relate symmetric non-determinism and
directedness more precisely, and consider systems of \CSA\ which
consist of several instances of some automaton.
In a similar line of work, von Gleissenthall et
al.~\citeout{GleissenthallKB19} use a form of sequentialization to the
verification of distributed systems implemented atop a Go library.
One of their restrictions is reminiscent of send directedness: they
allow participants to communicate to several others, however, they
require this communication to happen sequentially, \emph{``one
  interlocutor at a time''}.

\begin{figure}[t]
  \centering
  \begin{tabular}{ccc}
    $\p:$
    \begin{tikzpicture}[mycfsm, node distance = 0.6cm and 1cm]
      \node[state, initial, initial where=left] (s0) {};
      \node[state, right=of s0] (s1) {};
      \node[state, right=of s1] (s2) {};
            \path
      (s0) edge [loop] node [above] {$\PSEND{pr}{a}$} (s0)
      (s0) edge node [above] {$\PSEND{pr}{b}$} (s1)
      (s1) edge node [above] {$\PSEND{pq}{b}$} (s2)
      ;
    \end{tikzpicture}
    \qquad
    &
      $\ptp{q}:$
      \begin{tikzpicture}[mycfsm, node distance = 0.6cm and 1cm]
        \node[state, initial, initial where=above] (s0) {};
        \node[state, right=of s0] (s1) {};
        \node[state, right=of s1] (s2) {};
                \path
        (s0) edge node [above] {$\PRECEIVE{pq}{b}$} (s1)
        (s1) edge node [above] {$\PSEND{qr}{c}$} (s2)
        ;
      \end{tikzpicture} 
      \qquad
    &
      $\ptp{r}:$
      \begin{tikzpicture}[mycfsm, node distance = 0.6cm and 1cm]
        \node[state, initial, initial where=above] (s0) {};
        \node[state, right=of s0] (s1) {};
                \path
        (s0) edge node [above] {$\PRECEIVE{qr}{c}$} (s1)
        (s1) edge [loop] node [above] {$\PRECEIVE{pr}{a}$} (s1)  
        ;
      \end{tikzpicture}
  \end{tabular}
  \caption{Example of a non $\exists$-bounded system.}\label{fig:desai}
\end{figure}

\section{Additional examples}
\subsection{Examples for Section~\ref{sec:intro} (non-\SMC\ examples)}

The following example, implementing a simple rock-paper-scissors game,
is not \SMC\ but it is $1$\MC.
Messages $\msg{r}$ stands for ``rock'', $\msg{p}$ for ``paper'', and
$\msg{s}$ for ``scissors''.
At the end of the play, each participant sends their result
to the $\ptp{s}$erver. 
Note that this pattern cannot be specified in a synchronous way
without giving a clear advantage to one of the players.

\begin{center}

\begin{tabular}{ccc}
  $ \ptp{p}:$
  \begin{tikzpicture}[mycfsm, node distance = 0.9cm and 0.9cm]
            \node[state, initial, initial where=left] (s0) {};
    \node[state, right=of s0] (sent) {};
    \node[state, right=of sent] (received) {};
    \node[state, right=of received] (final) {};
        \path 
    (s0) edge [bend left=90] node [above] {$\PSEND{pq}{\grock}$} (sent)
    (s0) edge node [above] {$\PSEND{pq}{\gpaper}$} (sent)
    (s0) edge [bend right=90] node [above]  {$\PSEND{pq}{\gscissors}$} (sent)
        (sent) edge [bend left=90] node [above] {$\PRECEIVE{qp}{\grock}$} (received)
    (sent) edge node [above] {$\PRECEIVE{qp}{\gpaper}$} (received)
    (sent) edge [bend right=90] node [above]  {$\PRECEIVE{qp}{\gscissors}$} (received)
        (received) edge node [above] {$\PSEND{ps}{res}$} (final)
        ;
  \end{tikzpicture}
  &
    \qquad
    $\ptp{q}:$
    \begin{tikzpicture}[mycfsm, node distance = 0.9cm and 0.9cm]
                  \node[state, initial, initial where=left] (s0) {};
      \node[state, right=of s0] (sent) {};
      \node[state, right=of sent] (received) {};
      \node[state, right=of received] (final) {};
            \path 
      (s0) edge [bend left=90] node [above] {$\PSEND{qp}{\grock}$} (sent)
      (s0) edge node [above] {$\PSEND{qp}{\gpaper}$} (sent)
      (s0) edge [bend right=90] node [above]  {$\PSEND{qp}{\gscissors}$} (sent)
            (sent) edge [bend left=90] node [above] {$\PRECEIVE{pq}{\grock}$} (received)
      (sent) edge node [above] {$\PRECEIVE{pq}{\gpaper}$} (received)
      (sent) edge [bend right=90] node [above]  {$\PRECEIVE{pq}{\gscissors}$} (received)
            (received) edge node [above] {$\PSEND{qs}{res}$} (final)
            ;
    \end{tikzpicture}
      &
    \quad
    $\ptp{s}:$
    \begin{tikzpicture}[mycfsm, node distance = 0.4cm and 0.9cm]
                  \node[state, initial, initial where=left] (s0) {};
      \node[state, below=of s0] (s1) {};
      \node[state, below=of s1] (s2) {};
            \path 
      (s0) edge node [right] {$\PRECEIVE{ps}{res}$} (s1)
      (s1) edge node [right] {$\PRECEIVE{qs}{res}$} (s2)
            ;
    \end{tikzpicture}
\end{tabular}

 \end{center}

A more involved version of this game is given below. It is not \SMC,
but it is $1$\MC.
Note that $\ptp{s}$ is not directed, however the system is $\DIBI{1}$
since $\ptp{s}$ can only receive one of the expected messages.
\begin{center}
\begin{tabular}{ccc}
  $ \ptp{p}:$
  \begin{tikzpicture}[mycfsm, node distance = 0.6cm and 0.6cm]
            \node[state, initial, initial where=above] (s0) {\fontnode{0}};
        \node[state, left=of s0] (rock) {\fontnode{r}};
    \node[state, below=of s0] (paper) {\fontnode{p}};
    \node[state, right=of s0] (scissors) {\fontnode{s}}; 
        \node[state, left=of rock] (tie1) {\fontnode{t}};
    \node[state, above=of s0] (win1) {\fontnode{w}};
    \node[state, below=of rock] (lose1) {\fontnode{l}};
        \node[state, below=of paper] (win2) {\fontnode{w}};
        \node[state, below=of scissors] (tie2) {\fontnode{t}};
        \node[state, right=of scissors] (lose2) {\fontnode{l}};
    
    \node[state, above=of win1] (e1) {\fontnode{e}};
    \node[state, below=of win2] (e2) {\fontnode{e}};
        \node[state, above=of tie1] (t1) {\fontnode{e}};
    \node[state, right=of tie2] (t2) {\fontnode{e}};
              \path
    (win1) edge node [left] {$\PSEND{ps}{w}$} (e1)
    (win2) edge node [left] {$\PSEND{ps}{w}$} (e2)
          (tie1) edge node [left] {$\PSEND{ps}{t}$} (t1)
    (tie2) edge node [below] {$\PSEND{ps}{t}$} (t2)
          (s0) edge node [above] {$\PSEND{pq}{\grock}$} (rock)
    (s0) edge node [left] {$\PSEND{pq}{\gpaper}$} (paper)
    (s0) edge node [above]  {$\PSEND{pq}{\gscissors}$} (scissors)
        (rock) edge node [above] {$\PRECEIVE{qp}{\grock}$} (tie1)
    (rock) edge node [left] {$\PRECEIVE{qp}{\gpaper}$} (lose1)
    (rock) edge [bend left=20] node [above left]  {$\PRECEIVE{qp}{\gscissors}$} (win1)
        (paper) edge node [right] {$\PRECEIVE{qp}{\grock}$} (win2)
    (paper) edge node [above] {$\PRECEIVE{qp}{\gpaper}$} (tie2)
    (paper) edge node [below]  {$\PRECEIVE{qp}{\gscissors}$} (lose1)
        (scissors) edge node [above] {$\PRECEIVE{qp}{\grock}$} (lose2)
    (scissors) edge [bend right =20] node [above right] {$\PRECEIVE{qp}{\gpaper}$} (win1)
    (scissors) edge node [right]  {$\PRECEIVE{qp}{\gscissors}$} (tie2)
        ;
  \end{tikzpicture}
  &
    \quad
    $\ptp{q}:$
    \begin{tikzpicture}[mycfsm, node distance = 0.6cm and 0.6cm]
                  \node[state, initial, initial where=above] (s0) {\fontnode{0}};
            \node[state, left=of s0] (rock) {\fontnode{r}};
      \node[state, below=of s0] (paper) {\fontnode{p}};
      \node[state, right=of s0] (scissors) {\fontnode{s}}; 
            \node[state, left=of rock] (tie1) {\fontnode{t}};
      \node[state, above=of s0] (win1) {\fontnode{w}};
      \node[state, below=of rock] (lose1) {\fontnode{l}};
            \node[state, below=of paper] (win2) {\fontnode{w}};
            \node[state, below=of scissors] (tie2) {\fontnode{t}};
            \node[state, right=of scissors] (lose2) {\fontnode{l}};
          \node[state, above=of win1] (e1) {\fontnode{e}};
      \node[state, below=of win2] (e2) {\fontnode{e}};
 
            \path
      (win1) edge node [left] {$\PSEND{qs}{w}$} (e1)
      (win2) edge node [left] {$\PSEND{qs}{w}$} (e2)
                  (s0) edge node [above] {$\PSEND{qp}{\grock}$} (rock)
      (s0) edge node [left] {$\PSEND{qp}{\gpaper}$} (paper)
      (s0) edge node [above]  {$\PSEND{qp}{\gscissors}$} (scissors)
            (rock) edge node [above] {$\PRECEIVE{pq}{\grock}$} (tie1)
      (rock) edge node [left] {$\PRECEIVE{pq}{\gpaper}$} (lose1)
      (rock) edge [bend left=20] node [above left]  {$\PRECEIVE{pq}{\gscissors}$} (win1)
            (paper) edge node [right] {$\PRECEIVE{pq}{\grock}$} (win2)
      (paper) edge node [above] {$\PRECEIVE{pq}{\gpaper}$} (tie2)
      (paper) edge node [below]  {$\PRECEIVE{pq}{\gscissors}$} (lose1)
            (scissors) edge node [above] {$\PRECEIVE{pq}{\grock}$} (lose2)
      (scissors) edge [bend right =20] node [above right] {$\PRECEIVE{pq}{\gpaper}$} (win1)
      (scissors) edge node [right]  {$\PRECEIVE{pq}{\gscissors}$} (tie2)
            ;
    \end{tikzpicture}
  &
    \quad
    $\ptp{s}:$
    \begin{tikzpicture}[mycfsm, node distance = 0.6cm and 0.6cm]
      \node[state, initial, initial where=left] (s0) {\fontnode{0}};
            \node[state, above=of s0] (s1) {\fontnode{p}};
      \node[state, below=of s0] (s2) {\fontnode{q}};
      \node[state, right=of s0] (s3) {\fontnode{t}};
            \path 
      (s0) edge node [left] {$\PRECEIVE{ps}{w}$} (s1)
      (s0) edge node [left] {$\PRECEIVE{qs}{w}$} (s2)
      (s0) edge node [above] {$\PRECEIVE{ps}{t}$} (s3)
            ;
    \end{tikzpicture}
\end{tabular}

 \end{center}

\subsection{Example for Section~\ref{sec:kmc}: $\DIBI{k}$
  vs. $\CIBI{k}$ conditions}

We illustrate the difference between the $\DIBI{k}$ and $\CIBI{k}$
properties with the system below. It is adapted from the running
example of~\cite{LTY15} where we have removed mixed states (choosing
one interleaving for each outgoing transition). We refer to it as the
4 Player game in Table~\ref{tab:benchmarks}.
\[
\begin{array}{cc}
  M_\ptp{a}:
  \begin{tikzpicture}[mycfsm, node distance = 0.7cm and 1cm]
    \node[state, initial, initial where=above] (s0) {};
    \node[state, right=of s0] (s1) {};
    \node[state, below=of s1] (s2) {};
    \node[state, below=of s2] (s4) {};
    \node[state, right=of s1] (s3) {};
        \path
    (s0) edge [bend left] node [above] {$\PSEND{ab}{bwin}$} (s1)
    (s0) edge [bend right] node [below] {$\PSEND{ac}{cwin}$} (s1)
    (s1) edge node {$\PRECEIVE{ba}{sig}$} (s2)
    (s1) edge node {$\PRECEIVE{ca}{msg}$} (s3)
    (s2) edge node [left, near start] {$\PRECEIVE{ca}{msg}$} (s4)
    (s3) edge [bend left] node {$\PRECEIVE{ba}{sig}$} (s4)
    (s4) edge [bend left=50] node {$\PSEND{ad}{free}$} (s0)
    ;
  \end{tikzpicture}
  &
    M_\ptp{b}:
    \begin{tikzpicture}[mycfsm, node distance = 0.5cm and 1cm]
      \node[state, initial, initial where=above] (s0) {};
      \node[state, below left=of s0,xshift=-10pt] (s1) {};
      \node[state, below=of s1] at (s0|-s1) (s2) {};
      \path
      (s0) edge [bend right] node [left] {$\PRECEIVE{ab}{bwin}$} (s1)
      (s1) edge [bend right] node [left] {$\PSEND{bc}{close}$} (s2)
      (s0) edge node [left] {$\PRECEIVE{cb}{blose}$} (s2)
      (s2) edge [bend right=70] node [right] {$\PSEND{ba}{sig}$} (s0)
      ;
    \end{tikzpicture}
  \\
  M_\ptp{c}:
  \begin{tikzpicture}[mycfsm, node distance = 0.6cm and 1cm]
    \node[state, initial, initial where=above] (s0) {};
    \node[state, below left=of s0] (s2) {};
    \node[state, below right=of s2] (s4) {};
    \node[state, below right=of s0] (s5) {};
    \path
    (s0) edge [bend right] node [left] {$\PSEND{cd}{busy}$} (s2)
    (s2) edge [bend right] node [left] {$\PRECEIVE{ac}{cwin}$} (s4)
    (s2) edge node [above] {$\PRECEIVE{bc}{close}$} (s5)
    (s4) edge [bend right] node [right] {$\PSEND{cb}{blose}$} (s5)
    (s5) edge [bend right] node [right] {$\PSEND{ca}{msg}$} (s0)
    ;
  \end{tikzpicture}
  &
    M_\ptp{d}:
    \begin{tikzpicture}[mycfsm, node distance = 1cm and 1cm]
      \node[state, initial, initial where=above] (s0) {};
      \node[state, right=of s0] (s1) {};
      \path
      (s0) edge [bend left] node [above] {$\PRECEIVE{cd}{busy}$} (s1)
      (s1) edge [bend left] node [below] {$\PRECEIVE{ad}{free}$} (s0)
      ;
    \end{tikzpicture}
\end{array}
\]
This system is $\IBI{k}$ for all $k$ (and thus $\infIBI$): it is never
the case that $M_\ptp{b}$ (resp.\ $M_\ptp{c}$) can choose between
consuming $\msg{bwin}$ or $\msg{blose}$ (resp.\ $\msg{cwin}$ or
$\msg{close}$).
It is not $\DIBI{k}$ (for any $k$) because of the cyclic nature of the
protocol (both choices are available at each iteration).
However, this system is $\CIBI{k}$ because, $M_\ptp{a}$ need to
receive acknowledgements from both $M_\ptp{b}$ and $M_\ptp{c}$ before
starting a new iteration of the game; hence there is a dependency
between, e.g., $\PRECEIVE{ab}{bwin}$ and $\PSEND{cb}{blose}$.

\subsection{Example for Section~\ref{sec:new-completeness}: Local
  bound-agnosticity}
We illustrate the reason for using projections which preserve
$\emptyw$-transitions, i.e., $\epsproj{\kTS{S}}{p}$, to characterise
$k$-exhaustive systems, instead of projections which determinise the
automata, cf.~\cite{LTY15}.
Consider the system $S$ below.
\[
\begin{array}{ccc}
  \ptp{s}:
  \begin{tikzpicture}[mycfsm, node distance = 0.5cm and 0.3cm]
    \node[state, initial, initial where=above] (s0) {};
    \node[state, below left=of s0] (s1) {};
    \node[state, below=of s1] (s2) {};
    \node[state, below right=of s0] (s3) {};
        \path
    (s0) edge [bend right] node [left] {$\PSEND{sr}{x}$} (s1)
    (s0) edge [bend left] node [right] {$\PSEND{sr}{y}$} (s3)
    (s1) edge node [right] {$\PRECEIVE{ps}{a}$} (s2)
    ;
  \end{tikzpicture}
  \qquad
  &
    \ptp{p}:
    \begin{tikzpicture}[mycfsm, node distance = 0.5cm and 0.3cm]
      \node[state, initial, initial where=above] (s0) {};
      \node[state, below=of s0] (s1) {};
      \node[state, below=of s1] (s2) {};
            \path
      (s0) edge node [right] {$\PSEND{ps}{a}$} (s1)
      (s1) edge node [right] {$\PSEND{ps}{a}$} (s2)
      ;
    \end{tikzpicture}
    \qquad
  &
    \ptp{r}:
    \begin{tikzpicture}[mycfsm, node distance = 0.5cm and 0.3cm]
      \node[state, initial, initial where=above] (s0) {};
      \node[state, below left=of s0] (s1) {};
      \node[state, below right=of s0] (s2) {};
            \path
      (s0) edge [bend right] node [left] {$\PRECEIVE{sr}{x}$} (s1)
      (s0) edge [bend left] node [right] {$\PRECEIVE{sr}{y}$} (s2)
      ;
    \end{tikzpicture}
\end{array}
\]

The traditional projections ($\proj{\TS{k}{S}}{p}$) and projections
($\epsproj{\TS{k}{S}}{p}$) for $k \in \{1,2\}$ are given below
(up to (weak) bisimulation).
\[
\begin{array}{cc}
  \proj{\TS{1}{S}}{p} =  \proj{\TS{2}{S}}{p} =
  \begin{tikzpicture}[mycfsm, node distance = 0.5cm and 0.7cm]
    \node[state, initial, initial where=above] (s0) {};
    \node[state, below=of s0] (s1) {};
    \node[state, below=of s1] (s2) {};
        \path
    (s0) edge node [right] {$\PSEND{ps}{a}$} (s1)
    (s1) edge node [right] {$\PSEND{ps}{a}$} (s2)
    ;
  \end{tikzpicture}
  &
    \epsproj{\TS{1}{S}}{p} = 
    \begin{tikzpicture}[mycfsm, node distance = 0.5cm and 0.3cm]
      \node[state, initial, initial where=above] (s0) {};
      \node[state, below=of s0] (s1) {};
      \node[state, below left=of s1] (s2) {};
      \node[state, below right=of s1] (s3) {};
      \node[state, below=of s3] (s4) {};
            \path
      (s0) edge node [right] {$\PSEND{ps}{a}$} (s1)
      (s1) edge [bend right] node [above] {$\emptyw$} (s2)
      (s1) edge [bend left] node [above] {$\emptyw$} (s3)
      (s3) edge node [right] {$\PSEND{ps}{a}$} (s4)
      ;
    \end{tikzpicture}
  \\
  &
    \epsproj{\TS{2}{S}}{p} = 
    \begin{tikzpicture}[mycfsm, node distance = 0.5cm and 0.3cm]
      \node[state, initial, initial where=above] (s0) {};
      \node[state, below=of s0] (s1) {};
      \node[state, below left=of s1] (s2) {};
      \node[state, below=of s2] (s5) {};
      \node[state, below right=of s1] (s3) {};
      \node[state, below=of s3] (s4) {};
            \path
      (s0) edge node [right] {$\PSEND{ps}{a}$} (s1)
      (s1) edge [bend right] node [above] {$\emptyw$} (s2)
      (s1) edge [bend left] node [above] {$\emptyw$} (s3)
      (s3) edge node [right] {$\PSEND{ps}{a}$} (s4)
      (s2) edge node [left] {$\PSEND{ps}{a}$} (s5)
      ;
    \end{tikzpicture}
\end{array}
\]

Observe that we have $\proj{\TS{1}{S}}{p} \bisim \proj{\TS{2}{S}}{p}$,
but \emph{not}
\[
\epsproj{\TS{1}{S}}{p} \wbisim \epsproj{\TS{2}{S}}{p}
\]
Indeed, the system above is \emph{not} $1$\MC, but is $2$\MC.

\subsection{Examples for Section~\ref{sec:exist-bounded}:
  $\exists$-bounded vs.\ synchronisable systems}

\begin{example}\label{ex:not-cexist}
  $(M_\p, M_\q$) below is \emph{safe}, but not
  $\exists$(S)-$k$-bounded, nor $k$-exhaustive, for any $k$.
  \[
    \begin{array}{ccc}
      \p:
      \begin{tikzpicture}[mycfsm]
        \node[state, initial, initial where=left] (s0) {};
        \node[state, right=of s0] (s1) {};
        \node[state, right=of s1] (s2) {};
                \path
        (s0) edge [loop] node [above] {$\PSEND{pq}{a}$} (s0)
        (s0) edge node [below] {$\PSEND{pq}{b}$} (s1)
        (s1) edge [loop] node [above] {$\PRECEIVE{qp}{c}$} (s1)
        (s1) edge node [below] {$\PRECEIVE{qp}{d}$} (s2)
        ;
      \end{tikzpicture}
      \qquad    \qquad   
      &
        \ptp{q}:
        \begin{tikzpicture}[mycfsm]
          \node[state, initial, initial where=left] (s0) {};
          \node[state, right=of s0] (s1) {};
          \node[state, right=of s1] (s2) {};
                    \path
          (s0) edge [loop] node [above] {$\PSEND{qp}{c}$} (s0)
          (s0) edge node [below] {$\PSEND{qp}{d}$} (s1)
          (s1) edge [loop] node [above] {$\PRECEIVE{pq}{a}$} (s1)
          (s1) edge node [below] {$\PRECEIVE{pq}{b}$} (s2)
          ;
        \end{tikzpicture}
    \end{array}
  \] 
  For instance, execution $\acts$ below is
  $\mathit{max}\{m,n\}$-bounded. Hence, for any finite $k$, we can
  generate an execution that is not existentially (stable) $k$-bounded.
      \[\small
  \acts = 
  \underbracket{
    \PSEND{pq}{a} \cdots \PSEND{pq}{a}}_{n \text{ times}}
  \concat
  \PSEND{pq}{b}
  \concat 
  \underbracket{
    \PSEND{qp}{c} \cdots \PSEND{qp}{c}}_{m \text{ times}}
  \concat
  \PSEND{qp}{d}
  \concat
  \underbracket{
    \PRECEIVE{pq}{a} \cdots \PRECEIVE{pq}{a}}_{n \text{ times}}
  \concat
  \PRECEIVE{pq}{b}
  \concat 
  \underbracket{
    \PRECEIVE{qp}{c} \cdots \PRECEIVE{qp}{c}}_{m \text{ times}}
  \concat \PRECEIVE{qp}{d}
  \]
  Note that $\acts$ leads to a stable configuration (all
  sent messages are received).
\end{example}

\begin{example}
  The (non-$\infIBI$) system in Figure~\ref{fig:ex-nondirected} is
  \emph{not} $k$-synchronisable for any $k$, 
    due to executions consisting of the left branch of $M_\p$ and the
  right branch of $M_\q$ which are not synchronisable.
\end{example}

\begin{example}
  The system $(M_\p,M_\q)$ in Figure~\ref{fig:ex-unbounded} is
  \emph{not} $k$-synchronisable for any $k$.
      The system $(M_\p,N'_\q)$ is not $k$-synchronisable for any $k$
  since the second emission of message $\msg{b}$ cannot be received in the
  exchange from which it is sent.
        Instead, the system $(M_\p,N_\q)$ in Figure~\ref{fig:ex-unbounded}
  is $3$\synk\ since each of its executions can be rescheduled so to
  consists of the following $3$-exchange:
    
  \noindent
  {$
    \PSEND{pq}{a}
    \concat
    \PSEND{pq}{a}
    \concat
    \PSEND{qp}{b}
    \concat
    \PRECEIVE{pq}{a}
    \concat
    \PRECEIVE{pq}{a}
    \concat
    \PRECEIVE{qp}{b}
    $}.
\end{example}

\subsection{Example for Sections~\ref{sec:related}: mailbox communicating automata}

Consider the system $(M_\p, M_\rr, M_\q)$ below, with a mailbox
semantics, i.e., participant $\rr$ has one input queue to which both
participants $\p$ and $\q$ can send messages.
\[
\begin{array}{c}
  \ptp{p}:
  \begin{tikzpicture}[mycfsm, node distance = 0.5cm and 0.3cm]
    \node[state, initial, initial where=above] (s0) {};
    \node[state, below=of s0] (s1) {};
    \node[state, below=of s1] (s2) {};
    \path
    (s0) edge node [left] {$\PSEND{r}{a}$} (s1)
    (s1) edge node [left] {$\PSEND{r}{a}$} (s2)
    ;
  \end{tikzpicture}
    \qquad  \qquad
  \rr:
  \begin{tikzpicture}[mycfsm, node distance = 0.5cm and 0.3cm]
    \node[state, initial, initial where=above] (s0) {};
    \node[state, below=of s0] (s1) {};
    \node[state, below=of s1] (s2) {};
    \path
    (s0) edge node [left] {$\PRECEIVE{}{a}$} (s1)
    (s1) edge node [left] {$\PRECEIVE{}{a}$} (s2)
    ;
  \end{tikzpicture}
  \qquad  \qquad
  \ptp{q}:
  \begin{tikzpicture}[mycfsm, node distance = 0.5cm and 0.3cm]
    \node[state, initial, initial where=above] (s0) {};
    \node[state, below=of s0] (s1) {};
    \node[state, below=of s1] (s2) {};
    \path
    (s0) edge node [left] {$\PSEND{r}{a}$} (s1)
    (s1) edge node [left] {$\PSEND{r}{a}$} (s2)
    ;
  \end{tikzpicture}
\end{array}
\]
If this system executes with bound $k \leq 3$, one participant (either
$\p$ or $\q$) will be prevented to send at least one message.
This namely implies that the send action of participant may become
disabled after being enabled.
This is problematic for the current partial order reduction algorithm
and for the notion of $k$-closed sets used to prove our main results.

\subsection{Example for Section~\ref{sec:por}: (reduced) $\OBI{k}$}

The example below is \emph{reduced} $\OBI{k}$ for $k \geq 2$, but not
$\OBI{k}$ for any $k \geq 1$.
    $\TS{1}{S}$ includes a state where the queue $\ptp{pq}$ contains one
  message $\msg{a}$ and $M_\p$ is back and its initial state. At this point,
  $\PSEND{pr}{b}$ is fireable, but $\PSEND{pq}{a}$ is not.
    In $\RTS{2}{2}$, there is only one state from which $\p$ fires its send
  actions, both of which are enabled, hence the system is $\OBI{2}$.
    \[
  \begin{array}{ccc}
    \p:
    \begin{tikzpicture}[mycfsm, node distance = 0.6cm and 1cm]
      \node[state, initial, initial where=left] (s0) {};
      \node[state, right=of s0] (s1) {};
            \path
      (s0) edge [loop] node [above] {$\PSEND{pq}{a}$} (s0)
      (s0) edge node [above] {$\PSEND{pr}{b}$} (s1)
      ;
    \end{tikzpicture}
    \qquad
    &
      \ptp{q}:
      \begin{tikzpicture}[mycfsm, node distance = 0.6cm and 1cm]
        \node[state, initial, initial where=left] (s0) {};
                \path
        (s0) edge [loop] node [above] {$\PRECEIVE{pq}{a}$} (s0)
        ;
      \end{tikzpicture} 
      \qquad
    &
      \ptp{r}:
      \begin{tikzpicture}[mycfsm, node distance = 0.6cm and 1cm]
        \node[state, initial, initial where=left] (s0) {};
        \node[state, right=of s0] (s1) {};
                \path
        (s0) edge node [above] {$\PRECEIVE{pr}{b}$} (s1)
        ;
      \end{tikzpicture}
  \end{array}
  \]

\subsection{Example for Section~\ref{sec:por}: ordered list}
We illustrate the motivation to sort the list generated by
$\spartition{\_}$, see Definition~\ref{def:partition}, with the system
below.
\[
\p:
\begin{tikzpicture}[mycfsm, node distance = 0.5cm and 0.3cm]
  \node[state, initial, initial where=above] (s0) {0};
  \node[state, below=of s0] (s1) {1};
    \path
  (s0) edge node [right] {$\PSEND{pq}{a}$} (s1)
  ;
\end{tikzpicture}
\qquad
\q :
\begin{tikzpicture}[mycfsm, node distance = 0.5cm and 0.3cm]
  \node[state, initial, initial where=above] (s0) {0};
\end{tikzpicture}
\qquad
\s :
\begin{tikzpicture}[mycfsm, node distance = 0.5cm and 0.3cm]
  \node[state, initial, initial where=above] (s0) {0};
  \node[state, below=of s0] (s1) {1};
    \path
  (s0) edge [bend right] node [left] {$\PSEND{sr}{x}$} (s1)
  (s0) edge [bend left] node [right] {$\PSEND{sr}{y}$} (s1)
  ;
\end{tikzpicture}
\qquad
\rr :
\begin{tikzpicture}[mycfsm, node distance = 0.5cm and 0.3cm]
  \node[state, initial, initial where=above] (s0) {0};
\end{tikzpicture}
\]
If we were to build the $\kRTS{S}$ of this system without sorting the
list returned by $\spartition{s_0}$.
We may obtain
$\spartition{s_0} = \{\PSEND{sr}{x}, \PSEND{sr}{y} \} \concat \{
\PSEND{pq}{a} \}$, which produces $4$ transitions (and $5$ states).
Instead, if the list is sorted by ascending cardinality, we have
$\spartition{s_0} =\{ \PSEND{pq}{a} \} \concat \{\PSEND{sr}{x},
\PSEND{sr}{y} \}$,
which gives us an $\kRTS{S}$ with $3$ transitions (and $4$ states).

\begin{remark}
  Note that even though sorting sets of transitions by cardinality
  gives better performance in general, it does guarantee to find the
  smallest $\RTS{k}{S}$.
\end{remark}

 \section{Proofs for Section~\ref{sec:cfsm} (preliminaries)}

\propdirectediobi* 
\begin{proof}
  Immediate since each directed (\CSA) automaton has access to at most one channel
  from each state.
\end{proof}

\begin{restatable}{lemma}{lemvalidword}\label{lem:valid-word}
  Let $S$ be a system and $\acts \in \ASetC$. If
  $s_0 \kTRANSS{\acts}$, then $\acts$ is a \emph{valid} execution.
\end{restatable}
\begin{proof}
  By induction on the length of $\acts$.  The result follows trivially
  for $\acts = \emptyw$.  Assume it holds for $\acts$ and let us show
  that is also holds for $\acts \concat \action$.
    Assume $\chan{\action} = \ptp{pq}$.
    By induction hypothesis, for each prefix $\actsb$ of
  $\acts$, we have that $\ercvproj{\actsb}{sr}$ is a prefix of
  $\esndproj{\actsb}{sr}$ for any channel $\ptp{sr} \in \CSet$.
    Hence, for each prefix $\actsb$ of $\acts \concat \action$ we have
  that $\ercvproj{\actsb}{sr}$ is a prefix of $\esndproj{\actsb}{sr}$
  for any channel $\ptp{sr} \neq \ptp{pq} \in \CSet$.
    If $\action = \PSEND{pq}{a}$, the result still holds since 
  $\esndproj{\actsb}{sr}$ is longer or equal.
    The interesting case is when $\action = \PRECEIVE{pq}{a}$.
      Pose $\esndproj{\acts}{pq} = \ercvproj{\acts}{pq} \concat \word$
  (there is such $\word$ by induction hypothesis).
    Assume by contradiction that $\acts \concat \PRECEIVE{pq}{a}$ is
  not a valid word.
    Then, there is no $\word' \in \ASigmaC$ such that
  $\esndproj{\acts}{pq} = \esndproj{\acts \concat
    \PRECEIVE{pq}{a}}{pq} = \ercvproj{\acts \concat
    \PRECEIVE{pq}{a}}{pq} \concat \word'$.
  which implies that either $\word = \msg{b} \concat \word''$ or
  $\word = \emptyw$ ($\msg{b} \neq \msg{a}$).
    This contradicts the fact that
  $s_0 \kTRANSS{\acts} s \kTRANSS{\PRECEIVE{pq}{a}}$ since the channel
  $\ptp{pq}$ in $s$ is either empty or starts with $\msg{b}$.
\end{proof}

\begin{lemma}\label{lem:eqpeer-imp-same-state}
  Let $S$ be a system.
    If $s_0 \TRANSS{\actsb_0} s$, $s \TRANSS{\acts} t$, and
  $s \TRANSS{\acts'} t'$ such that $\eqpeer{\acts}{\acts'}$, then
  (1) $t = t'$ and (2) $\acts_0 \concat \acts \resche \acts_0 \concat \acts'$.
\end{lemma}
\begin{proof}
  Item (1) follows from the fact that the automata are deterministic
  hence, they all terminate in the same state, and the queues are
  consumed uniformly in both executions.
    Item (2) follows from the fact that both executions are valid, by
  Lemma~\ref{lem:valid-word}.
\end{proof}

\section{Proofs for Section~\ref{sec:kmc} ($k$\MC)} 

\thmdecidabilityall*

\begin{proof}
  We first observe that decidability follows straightforwardly since
  for any finite $k$, both $\RS_k(S)$ and $\kTRANSS{}$ are finite.
    We follow the proof of~\cite[Theorem 6.3]{bollig2010}.
    Let $n$ be the maximum of
  $\{ \lvert Q_\p \rvert \st \p \in \PSet \}$, then there are at most
  $n \lvert \PSet \rvert$ local states in $S$.

  \proofsub{$k$-exhaustivity}
    We check whether $S$ is \emph{not} $k$-exhaustive, i.e., for each
  sending state $q_\p$ and send action from $q_\p$, we check whether
  there is a reachable configuration from which this send action cannot
  be fired.
    Hence, we need to search $\RS_k(S)$, which has an exponential number
  of states (wrt.\ $k$).
    Following~\cite[Theorem 6.3]{bollig2010}, each configuration
  $s \in \RS_k(S)$ may be encoded in space
  \[
  \lvert \PSet \rvert \log n
    + \lvert \CSet \rvert k \, \log {\lvert \ASigma \rvert}
  \]
    We also need one bit to remember whether we are looking for
  $q_\p$ or whether we are looking for the matching action.
    We need to store at most
  $\lvert \PSet \rvert n \lvert \CSet \rvert \lvert \ASigma \rvert^k$
  configurations, hence the problem can be decided in polynomial
  space when $k$ is given in unary.

                                          Next, we show that the problem is \PSPACE{}-hard.
    From~\cite[Proposition 5.5]{GenestKM07}, we know that checking
  existentially stable $k$-boundedness for a system with the stable
  property is \PSPACE{}-complete.
    By Theorem~\ref{thm:kuske-imp-classical}, this problem can be reduced
  to checking whether the system is $k$-exhaustive, which implies that
  checking $k$-exhaustivity must be \PSPACE{}-hard.

  \proofsub{$\OBI{k}$}
    For each sending state $q_\p$, we check whether there is a reachable
  configuration from which not all send actions can be fired, and thus
  we reason similarly to the $k$-exhaustivity case.
    Next, we show that checking $\OBI{k}$ is \PSPACE{}-hard. For this we
  adapt the construction from~\cite[Theorem 3]{Bouajjani2018} which
  reduces the problem of checking if the product of a set of finite
  state automata has an empty language to checking
  $1$-synchronisability.
    We use the same construction as theirs (which is $\OBI{1}$) but
  instead of adding states and transitions to ensure that the system
  breaks $1$-synchronisability when each automata is in a final state,
  we add states and transitions that violate $\OBI{1}$ (using a
  construction like the one in Example~\ref{ex:kobi-not-por}).

  \proofsub{$\IBI{k}$}
    For each non-directed receiving state $q_\p$, we check whether there
  is a reachable configuration from which more than one receive action
  can be fired, and thus we reason similarly as for $k$-exhaustivity.
    Showing that $\IBI{k}$ is \PSPACE{}-hard is similar to the $\OBI{k}$
  case.

  \proofsub{$\DIBI{k}$}
    There are two components of this property, one is equivalent to
  $\IBI{k}$, the other requires to guarantee that no matching send
  action is fired from an already enabled receive state. 
    Hence, for each non-directed receiving state $q_\p$, we check
  whether there is a reachable configuration from which one receive
  action of $\p$ is enabled, followed by a send action that matches
  another receive.
    We can proceed as in the case for $k$-exhaustivity with additional
  space to remember whether we are looking for the receiving state or
  for a matching send action.
    Showing that $\DIBI{k}$ \PSPACE{}-hard is similar to the $\OBI{k}$
  case.

  \proofsub{${k}$-safety}
    For eventual reception, we proceed as in $\DIBI{k}$ for each
  receiving state and element of the alphabet (check if such a
  configuration is reachable, then we search for a matching receive).
    For progress, we proceed as in $\DIBI{k}$ for each receiving state $q_\p$
  (check if such a configuration is reachable, then we search for a
  move by $\p$).
    Showing that checking ${k}$-safety \PSPACE{}-hard is similar to the
  $\OBI{k}$ case.
                                                \end{proof}

\begin{restatable}{lemma}{lemkclosedpclosed}\label{lem:kclosed-pclosed}
  Let $S$ s.t.\ $s \in \RS_{k}(S)$ and
  $\paset \subseteq \ASetC$ such that $\kclosed{\paset}{s}$, then
    $\pclosed{\paset}{s}$.
\end{restatable}
\begin{proof}
  The result follows from Definition~\ref{def:kclosed}, since
  $\kTRANSS{} \subseteq \bTRANSS{}{k+1}$.
\end{proof}

\lemkmckclosed*
\begin{proof}
  The non-emptiness of $\paset$ follows easily from the assumption
  that $S$ is $k$-exhaustive (Definition~\ref{def:exhaustive}).
    We have to show the following two conditions hold:
  
  \proofsub{1}
  $\forall \acts \in \paset \qst \exists s' \in \RS_k(S) \qst s
  \kTRANSS{\acts} s'$, which follows trivially from the definition of $\paset$.
  
  \proofsub{2}
  For all $\acts_0 \concat \PSEND{sr}{b} \concat \acts_1
  \in \paset$ such that $s \TRANSS{\acts_0} \csconf q w$ and
  for all $(q_\s, \action, q'_\s) \in \delta_\s$ there is
  $\acts_0 \concat  \action \concat \acts_2 \in
  \paset$.
    For this part, take
  $\acts_0 \concat \PSEND{sr}{b} \concat \acts_1 \in \paset$ such that
  $s \TRANSS{\acts_0} s' = \csconf q w$ (with $\s \neq \p$ by
  definition of $\paset$).
    By definition of $\paset$, we have $s' = \csconf q w \in \RS_k(S)$.
                  
  Since $S$ is $k$-exhaustive, for each
  $(q_\s, \PSEND{st}{c}, q'_\s) \in \delta_\s$ there is $\actsb$ s.t.\
  we obtain the following situation
    (where each arrow indicated a $k$-bounded execution):
  \[
  \begin{tikzpicture}[baseline=(current bounding box.center)]
    \node (s1) {$s'$};
    \node[right=of s1] (s2) {};
    \node[below=of s1] (s3) {$t$};
    \node at (s2|-s3) (s4) {$t'$};
        \path[->] (s1) edge [above] node {$\PSEND{sr}{b}$} (s2);
    \path[->] (s3) edge [above] node {$\PSEND{st}{c}$} (s4);
        \path[->] (s1) edge [left] node {$\actsb$} (s3);
      \end{tikzpicture}
    \qquad
  \text{with } \s \notin \actsb
  \]
  There are two cases: 
  \begin{itemize}
  \item If $\p \notin \actsb$, we have that the local state of $\p$ in
    configurations $s$, $s'$ and $t$ is the same.
        Hence, by $k$-exhaustivity:
    $t' \kTRANSS{\actsb'} \kTRANSS{\PSEND{pq}{a}}$ with
    $\p \notin \actsb$.
        Therefore,
    $\acts_0 \concat \actsb \concat \PSEND{st}{c} \concat \actsb' \in
    \paset$ as required.
      \item If there is no $\acts$ such that $\p \notin \actsb$, then
    there must be a dependency chain in $\actsb$ that prevents
    $\PSEND{st}{c}$ to be fired without $\p$ making a move. 
        Since $s \notin \actsb$, we must have some $\PRECEIVE{st}{d}$ in
    $\actsb$ such that $\PRECEIVE{st}{d}$ depends on an action by
    $\p$.
        The smallest such chain is of the form:
    $\PSEND{pt}{x} \concat \PRECEIVE{pt}{x} \concat
    \PRECEIVE{st}{y}$.
        Without loss of generality, pose
    $\actsb = \PSEND{pt}{x} \concat \PRECEIVE{pt}{x} \concat
    \PRECEIVE{st}{y}$ (we reason similarly with a longer chain).

    Take $\acts_3$ s.t.\ $s_0 \kTRANSS{\acts_3} s$, since $S$ is
    reduced $\OBI{k}$ and $k$-exhaustive, there are $t''$ and
    $\actsb_0$ such that $s_0 \rkTRANSS{\actsb_0} t''$, and $\acts_4$
    s.t.\ $t' \kTRANSS{\acts_4} t''$, with
        \[
    \actsb_0 \eqpeerop \acts_3 \concat \acts_0 \concat \PSEND{pt}{x}
    \concat \PRECEIVE{pt}{x} \concat \PRECEIVE{st}{y} \concat
    \PSEND{st}{c} \concat \acts_4
    \]
        by Lemma~\ref{lem:por-trace-equiv} (2).
        Hence, due to the dependency chain within $\actsb$, we must have:
    \[ 
    \actsb_0 = 
    \actsb_1 \concat \PSEND{pt}{x} \concat
    \actsb_2 \concat \PRECEIVE{pt}{x}  \concat
    \actsb_3 \concat \PRECEIVE{st}{y} \concat
    \actsb_4 \concat \PSEND{st}{c}  \concat
    \actsb_5 
    \]
    with $\s \notin \actsb_2 \concat \actsb_3 \concat \actsb_4 $.
    There are three cases:
    \begin{itemize}
    \item Either $\PSEND{sr}{b}$ is $k$-enabled immediately after
      $\actsb_1$, in which case we have a contradiction with the fact
      that $S$ is reduced $\OBI{k}$, 
          \item  $\PSEND{sr}{b}$ is $k$-enabled strictly after $\actsb_1$
      and strictly before $\PSEND{st}{c}$, then we have a
      contradiction with the fact that $S$ is reduced $\OBI{k}$, or
          \item $\PSEND{sr}{b}$ is not $k$-enabled along $\actsb_0$, which
      is also a contradiction with the fact that $S$ is reduced
      $\OBI{k}$.
            \qedhere
    \end{itemize}
  \end{itemize}
                                \end{proof}

Given $\acts = \action_1 \cdots \action_n \in \ASetC$, we write
$\subj{\acts}$ for the set
$\bigcup_{1\leq i \leq n} \{\subj{\action_i}\}$.

\begin{restatable}{lemma}{lemindeppath}\label{lem:indep-path}
  If $s \kTRANSS{\acts} t$ and $s \kTRANSS{\actsb}t'$ and
  $\subj{\acts} \cap \subj{\actsb} = \emptyset$, then there is $s'$ such that
  $t \kTRANSS{\actsb}s'$ and $t' \kTRANSS{\acts} s'$.
\end{restatable}
\begin{proof}
  Straightforward: the executions are independent from one another.
\end{proof}

\lemnewkclosedset*
\begin{proof}
      Let us pose $\subj{\action} = \p$.

  \proofsub{1}
      We first observe that $\hat\paset$ validates condition (1) of
  Definition~\ref{def:kclosed}, i.e.,
  $\forall \acts \in \hat\paset \qst \exists s'' \in \RS_k(S) \qst s'
  \bTRANSS{\acts}{k} s''$, by definition of $\hat\paset$.
    We then show that $\hat\paset$ validates the second condition of
  $k$-closure.
      There are two cases depending on whether the execution is in
  $\hat\paset_1$ or $\hat\paset_2$.
    \begin{enumerate}
  \item Take
    $\acts = \acts_0 \concat \PSEND{sr}{a} \concat \acts_1 \in
    \hat\paset_1$,
    then by definition of $\hat\paset_1$, we have $\p \neq \s$ and
    $\acts \in \paset$.
        Hence, posing $s \TRANSS{\acts_0} \csconf q w$, we have that for all
    $(q_\s, \action', q'_\s) \in \delta_\s$ there is
    $\acts_0 \concat \acts_1 \concat \action' \concat \acts_2 \in \paset$,
    with $\subj{\action'} \notin \acts_1$,
    since
    $\paset$ is $k$-closed by assumption.
    \begin{enumerate}
    \item If $\p \notin \acts_2$, then
      $\acts_0 \concat \acts_1 \concat \action' \concat \acts_2 \in
      \hat\paset_1$, as required.
          \item If $\p \in \acts_2$, then there are two cases depending on
      whether $\action$ is a send or a receive action.
      \begin{itemize}
      \item If $\action = \PRECEIVE{qp}{a}$, then we must have
        $\acts_2 = \acts_3 \concat \PRECEIVE{qp}{a} \concat \acts_4$
        with $\p \notin \acts_3$, since $S$ is $\IBI{k}$ (only one
        receive action can be enabled at $\p$).
                Thus
        $\acts_0 \concat \acts_1 \concat \action' \concat \acts_3
        \concat \acts_4 \in \hat\paset_2$, as required.
                      \item If $\action = \PSEND{pq}{a}$, then we must have $\acts_2 =
        \acts_3 \concat \PSEND{pt}{b} \concat \acts_4$ with $\p
        \notin \acts_3$. 
                Since $\paset$ is $k$-closed, we also have
        $\acts_0 \concat \acts_1 \concat \action' \concat \acts_3
        \concat \acts_4 \concat \PSEND{pq}{a} \concat \acts_5 \in \paset$,
      
        for some $\acts_4, \acts_5$ s.t.\ $\p \notin \acts_4$.
                Thus,
        $\acts_0 \concat \acts_1 \concat \action' \concat \acts_3
        \concat \acts_4 \concat \acts_5 \in \hat\paset_2$,
        as required.
              \end{itemize}
    \end{enumerate}

    \item Take $\acts = \acts_0 \concat \PSEND{sr}{a} \concat \acts_1 \in
    \hat\paset_2$. There are two cases:
    \begin{enumerate}
    \item If $\acts_0 = \acts_2 \concat \acts_3$ and
      $\acts_2 \concat \action \concat \acts_3 \concat \PSEND{sr}{a}
      \concat \acts_1 \in \paset$,
      then posing
      $s \TRANSS{\acts_2 \concat \action \concat \acts_3} \csconf q
      w$,
      we have that for all $(q_\s, \action', q'_\s) \in \delta_\s$
      there is
      $\acts_2 \concat \action \concat \acts_3\concat \acts_5 \concat
      \action' \concat \acts_4 \in \paset$
      (for some $\acts_4$ and $\acts_5$ s.t.\ $\s \notin \acts_5$)
      since $\paset$ is $k$-closed by assumption.
            Thus,
      $\acts_2 \concat \acts_3\concat \acts_5 \concat \action' \concat
      \acts_4 \in \hat\paset_2$, as required.

    \item If $\acts_1 = \acts_2 \concat \acts_3$ and
      $\acts_0 \concat \PSEND{sr}{a} \concat \acts_2 \concat \action
      \concat \acts_3 \in \paset$, then
            $\p \notin \acts_0 \concat \PSEND{sr}{a} \concat \acts_2$ (hence
      $\p \neq \s$) and,
            posing $s \TRANSS{\acts_0} \csconf q w$, we have that for all
      $(q_\s, \action', q'_\s) \in \delta_\s$ there is
      $\acts_0 \concat \acts_8 \concat \action' \concat \acts_4 \in
      \paset$
      (for some $\acts_4$ and $\acts_8$ s.t.\ $\s \notin \acts_8$)
      since $\paset$ is $k$-closed by assumption.
      \begin{itemize}
      \item if $\p \notin \acts_4$, then
        $\acts_0 \concat \acts_8 \concat \action' \concat \acts_4 \in
        \hat\paset_1$, as required.
      \item if $\p \in \acts_4$, there are two cases depending on
        whether $\action$ is a receive or send action.
        \begin{itemize}
        \item if $\action$ is a receive action, then we must have
          $\acts_4 = \acts_5 \concat \action \concat \acts_6$ with
          $\p \notin \acts_5$, thus
          $\acts_0 \concat \acts_8 \concat \action' \concat \acts_5
          \concat \acts_6 \in \hat\paset_2$,
          as required, since $S$ is $\IBI{k}$ (only one receive action
          can be enabled at $\p$)
                            \item if $\action$ is a send action, pose $\action =
          \PSEND{pq}{c}$, then we must have
             $\acts_4 = \acts_5 \concat \PSEND{pt}{b} \concat \acts_6$ with
             $\p \notin \acts_5$.
                          Since $\paset$ is $k$-closed, we must also have
             $\acts_0 \concat \action' \concat \acts_5 \concat \acts_9
             \concat \PSEND{pq}{c} \concat \acts_7 \in \paset$
             (for some $\acts_7$ and $\acts_9$ s.t.\
             $\p \notin \acts_9$).
                          Thus,
             $\acts_0 \concat \action' \concat \acts_5 \concat \acts_9
             \concat \acts_7 \in \hat\paset_2$, as required.
                                  \end{itemize}
      \end{itemize}
    \end{enumerate}
  \end{enumerate}

  \proofsub{2}
  Take $\actsb \in \hat\paset$, by definition of $\hat\paset$, there
  are two cases:
  \begin{enumerate}
  \item If $\actsb \in \hat\paset_1$, then $\actsb =
    \acts \in \paset$ and since $\subj{\action} \notin \actsb$,
        $s \kTRANSS{\actsb} t \kTRANSS{\action} t'$ and
        $s' \kTRANSS{\actsb} t'$ by Lemma~\ref{lem:indep-path}. In picture, we have
  \[
    \begin{tikzpicture}[node distance=0.7cm and 0.7cm]
      \node (s1) {$s$};
      \node[right=of s1] (s2) {$s'$};
      \node[below=of s1] (s3) {$t$};
      \node at (s2|-s3) (s4) {$t'$};
            \path[->] (s1) edge [above] node {$\action$} (s2);
      \path[->] (s3) edge [above] node {$\action$} (s4);
            \path[->] (s1) edge [left] node {$\actsb = \acts$} (s3);
      \path[->] (s2) edge [right] node {$\actsb = \acts$} (s4);
    \end{tikzpicture}
      \]
  Finally, we have $\eqpeer{\acts \concat \action}{\action \concat
    \actsb}$ since $\subj{\action} \notin \actsb$.
  \item If $\actsb \in \hat\paset_2$, then there is
  $\acts = \acts_0 \concat \action \concat \acts_1 \in \paset$ s.t.\
  $\actsb = \acts_0 \concat \acts_1$ and $\subj{\action} \notin
  \acts_0$.
    Thus, by Lemma~\ref{lem:indep-path} we have
  $s \kTRANSS{\acts_0 \concat \action \concat \acts_1 } t$ and
  $s \kTRANSS{\action} s' \kTRANSS{\acts_0 \concat \acts_1} t$, i.e.,
    \begin{equation}
        \begin{tikzpicture}[node distance=0.7cm and 0.7cm]
      \node (s1) {$s$};
      \node[right=of s1] (s2) {$s'$};
      \node[below=of s1] (s3) {$.$};
      \node[below=of s3] (s5) {$.$};
      \node[below=of s5] (s7) {$t$};
            \path[->] (s1) edge [above] node {$\action$} (s2);
            \path[->] (s1) edge [left] node {$\acts_0$} (s3);
      \path[->] (s2) edge [bend left, right] node {$\acts_0$} (s5);
      \path[->] (s3) edge [left] node {$\action$} (s5);
      \path[->] (s5) edge [left] node {$\acts_1$} (s7);
    \end{tikzpicture}
  \end{equation}
  Finally, we have
  $\eqpeer{\acts_0 \concat \action \concat \acts_1}{\action \concat \acts_0 \concat \acts_1}$
  since $\subj{\action} \notin \acts_0$.
  \end{enumerate}

  \proofsub{3}
  The ($\Rightarrow$) direction is trivial from the definition of
  $\hat\paset$.
    Let us show that $\hat\paset = \emptyset \implies \paset=\emptyset$
  by contradiction.
    Assume  $\hat\paset = \emptyset$ and  $\paset \neq \emptyset$.
    This implies that for all $\acts \in \paset \qst \p \in \acts$.
    Pose $\acts = \acts_0 \concat \hat\action \concat \acts_1$, with
  $\p \notin \acts_0$, $\action \neq \hat\action$.
  \begin{itemize}
  \item If $\action$ is a receive action, then $\hat\action$ is also a
    receive action ($\p \notin \acts_0$), thus
    $\action \neq \hat\action$ contradicts the assumptions that
    $s\kTRANSS{\action}$ and $\p \notin \acts_0$.
      \item If $\action$ is a send action, then $\hat\action$ is also a
    send action ($\p \notin \acts_0$), thus it is a contradiction with
    the fact that $\kclosed{\paset}{s}$.
    \qedhere
  \end{itemize}
\end{proof}

\lemclosedsetpaths*
\begin{proof}
  By replicated application of Lemma~\ref{lem:new-kclosed-set}
  (parts~\ref{it:new-kclosed-set-part-a}
  and~\ref{it:new-kclosed-set-part-c}), for all $1 \leq i
  \leq n$, there is $\emptyset \neq \paset_i \subseteq \ASetC$ such that
  $\kclosed{\paset_i}{s_i}$.
      In addition, by Lemma~\ref{lem:new-kclosed-set}
  (part~\ref{it:new-kclosed-set-part-b}), for all $1 \leq i < n$,
  and for all $\acts_{i+1} \in \paset_{i+1}$, there is $\acts_i \in
  \paset_i$ such that either
  \begin{itemize}
  \item $s_{i+1} \kTRANSS{\acts_{i+1}} t_{i+1} $, and
    $s_i \kTRANSS{\acts_i} t_i$, with $t_i = t_{i+1}$, or
  \item $s_{i+1} \kTRANSS{\acts_{i+1}} t_{i+1} $, $s_i \kTRANSS{\acts_i} t_i$, and $t_i \kTRANSS{\action_i} t_{i+1}$.
  \end{itemize}

  The rest of the proof is by induction on $n$.
  
  \proofsub{Base case} 
    If $n=2$, then the result follows directly by instantiating
  Lemma~\ref{lem:new-kclosed-set} with $s_1 = s$, $s_n = s'$, and
  $\action_1 =\action$, in particular, we have
    $\actsb = \action_1$ or $\actsb = \emptyw$ (hence
  $\lvert \actsb \rvert < n$).

  \proofsub{Inductive case}
    Assume the result holds for $n=i$ (i.e.,
  $\eqpeer{\acts_1 \concat \actsb}{\action_1 \cdots \action_{i-1} \cdot
    \acts_i}$) and let us show that it holds for $n=i{+}1$.
    We have the following situation:
      \[
    \begin{tikzpicture}[node distance=0.8cm and 1cm]
      \node (s1) {$s_i$};
      \node[gray,right=of s1] (s2) {$s_{i+1}$};
      \node[below=of s1] (s3) {$t_i$};
      \node[gray] at (s2|-s3) (s4) {$t_{i+1}$};
            \path[gray,->] (s1) edge [above] node {$\action_i$} (s2);
      \path[gray,->] (s3) edge [above] node {$\actsb'$} (s4);
            \path[->] (s1) edge [right] node  {$\acts_i$} (s3);
      \path[gray,->] (s2) edge [right] node {$\acts_{i+1}$} (s4);
            \node[ left=of s1] (sm1) {};
      \node[] at (sm1|-s3)  (tm1) {};
                              \node[left=of sm1] (sm2) {$s_1$};
      \node[] at (sm2|-tm1)  (tm2) {$t_1$};
            \path[dashed,->] (sm2) edge [above] node {$\action_1 \cdots \action_{i-1}$} (s1);
      \path[dashed,->] (tm2)  edge [above] node {$\actsb$} (s3);
      \path[->] (sm2) edge [right] node {$\acts_1$} (tm2);
    \end{tikzpicture}
      \]

  By Lemma~\ref{lem:new-kclosed-set}, we have
  either
  \begin{enumerate}
  \item \label{it:empty-actprime}
    $t_i = t_{i+1}$, $\actsb' = \emptyw$, and
    $\eqpeer{\acts_i \concat \emptyw}{\action_i \concat \acts_{i+1}}$.
      \item \label{it:action-actprime}
    $\actsb' = \action_i$, $\acts_i = \acts_{i+1}$ 
    and $\eqpeer{\acts_i \concat \action_i}{\action_i \concat \acts_{i}}$.
  \end{enumerate}

  We have to show that
   \[
   \eqpeer{\acts_1 \concat \actsb \concat \actsb'}{\action_1 \cdots \action_{i-1} \concat \action_i \concat
     \acts_{i+1}}
  \]

  \begin{itemize}
  \item Assume case~\eqref{it:empty-actprime} holds.
    \[
    \begin{array}{ccl@{\quad}r}
      {\acts_1 \concat \actsb} & \eqpeerop & {\action_1 \cdots \action_{i-1} \cdot
                                              \acts_i} & \text{by induction hypothesis}                                          
      \\
                                        & \eqpeerop & {\action_1 \cdots \action_{i-1} \cdot
                                              \acts' \concat \action_i \concat \acts'' } 
                                                       & \text{posing } \acts_i =  \acts' \concat \action_i \concat \acts''
                                                          \text{ with } \subj{\action_i} \notin \acts'
      \\
                                        & \eqpeerop & {\action_1 \cdots \action_{i-1} \cdot
                                              \action_i \concat \acts' \concat \acts'' }
                                                       & \text{since } \subj{\action_i} \notin \acts'
      \\
                                        & \eqpeerop & {\action_1 \cdots \action_{i-1} \cdot
                                              \action_i \concat \acts_{i+1}}
                                                       & \text{by Lemma~\ref{lem:new-kclosed-set} }
    \end{array}
    \]
        Finally, since $\actsb' = \emptyw$ in this case, we have
    $\acts_1 \concat \actsb \concat \actsb' = \acts_1 \concat \actsb$,
    hence
    \[
    \eqpeer{\acts_1 \concat \actsb}{\action_1 \cdots \action_{i-1} \cdot
      \action_i \concat \acts_{i+1}}
    \]
    as required.
      \item  Assume case~\eqref{it:action-actprime} holds.
    \[
    \begin{array}{ccl@{\quad}r}
      {\acts_1 \concat \actsb} & \eqpeerop & {\action_1 \cdots \action_{i-1} \cdot
                                              \acts_i} & \text{by induction hypothesis}     
      \\
      {\acts_1 \concat \actsb \concat \action_i } & \eqpeerop & {\action_1 \cdots \action_{i-1} \cdot
                                                               \acts_i \concat \action_i}
                                                       & \text{by Lemma~\ref{lem:eqpeer-imp-same-state}}
      \\
                                        & \eqpeerop & {\action_1 \cdots \action_{i-1} \concat
                                            \action_i \concat \acts_i}
                                                       & \text{by  case~\eqref{it:action-actprime}}
      \\
                                        &  \eqpeerop & {\action_1 \cdots \action_{i-1} \concat
                                            \action_i \concat \acts_{i+1}}
                                                       & \text{by  case~\eqref{it:action-actprime}}                            
      \\
      {\acts_1 \concat \actsb \concat \actsb' } & \eqpeerop &    {\action_1 \cdots \action_{i-1} \concat
                                                             \action_i \concat \acts_{i+1}}
                                                       & \actsb' = \action_i
    \end{array}  
    \]  
  \end{itemize}
    In both cases, we have $\lvert \actsb \concat \actsb' \rvert \leq i$
  since $\lvert \actsb \rvert < i$ by induction hypothesis and
  $\acts = \emptyw$ (resp.\ $\actsb' = \action_i$) by
  case~\eqref{it:empty-actprime} (resp.\
  case~\eqref{it:action-actprime}).
\end{proof}

\lemexistpatheqpeer*
\begin{proof}
            We show the result by induction on the length of $\acts$.

  \proofsub{Base case} 
  If $\acts = \emptyw$, then the result holds trivially with 
  $s = s' = t = t' \in \RS_k(S)$.

  \proofsub{Inductive case}
    Assume that for all  $s \in RS_k(S)$ and $s' \in \RS_{k{+}1}(S)$ 
  such that $s \bTRANSS{\acts}{k+1} s'$, with
  $\lvert \acts \rvert < n$, 
    there is $t \in \RS_k(S)$ and $\actsb, \, \actsb' \in \ASetC$, such
  that $s \kTRANSS{\actsb} t$, $s' \bTRANSS{\actsb'}{k+1} t$,
  and $\eqpeer{\actsb}{\acts \concat \actsb'}$.

  Take $s \in RS_k(S)$ and $s' \in \RS_{k{+}1}(S)$ such that
  $s \bTRANSS{\acts}{k+1} s'$, with
    $ \acts = \action_1
  \cdots \action_n$ (i.e., $\lvert \acts \rvert= n)$, assuming that 
  \[
  s = s_1
  \bTRANSS{\action_1}{k+1} s_2 \bTRANSS{\action_2}{k+1} \cdots
  \bTRANSS{\action_n}{k+1} s_{n+1} = s'
  \]
          There are two cases depending on the direction of $\action_1$.
  \begin{enumerate}
  \item If $\action_1 = \PRECEIVE{pq}{a}$, then $s_2 \in \RS_k(S)$ since
    $s_1 \in \RS_k(S)$.
        Thus, by induction hypothesis, there is $t \in \RS_k(S)$ and
    $\actsb, \, \actsb' \in \ASetC$, such that $s_2 \kTRANSS{\actsb} t$
    and $s' \bTRANSS{\actsb'}{k+1} t$ and
        $\eqpeer{\actsb}{\action_2 \cdots \action_n \concat \actsb'}$.
        Hence,
    $\eqpeer{\action_1 \concat \actsb}{\action_1 \concat \action_2
      \cdots \action_n \concat \actsb'}$,
    as required since $s_1 \kTRANSS{\action_1} s_2$.
      \item If $\action_1 = \PSEND{pq}{a}$, then by
    Lemma~\ref{lem:k-mc-kclosed}, the set
    $\paset_1 = \{ \actsb \st s
    \kTRANSS{\actsb}\kTRANSS{\PSEND{pq}{a}} \land \p \notin \acts \}$
    is non-empty and $\kclosed{\paset_1}{s}$.

    Therefore, by Lemma~\ref{lem:kclosed-pclosed}, $\pclosed{\paset_1}{s}$
    and by Lemma~\ref{lem:new-kclosed-set}, the set
    \[
    \paset_2 = \left\{
      \acts \st
      \acts \in \paset_1 
      \land 
      \subj{\action_1}\notin \acts
    \right\} 
    \cup
    \left\{
      \acts_1 \concat \acts_2 \st
      \acts_1 \concat \action_1 \concat \acts_2 \in \paset_1 
      \land
      \subj{\action_1} \notin \acts_1
    \right\}
    \]
    is $k{+1}$-closed for $s_2$
    and
    $\paset_1 = \paset_2 = \left\{
      \acts \st
      \acts \in \paset_1 
      \land 
      \subj{\action_1}\notin \acts
    \right\}$ by definition of $\paset_1$.

    Hence, since $S$ is $\BA{k{+}1}$ by assumption, we can apply
    Lemma~\ref{lem:closed-set-paths} and obtain that
        there is $\actsb_2 \in \paset_2$ and $\hat\acts', \, \actsb_{n+1} \in \ASetC$ such
    that 
    \[
    s_2 \bTRANSS{\actsb_2}{k+1} t_2 \bTRANSS{\hat\acts'}{k+1} t_{n+1} 
        \quad \text{and} \quad
        s' = s_{n+1} \kTRANSS{\actsb_{n+1}} t_{n+1}
    \]
    for some $t_2, t_{n+1} \in \RS_{k+1}(S)$ with
        $\lvert \hat\acts \rvert < n$ and
        \[
    \eqpeer{\actsb_2 \concat \hat\acts'}{\action_2 \cdots \action_{n} \concat \actsb_{n+1}}
    \]
        
    We have $t_2 \bTRANSS{\hat\acts'}{k+1} t_{n+1}$, with
    $\lvert {\hat\acts'} \rvert < n$, 
    with $t_2 \in \RS_k(S)$,
    thus by induction hypothesis,
    there is $t \in \RS_k(S)$ such that
    $t_{n+1} \bTRANSS{\hat\actsb'}{k+1} t$,
    $s_2 \bTRANSS{\hat\actsb}{k} t$ and
    $\eqpeer{\hat\actsb}{\hat\acts' \concat \hat\actsb'}$, as pictured
    below (where {\color{red} red} parts are in $\kTRANSS{}$ and the
    rest in $\bTRANSS{}{k+1}$).
            \[
      \begin{tikzpicture}[baseline=(current bounding box.center)]
        \node (s1) {$s_2$};
        \node[right=of s1,xshift=1cm] (s2) {$s_{n+1} = s'$};
        \node[below=of s1] (s3) {$t_2$};
        \node at (s2|-s3) (s4) {$t_{n+1}$};
        \node [below=of s4] (ft) {$t$};
                \path[->] (s1) edge [above] node {$\action_2 \cdots \action_n$} (s2);
        \path[->] (s3) edge [above] node {$\hat\acts'$} (s4);
                \path[->] (s1) edge [right] node {$\actsb_2$} (s3);
        \path[->] (s2) edge [right] node {$\actsb_{n+1}$} (s4);
                \node[left=of s1] (sm1) {$s = s_{1}$};
        \node at (sm1|-s3)  (tm1) {$t_{1}$};
        \path[->] (sm1) edge [above] node {$\action_1$} (s1);
        \path[red,->] (sm1) edge [left] node {$\actsb_1 = \actsb_2$} (tm1);
        \path[red,->] (tm1) edge [above] node {$\action_1$} (s3);
                \path[red,->,dashed] (s3) edge [below] node {$\hat\actsb$} (ft);
        \path[->,dashed] (s4) edge [right] node {$\hat\actsb'$} (ft);
      \end{tikzpicture}
          \]

    We have to show that
    \[
    \eqpeer
    {\actsb_1 \concat \action_1 \concat \hat\actsb}
    {\action_1 \cdots \action_n \concat \actsb_{n+1} \concat \hat\actsb'}
    \]

    By Lemma~\ref{lem:closed-set-paths},
    \[
    \eqpeer{\actsb_1 \concat \hat\acts'}{\action_2 \cdots \action_n \concat \actsb_{n+1}}
    \]
           Prefixing each execution with $\action_1$, we have:
        \[
    \eqpeer{\action_1 \concat \actsb_1 \concat \hat\acts'}
    {\action_1 \concat \action_2 \cdots \action_n \concat \actsb_{n+1}}
    \]
    and since $\subj{\action_1} \notin \actsb_1$, we have:
    \[
    \eqpeer{\actsb_1 \concat \action_1 \concat \hat\acts'}
    {\action_1 \concat \action_2 \cdots \action_n \concat \actsb_{n+1}}
    \]
        Adding $\hat\actsb'$ on each side of the equation, we obtain:
    \[
    \eqpeer{\actsb_1 \concat \action_1 \concat \hat\acts' \concat \hat\actsb'}
    {\action_1 \concat \action_2 \cdots \action_n \concat \actsb_{n+1} \concat \hat\actsb'}
    \]
        By induction hypothesis, we have
            $
    \eqpeer{\hat\actsb}{\hat\acts' \concat \hat\actsb'}
    $.
        Hence, we obtain
    \[
    \eqpeer{\actsb_1 \concat \action_1 \concat \hat\actsb}
    {\action_1  \cdots \action_n \concat \actsb_{n+1} \concat \hat\actsb'}
    \]
    as required.
        \qedhere
  \end{enumerate}
\end{proof}

\begin{restatable}{lemma}{lemexistpathtok}\label{lem:exist-path-to-k}
  If $S$ is reduced $\OBI{k}$, $\IBI{k{+}1}$, and $k$-exhaustive, then
  for all $s \in \RS_{k+1}(S)$, there is $t \in \RS_k(S)$ such that
  $s \bTRANSR{k+1} t$.
  \end{restatable}
\begin{proof}
  Direct consequence of Lemma~\ref{lem:exist-path-eqpeer}.
\end{proof}

\begin{restatable}{lemma}{lemkexhaukplusexhau}\label{lem:k-exhau-kplus-exhau}
  If $S$ is reduced $\OBI{k}$, $\IBI{(k{+}1)}$, and $k$-exhaustive,
  then it is reduced $\OBI{(k{+}1)}$ and $(k{+}1)$-exhaustive.
\end{restatable}
\begin{proof}

  \proofsub{$\OBI{(k{+}1)}$} By contradiction, assume $S$ is reduced
  $\OBI{k}$ but not reduced {$\OBI{(k{+}1)}$}.
    Then, there must be
  $s = \csconf q w \in \RTS{k+1}{S} \setminus \kRTS{S}$ such that there
  is $\p \in \PSet$, 
  $s \bTRANSS{\PSEND{pr}{b}}{k{+}1}$,
  $(q_\p, \PSEND{pr}{b}, q'_\p) \in \delta_\p$, and 
  $\neg(s\kTRANSS{\PSEND{pt}{c}})$. 
    By Lemma~\ref{lem:exist-path-to-k} and
  Lemma~\ref{lem:por-trace-equiv} (2), there is $t' \in \RTS{k}{S}$ such
  that $s \rbTRANSS{\acts}{k+1} t'$.
    There are two cases:
  \begin{enumerate}
  \item If $\PRECEIVE{pt}{x} \notin \actsb_1$, then we have
    $t' \bTRANSS{\PSEND{pr}{b}}{k+1}$, and
    $\neg(t' \bTRANSS{\PSEND{pt}{c}}{k+1})$, hence
    $\neg(t' \bTRANSS{\PSEND{pt}{c}}{k})$. 
        \begin{itemize}
    \item If $t' \bTRANSS{\PSEND{pr}{b}}{k}$ we have a contradiction
      with the fact that $S$ is reduced $\OBI{k}$.
          \item If $\neg(t' \bTRANSS{\PSEND{pr}{b}}{k})$ then both queues are
      full at $t'$.
            Since $S$ is $k$-exhaustive, both actions are enabled along a
      $k$-bounded execution from $t'$. 
            However, one action must be enabled before the other, in any
      execution, contradicting the fact that $S$ is reduced $\OBI{k}$.
    \end{itemize}
  \item If $\PRECEIVE{pt}{x} \in \actsb_1$,
    $t' \bTRANSS{\PSEND{pr}{b}}{k+1}$, and
    $t' \bTRANSS{\PSEND{pt}{c}}{k+1}$. Then the queue $\ptp{pt}$ must
    still be holding $k$ messages at $t'$.
        Hence, $\neg(t' \bTRANSS{\PSEND{pt}{c}}{k})$ and we reason as
    above to reach a contradiction with the fact that $S$ is reduced
    $\OBI{k}$.
  \end{enumerate}
  
  \proofsub{$(k{+}1)$-exhaustive} By
  contradiction, assume $S$ is $k$-exhaustive, but not
  $(k{+}1)$-exhaustive.  Then, there must be
  $s = \csconf q w \in \RS_{k+1} \setminus \RS_k(S)$ such that there
  is $\p \in \PSet$,
    with $q_\p$ a sending state and the following does \emph{not} hold:
  \begin{equation}\label{eq:kpone-snd}
    \forall (q_\p, \PSEND{pq}{a}, q'_\p) \in \delta_\p \qst \exists
    \acts \in \ASetC
    \qst
    s  \bTRANSS{\acts}{k+1}\bTRANSS{\PSEND{pq}{a}}{k+1}  \text{ and } \p \notin
    \acts
  \end{equation}
      By Lemma~\ref{lem:exist-path-to-k}, there is $s' \in \RS_k(S)$
  such that $s \bTRANSS{\acts}{k+1} s'$.

  \begin{enumerate}
  \item   If $\p \notin \acts$, then
    $s' \kTRANSS{\acts'}\kTRANSS{\PSEND{pq}{a}}$ (by $k$\MC\ and
    $s' \in \RS_k(S)$), i.e., a contradiction.
  \item If $\p \in \acts$. There are two cases:
    \begin{enumerate}
    \item $\acts = \acts_1 \concat \PSEND{pq}{a} \concat \acts_2$ with
      $\p \notin \acts_1$, hence $\PSEND{pq}{a}$ can be fired from $s$,
      a contradiction with the assumption that~\eqref{eq:kpone-snd}
      above does not hold.
    \item $\acts = \acts_1 \concat \PSEND{pt}{b} \concat \acts_2$ with
      $\p \notin \acts_1$ and $\msg{a} \neq \msg{b}$.
                  This implies that 
      \[
      s \bTRANSS{\acts_1}{k+1} \bTRANSS{\PSEND{pq}{a}}{k+1}
      \text{ since $S$ is $\infBA$}
      \]
      which contradicts the assumption
      that~\eqref{eq:kpone-snd} does not hold.
                  \qedhere
    \end{enumerate}
      \end{enumerate}
    \end{proof}

\begin{restatable}{lemma}{lemexistpatheqpeergeneral}\label{lem:exist-path-eqpeer-general}
    If $S$ is (reduced) $\OBI{k}$, $\infIBI$, and $k$-exhaustive, then for all
  $s \in \RS(S)$ such that $s_0 \TRANSS{\acts} s$, there is
  $t \in \RS_k(S)$ and $\actsb, \, \acts' \in \ASetC$, such that
  $s_0 \kTRANSS{\actsb} t$, $s \TRANSS{\acts'} t$, and
  ${\actsb} \eqpeerop {\acts \concat \acts'}$.
  \end{restatable} 
\begin{proof}
  We first note that in this case $\resche$ and $\eqpeerop$ coincide
  since we only consider executions starting from $s_0$, see
  Lemma~\ref{lem:valid-word}; thus we show that
  ${\actsb} \resche {\acts \concat \acts'}$.

      From Lemma~~\ref{lem:k-exhau-kplus-exhau}, we know that $S$ is
  $n$-exhaustive (for any $n \geq k$).
    Hence, we obtain the result by repeated applications of
  Lemma~\ref{lem:exist-path-eqpeer} (with $s = s_0$) using the fact
  that $\resche$ is a congruence.
\end{proof}

\begin{restatable}{lemma}{lemkmckplusmc}\label{lem:k-mc-kplus-mc}
  If $S$ is $\OBI{k}$, $\infIBI$, and $k$\MC, then it is
  $\OBI{k{+}1}$ and $(k{+}1)$\MC.
\end{restatable}
\begin{proof}
  By Lemma~\ref{lem:reduced-k-mc-kplus-mc} and
  Lemma~\ref{lem:obi-imp-reduced-obi}.
\end{proof}

\begin{restatable}{lemma}{lemreducedkmckplusmc}\label{lem:reduced-k-mc-kplus-mc}
  If $S$ is reduced $\OBI{k}$, $\infIBI$, and $k$\MC\, then it is
  $\OBI{k{+}1}$ and $(k{+}1)$\MC.
\end{restatable}
\begin{proof}
  Assume by contradiction, that $S$ is $k$\MC, but not $(k{+}1)$-safe.
  Then, there must be $s  = \csconf q w  \in \RS_{k+1} \setminus
  \RS_k(S)$ such that at least one of the following conditions does
  \emph{not} hold.

  \begin{enumerate}
  \item\label{ite:kpone-rcv} For all $\p\q \in \CSet$, if $w_{\p\q} = \msg{a} \concat w'$, then $s \bTRANSR{k+1}
    \bTRANSS{\PRECEIVE{pq}{a}}{k+1}$.
  \item\label{ite:kpone-progr} For all $\p \in \PSet$, if $q_\p$ is a \emph{receiving} state,
    then $s \bTRANSR{k+1} \bTRANSS{\PRECEIVE{qp}{a}}{k+1}$ for some $\q \in
    \PSet$ and $\msg{a} \in \ASigma$.
  \end{enumerate}

  Note that $S$ is $\OBI{(k{+}1)}$ and $k{+}1$-exhaustive by
  Lemma~\ref{lem:k-exhau-kplus-exhau}.
 
  By Lemma~\ref{lem:exist-path-to-k}, there is $s' \in \RS_k(S)$
  such that $s \bTRANSS{\acts}{k+1} s'$.

  \proofsub{1}
  Assume that Item~\ref{ite:kpone-rcv} above does not hold, i.e.,
  we have $w_{\p\q} = \msg{a} \concat w'$ for some $\p\q \in \CSet$, but
  each path $\acts$ from $s$ does not contain $\PRECEIVE{pq}{a}$.
    Observe that for the first occurrence of $\PRECEIVE{pq}{b}$ in
  $\acts$, we must have $\msg{a} = \msg{b}$ (since
  $w_{\p\q} = \msg{a} \concat w'$), but we cannot have
  $\PRECEIVE{pq}{a} \in \acts$ by contradiction hypothesis.
    This implies that we have $w'_{\p\q} = \msg{a} \concat w' \concat
  w''$ in $s'$, and since $S$ is $k$\MC\ and $s' \in \RS_k(S)$, we must
  have $s' \kTRANSR\kTRANSS{\PRECEIVE{pq}{a}} $.
    Thus, we have $s \bTRANSS{\acts}{k+1} s'
  \kTRANSR\kTRANSS{\PRECEIVE{pq}{a}}$, a contradiction.
    
  \proofsub{2} Assume that Item~\ref{ite:kpone-progr} above does not
  hold, i.e., there is $\p \in \PSet$ such that $q_\p$ is a receiving
  state but for each path $\acts$ from $s$, $\acts$ does not allow
  $\p$ to fire a (receive) action.
    Hence, by contradiction hypothesis we have $\PRECEIVE{qp}{a} \notin
  \acts$ for any $\msg{a}$ and $\q$.
    Hence $\p$ is still in state $q_\q$ in configuration $s'$.  
    Since $S$ is $k$\MC\ and $s' \in \RS_k(S)$, we must
  have $s' \kTRANSR\kTRANSS{\PRECEIVE{qp}{a}} $.
    Thus, we have $s \bTRANSS{\acts}{k+1} s'
  \kTRANSR\kTRANSS{\PRECEIVE{qp}{a}}$, a contradiction.
    \end{proof}

\thmsoundness* 
\begin{proof}
  By Theorem~\ref{thm:reduced-soundness} and
  Lemma~\ref{lem:obi-imp-reduced-obi}.
\end{proof}

\thmreducedsoundness*
\begin{proof}
  Direct consequence of Lemma~\ref{lem:reduced-k-mc-kplus-mc}.
\end{proof}

\begin{restatable}{lemma}{lemsafeksafe}\label{thm:safe-kexh-imp-ksafe}
    Let $S$ be (reduced) $\OBI{k}$ and $\infIBI$.
    If $S$ is safe and $k$-exhaustive, then it is $k$\MC.
\end{restatable} 
\begin{proof}
  We show that $S$ is $k$-safe.
    By contradiction, assume there is $S$ safe, $k$-exhaustive, and
  \emph{not} $k$-safe.
    Since $S$ is not $k$-safe, then there is
  $s = \csconf{q}{w} \in \RS_k(S)$ such that at least one of the two
  cases below hold.
  \begin{enumerate}
  \item $w_{\ptp{pq}} = \msg{a} \concat \word$ and there is no
    execution $\acts$ such that $s \kTRANSS{\acts} \kTRANSS{\PRECEIVE{pq}{a}}$.
        By safety, there is $\actsb$ and $n > k$ such that
    $s \bTRANSS{\actsb}{n} s' \bTRANSS{\PRECEIVE{pq}{a}}{n} s''$.
        By Lemma~\ref{lem:exist-path-eqpeer}, we can extend
    $\actsb \concat \PRECEIVE{pq}{a}$ such that there is an equivalent
    $k$-bounded execution, which contradicts this case.
      \item $q_\p$ is a receiving state and there is no execution $\acts$
    such that $s \kTRANSS{\acts} \kTRANSS{\PRECEIVE{qp}{a}}$; then we
    reason similarly as above using
    Lemma~\ref{lem:exist-path-eqpeer}.
        \qedhere
  \end{enumerate}
\end{proof}

\begin{lemma}\label{lem:kri-imp-kba}
  If $S$ is $\DIBI{k}$, then it is $\IBI{k}$.
\end{lemma}
\begin{proof}
  Straightforward.
\end{proof}

\begin{lemma}\label{lem:kri-imp-kba-depends}
  If $S$ is $\CIBI{k}$, then it is $\IBI{k}$.
\end{lemma}
\begin{proof}
  Straightforward.
\end{proof}

\begin{lemma}\label{lem:dep-chain}
  If
    $\dependsin{s}{\action}{\acts}{\action'}$, then there is
  a subsequence $\actsb$ of $\acts$ such that
  \begin{itemize}
  \item  $\bindep{s}{\action}{\action'}$ and $\actsb = \emptyw$, or
  \item $\actsb = \action_1 \cdots \action_n$ ($n \geq 1$),
    $\bindep{s}{\action}{\action_1}$,
    $\forall 1 \leq i < n \qst \bindep{s}{\action_i}{\action_{i+1}}$,
    $\bindep{s}{\action_n}{\action'}$.
          \end{itemize}
  \end{lemma}
\begin{proof}
  By induction on the length of $\acts$.

  \proofsub{Base case}
    If $\dependsin{s}{\action}{\emptyw}{\action'}$, then we must have
  $\bindep{s}{\action}{\action'}$ by definition.

  \proofsub{Inductive case}
    Assume the result holds for $\acts$ and let us show it holds for
  $\action'' \concat \acts$.
    There are two cases:
  \begin{itemize}
  \item If $\dependsin{s}{\action}{\acts}{\action'}$ and we have the
    result by induction hypothesis, since any subsequence of
    $\acts$ is a subsequence of $\action'' \concat \acts$. 
  \item If 
    $\bindep{s}{\action}{\action''}$ and
    $\dependsin{s}{\action''}{\acts}{\action'}$.
        Then by induction hypothesis there is a subsequence
    $\action_1 \cdots \action_n$ of $\acts$ such that
    $
    \action'' \prec \action_1 \prec \cdots \action_n \prec \action'
    $
    hence we have the result with the subsequence
    $\action'' \concat \action_1 \cdots \action_n$.
        \qedhere
  \end{itemize}
\end{proof}

\begin{lemma}\label{lem:dep-chain-equi}
  Let $S$ be a system, $s \in \RS(S)$, and
  $\acts = \acts_1 \concat \action \concat \acts_2 \concat \action'
  \concat \acts_3$
  such that $s_0 \TRANSS{\acts}$ and
  $\dependsin{s}{\action}{\acts_2}{\action'}$, with
  $s_0 \TRANSS{\acts_1} s$.
    Then for all valid $\actsb$ such that $\actsb \resche \acts$, there
  are $\actsb_1$, $\actsb_2$, $\actsb_3$, and $t \in \RS(S)$ such that

  \begin{enumerate}
  \item  $\actsb = \actsb_1 \concat \action \concat \actsb_2 \concat \action'
    \concat \actsb_3$,
      \item $\proj{\actsb_1}{\subj{\action}} = \proj{\acts_1}{\subj{\action}}$
      \item $\proj{\actsb_1 \concat \action \concat \actsb_2
    }{\subj{\action'}} = \proj{\acts_1 \concat \action \concat \acts_2
    }{\subj{\action'}}$, and
      \item $\dependsin{t}{\action}{\actsb_2}{\action'}$, with
  $s_0 \TRANSS{\actsb_1} t$.
  \end{enumerate}
  \end{lemma}
\begin{proof}
        By Lemma~\ref{lem:dep-chain}, there is a subsequence $\action_1 \cdots \action_n$ of $\acts_2$
  such that
  \[
  \bindep{s}{\action= \action_0}{\action_1}
  \text{ and }
  \forall 1 \leq i < n \qst \bindep{s}{\action_i}{\action_{i+1}}
  \text{ and }
  \bindep{s}{\action_n}{\action' = \action_{n+1}}
  \]
    Take the shortest such subsequence (smallest $n$), we show that the
  relative order between each pair of actions must be preserved.
    By definition, for each $\bindep{s}{\action_j}{\action_{j+1}}$
  ($0 \leq j \leq n+1 $) to hold there are two cases:
  \begin{itemize}
  \item If $\subj{\action_j} = \subj{\action_{j+1}}$, then it is not
    possible to swap $\action_j$ and $\action_{j+1}$ while preserving
    $\resche$-equivalence.
      \item If $\subj{\action_j} \neq \subj{\action_{j+1}}$, then
    $\chan{\action_j} = \chan{\action_{j+1}}$, and there are two
    cases depending on whether the queue $\chan{\action_j}$ is empty
    when $\action_j$ is fired.
            \begin{itemize}
    \item If the queue is empty, then we cannot swap $\action_j$ and
      $\action_{j+1}$ without invalidating the execution since they
      are matching send and receive actions.
          \item If the queue is not empty, since
      $w_{\chan{\action_{j}}} = \emptyw$ (at $s$) there must be
      another \emph{send} action $\action_l$ with $l < j$ such that
      $\chan{\action_l} = \chan{\action_{j+1}}$.
            Therefore, we have $\bindep{s}{\action_l}{\action_{j+1}}$,
      and thus
      $\action_1 \cdots \action_l \cdots \action_{j+1} \cdots \action_n$
      is a (striclty) shorter subsequence of
            $\acts_2$ which is dependency chain, a contradiction.
    \end{itemize}   
  \end{itemize}
    Since each pair of actions cannot be swapped without invalidating
  the sequence or break $\resche$-equivalence, we must conclude that
  any $\actsb$ has the required form and that the
  $\dependsin{t}{\action}{\actsb_2}{\action'}$ property holds since
  $\actsb_2$ must contain the subsequence
  $\action_1 \cdots \action_n$.
    \end{proof}

\begin{lemma}\label{lem:kri-kexh-imp-kba}
  If $S$ is reduced $\OBI{k}$, $\DIBI{k}$ and $k$-exhaustive, then it is $\IBI{k{+}1}$.
\end{lemma}
\begin{proof}
  From Lemma~\ref{lem:sibi-imp-cibi}
  and~\ref{lem:kri-kexh-imp-kba-depends}.
\end{proof}

\begin{lemma}\label{lem:sibi-imp-cibi}
  If $S$ is $\DIBI{k}$, then it is $\CIBI{k}$.
\end{lemma}
\begin{proof}
  By contradiction, take $s = \csconf{q}{w} \in \RS_k(S)$ such that
  the condition for $\CIBI{k}$ do not hold while the condition for
  $\DIBI{k}$ does.
    Then, we must have $s \kTRANSS{\PRECEIVE{qp}{a}} s'$ and
  $s' \kTRANSS{\acts} \kTRANSS{\PSEND{sp}{b}}$ such that
  $\neg(\dependsin{s}{\PRECEIVE{qp}{a}}{\acts}{\PSEND{sp}{b}})$. 
    However, the existence of an execution
  $s' \kTRANSS{\acts} \kTRANSS{\PSEND{sp}{b}}$ contradicts
  Definition~\ref{def:non-csa-mc}.
\end{proof} 

\begin{lemma}\label{lem:kri-kexh-imp-kba-depends}
  If $S$ is reduced $\OBI{k}$, $\CIBI{k}$ and $k$-exhaustive, then it is $\IBI{k{+}1}$.
\end{lemma}
\begin{proof}
  Take $s \in \RS_{k}(S)$ and $s' \in \RS_{k+1}(S)$
  such that $s \bTRANSS{\acts}{k+1} s'$.
    We show by induction on the length of $\acts$ that
  $s' \bTRANSS{\acts'}{k+1} t'$ for some $t' \in \RS_k(S)$, and there
  is $\actsb$ such that $s \kTRANSS{\actsb} t'$ with
  $\eqpeer{\actsb}{\acts \concat \acts'}$, and for all prefix
  $\acts'_0$ of $\acts'$, if
  $s' \bTRANSS{\acts'_0}{k+1} s'' = \csconf{q}{w}$, $s''$ validates the
  following condition, for all $\p \in \PSet$:
                  \[ s'' \bTRANSS{\PRECEIVE{qp}{a}}{k+1} t \implies
  \forall \action \in \ASet \qst s \bTRANSS{\action}{k+1} \, \land \,
  \subj{\action} = \p \implies \action = \PRECEIVE{qp}{a}
  \]

  \proofsub{Base case}
    Assume $\acts = \action$.
    If $\action = \PRECEIVE{pq}{a}$, then $s' \in \RS_{k}(S)$, and we
  have result since $S$ is $\CIBI{k}$ (via Lemma~\ref{lem:kri-imp-kba-depends}),
  with $s' = t'$. 
    If $\action = \PSEND{pq}{a}$, then since $S$ is $k$-exhaustive, we
  have $s \kTRANSS{\actsb} t \kTRANSS{\PSEND{pq}{a}} t'$, with
  $\p \notin \actsb$.
    Hence, we have $s' \bTRANSS{\actsb}{k+1} t'$.
    We show that for all prefix $\actsb_0$ of $\actsb$, if
  $s' \bTRANSS{\actsb_0}{k+1} t''$, then $t''$ validates the
  $\IBI{k{+}1}$ condition.  
    We have the following situation:
  \[
  \begin{tikzpicture}[baseline=(current bounding box.center)]
        \node (s0) {$s$};
    \node[right=of s0, xshift=1cm] (s) {$s'$};

        \node[below=of s0] (sp) {$s''$};
    \node[below=of sp] (spp) {$t$};
    \node at (s|-sp) (t) {$t''$};
    \node at (t|-spp)  (tp) {$t'$};
        \path[->] (s0) edge [above] node {$\action = \PSEND{pq}{a}$} (s);
    \path[->] (s) edge [right] node {$\actsb_0$} (t);
    \path[->] (s0) edge [left] node {$\actsb_0$} (sp);
    \path[->] (sp) edge [above] node {$\action$} (t);
    \path[->] (spp) edge [above] node {$\action$} (tp);
    \path[->] (t) edge [right] node {$\actsb_1$} (tp);
    \path[->] (sp) edge [left] node {$\actsb_1$} (spp);
  \end{tikzpicture}
    \]
                                                                                      
        Assume by contradiction that $t'' \bTRANSS{\PRECEIVE{sr}{b}}{k+1}$ and
  $t'' \bTRANSS{\PRECEIVE{tr}{c}}{k+1}$.
    If these two transitions are also enabled at $s''$, we have a
  contradiction with the fact that $S$ is $\CIBI{k}$.
    Hence, we have that either participant $\rr$ has made a move through
  $\action$, hence $\p = \ptp{r}$, an additional receive action in
  $\rr$ becomes enabled because $\ptp{sr} = \ptp{pq}$, or
  $\ptp{tr} = \ptp{pq}$ (i.e., the queue $\ptp{sr}$ (resp.\
  $\ptp{tr}$) is empty in $s$ and $s''$).
      \begin{itemize}
  \item If $\p = \ptp{r}$, then if we pose $\actsb_0 = \actsb$, we
    have $t' \bTRANSS{\PRECEIVE{sr}{b}}{k+1}$ and
    $t' \bTRANSS{\PRECEIVE{tr}{c}}{k+1}$, which contradicts the fact
    that $S$ is $\CIBI{k}$.
  \item If $\ptp{sr} = \ptp{pq}$ (i.e.,
    $\PRECEIVE{sr}{b} = \PRECEIVE{pq}{a}$), then we have
    $s'' \kTRANSS{\PRECEIVE{tq}{c}} v$ for some $v$.
        Since $S$ is $k$-exhaustive, we also have
    $v \kTRANSS{\actsb_2} \kTRANSS{\PSEND{pq}{a}}$ with
    $\p \notin \actsb_2$.
        By $\CIBI{k}$, we have that for all such $\actsb_2$, we have
    $ \dependsin{s''}{\PRECEIVE{tq}{c}}{\actsb_2}{\PSEND{pq}{a}}$,
    which is a contradiction with Lemma~\ref{lem:dep-chain-equi} since
    the two actions are swapped in $k{+}1$. 
                                                                                              \item The case $\ptp{tr} = \ptp{pq}$ is symmetric to the one above.
  \end{itemize}

  \proofsub{Inductive case} 
    Assume the result holds for $\acts$ and let us show it holds for
  $\acts \concat \action$.
    Assume that we have the following situation, where the dashed edges
  need to be shown to exist.
    \[
  \begin{tikzpicture}[baseline=(current bounding box.center)]
        \node (s) {$s$};    
    \node[right=of s, xshift=1cm] (sp) {$s'$};
    \node[right=of sp, xshift=1cm] (spp) {$s''$};
        \node[below=of sp] (tp) {$t'$};
    \node[below=of spp] (tpp) {$t''$};
            \path[->] (s) edge [above] node {$\acts$} (sp);
    \path[->] (sp) edge [above] node {$\action$} (spp);
    \path[->] (s) edge [bend right, above] node {$\actsb$} (tp);
    \path[->] (sp) edge [left] node {$\acts'$} (tp);
    \path[->,dashed] (spp) edge [left] node {} (tpp);
    \path[->,dashed] (tp) edge [left] node {} (tpp);
  \end{tikzpicture}
    \]
  with $s, t' \in \RS_k(S)$ and $s', s'' \in \RS_{k+1}(S)$.

  By induction hypothesis, all configurations between $s'$ and $t'$ and
  between $s'$ and $s''$ are $\IBI{k{+}1}$ and $\OBI{k{+}1}$ (by
  Lemma~\ref{lem:reduced-k-mc-kplus-mc}), hence, we can use a similar
  reasoning to that of Lemma~\ref{lem:new-kclosed-set} to show that
  either $s'' \bTRANSS{\acts'}{k+1} t''$ (with
  $t' \bTRANSS{\action}{k+1} t''$) or $s'' \bTRANSS{\acts'}{k+1} t' $
  (with $t' = t''$).
    \begin{itemize}
  \item If $s'' \bTRANSS{\acts'}{k+1} t''$ (with
    $t' \bTRANSS{\action}{k+1} t''$), then we proceed as in the base
    case with $s \coloneqq t'$ and $s' \coloneqq t''$.
      \item If $s'' \bTRANSS{\acts'}{k+1} t' $ (with $t' = t''$), then we
    only have to show that all configurations on $\acts'$ validate the
    condition. Since there is an equivalent $k$-bounded execution, any
    violation would contradict the hypothesis that $S$ is $\CIBI{k}$.
        \qedhere
  \end{itemize}
  \end{proof}

\lemkrikexhimpkripreduced*
\begin{proof}
    We note that since $S$ is reduced $\OBI{k}$, $\DIBI{k}$ and
  $k$-exhaustive, we have that $S$ is $\IBI{k{+}1}$ by
  Lemma~\ref{lem:kri-kexh-imp-kba}.
      We show this result by contradiction, using
  Lemma~\ref{lem:kri-kexh-imp-kba} and
  Lemma~\ref{lem:exist-path-eqpeer}.
    Assume, by contradiction, that there is $s \in RS_k(S)$ and
  $s'= \csconf{q}{w} \in \RS_{k{+}1}(S)$ such that
  $s \bTRANSS{\acts}{k+1} s'$ with $\p \in \PSet$ s.t.\
  \begin{enumerate}
  \item \label{en:lem:kri-kexh-imp-krip-rcv}
    $s' \bTRANSS{\PRECEIVE{qp}{a}}{k+1}$, and
    $s' \bTRANSS{\PRECEIVE{sp}{b}}{k+1}$, or
      \item \label{en:lem:kri-kexh-imp-krip-sndlater}
    $s' \bTRANSS{\PRECEIVE{qp}{a}}{k+1}$, and
    $ \exists (q_\p, \PRECEIVE{sp}{b}, q'_\p) \in \delta_\p \qst \s
    \neq \q \land
    s \bTRANSR{k+1} \bTRANSS{\PSEND{sp}{b}}{k+1}$
      \end{enumerate}

  \noindent
  \eqref{en:lem:kri-kexh-imp-krip-rcv}  follows from Lemma~\ref{lem:kri-kexh-imp-kba}.

  \noindent
  \eqref{en:lem:kri-kexh-imp-krip-sndlater}
    Assume there is $s'$ such that $s' \bTRANSS{\PRECEIVE{qp}{a}}{k+1}$, and
  $ \exists (q_\p, \PRECEIVE{sp}{b}, q'_\p) \in \delta_\p \qst \s \neq
  \q \land s \bTRANSS{\acts'}{k+1} \bTRANSS{\PSEND{sp}{b}}{k+1} s''$.
    By Lemma~\ref{lem:exist-path-eqpeer}, there is $t \in \RS_k(S)$ such
  that $s \kTRANSS{\actsb} t$ and $s'' \bTRANSS{\acts'}{k+1} t$ with
  $\eqpeer{\actsb}{\acts \concat \acts \concat \PSEND{sp}{b} \concat
    \acts''}$.
    Hence both $\PRECEIVE{qp}{a}$ and $\PSEND{sp}{b}$ appear in $\actsb$
  which contradicts the fact that $S$ is $\DIBI{k}$.
                \end{proof}

\begin{restatable}{lemma}{lemkrikexhimpinfrip}\label{lem:kri-kexh-imp-infrip}
  If $S$ is $\OBI{k}$, $\DIBI{k}$ and $k$-exhaustive, then it is $\infIBI$.
\end{restatable}
\begin{proof}
  Direct consequence of Lemma~\ref{lem:kri-kexh-imp-krip-reduced},
  Lemma~\ref{lem:k-exhau-kplus-exhau}, and
  Lemma~\ref{lem:obi-imp-reduced-obi}.
\end{proof}

\begin{restatable}{lemma}{lemkrikexhimpcibipreduced}\label{lem:kri-kexh-imp-cibi-reduced}
  If $S$ is reduced $\OBI{k}$, $\CIBI{k}$, and $k$-exhaustive, then it is $\CIBI{k{+}1}$.
\end{restatable}
\begin{proof}
  We first note that since $S$ is reduced $\OBI{k}$, $\CIBI{k}$ and
  $k$-exhaustive, we have that $S$ is $\IBI{k{+}1}$ by
  Lemma~\ref{lem:kri-kexh-imp-kba-depends}.

  We show this result by contradiction, using
  Lemma~\ref{lem:kri-kexh-imp-kba-depends} and
  Lemma~\ref{lem:exist-path-eqpeer}.
    Assume, by contradiction, that there is $s \in RS_k(S)$ and
  $s'= \csconf{q}{w} \in \RS_{k{+}1}(S)$ such that
  $s \bTRANSS{\acts}{k+1} s'$ with $\p \in \PSet$,
  $(q_\p, \PRECEIVE{sp}{b}, q'_\p) \in \delta_\p$ and $\s \neq \q$
  s.t.\
  \begin{enumerate}
  \item \label{en:lem:kri-kexh-imp-krip-rcv-cibi}
    $s' \bTRANSS{\PRECEIVE{qp}{a}}{k+1}$, and
    $s' \bTRANSS{\PRECEIVE{sp}{b}}{k+1}$, or
      \item \label{en:lem:kri-kexh-imp-krip-sndlater-cibi}
    $s' \bTRANSS{\PRECEIVE{qp}{a}}{k+1}s''$,
        $s'' \bTRANSS{\acts'}{k+1} \bTRANSS{\PSEND{sp}{b}}{k+1} t$, and
    $\neg (\dependsin{s'}{\PRECEIVE{qp}{a}}{\acts'}{\PSEND{sp}{b}})$
          \end{enumerate}

  \noindent
  \eqref{en:lem:kri-kexh-imp-krip-rcv-cibi} is a contradiction with
  Lemma~\ref{lem:kri-kexh-imp-kba-depends}.
  
  \noindent 
  \eqref{en:lem:kri-kexh-imp-krip-sndlater}
      By Lemma~\ref{lem:exist-path-eqpeer}, there is $t' \in \RS_k(S)$
  such that $s \kTRANSS{\actsb} t'$ and $t \bTRANSS{\acts''}{k+1} t'$
  with
  $\eqpeer{\actsb}{\acts \concat \PRECEIVE{qp}{a} \concat \acts'
    \concat \PSEND{sp}{b} \concat \acts''}$.
        There are two cases:
  \begin{enumerate}
  \item \label{en:kri-kexh-imp-cibi-reduced-snd-first}
 If
    $\actsb = \actsb_1 \concat \PSEND{sp}{b} \concat \actsb_2 \concat
    \PRECEIVE{qp}{a} \concat \actsb_3$,
    with 
    $\proj{\actsb_1 \concat \PSEND{sp}{b} \concat \actsb_2}{p}
    = \proj{\acts}{p}
    $
    and
    $\proj{\actsb_1}{s} = 
    \proj{\acts \concat \PRECEIVE{qp}{a} \concat \acts'}{s}$,
        then we have a contradiction with the assumption that $S$ is
    $\CIBI{k}$ since $\p$ can receive  $\msg{b}$ and $\msg{a}$
    after having executed
    $\proj{\actsb_1 \concat \PSEND{sp}{b} \concat \actsb_2}{p}$, i.e.,
    both messages are in the queue.
      \item If
    $\actsb = \actsb_1 \concat \PRECEIVE{qp}{a} \concat \actsb_2
    \concat \PSEND{sp}{b} \concat \actsb_3$,
        with 
    $\proj{\actsb_1}{p} = \proj{\acts}{p}$
    and
    $\proj{\actsb_1 \concat \PRECEIVE{qp}{a} \concat \actsb_2}{s}
    =
    \proj{\acts \concat \PRECEIVE{qp}{a} \concat \acts'}{s}$,
        then we must have
    $\dependsin{\hat{s}}{\PRECEIVE{qp}{a}}{\actsb_2}{\PSEND{sp}{b}}$
    (assuming $s_0 \TRANSS{\actsb_1} \hat{s}$) since $S$ is
    $\CIBI{k}$.
        By Lemma~\ref{lem:dep-chain-equi}, we must also have
    $\dependsin{{s'}}{\PRECEIVE{qp}{a}}{\acts'}{\PSEND{sp}{b}}$, a
    contradiction.
            \qedhere 
  \end{enumerate}
\end{proof}

\begin{restatable}{lemma}{lemkrikexhimpinfripcibi}\label{lem:kri-kexh-imp-infrip-cibi}
  If $S$ is $\OBI{k}$, $\CIBI{k}$ and $k$-exhaustive, then it is $\infIBI$.
\end{restatable}
\begin{proof}
  By Lemma~\ref{lem:k-exhau-kplus-exhau},
  Lemma~\ref{lem:kri-kexh-imp-cibi-reduced}, and Lemma~\ref{lem:obi-imp-reduced-obi}.
\end{proof}

\lemkrikexhimpinfripBOTH*
\begin{proof}
  By Lemma~\ref{lem:kri-kexh-imp-infrip-cibi} and Lemma~\ref{lem:kri-kexh-imp-infrip}.
\end{proof}

 \subsection{Proofs for Section~\ref{sec:new-completeness} (local-bound agnosticity)}

\begin{restatable}{lemma}{lemkexhimpweakbisimsys}\label{lem:kexh-imp-weakbisim-sys}
  If $S$ is (reduced) $\OBI{k}$, $\infIBI$, and $k$-exhaustive, then
  
  $\forall \p \in \PSet \qst \epsproj{\kTS{S}}{p} \wbisim
  \epsproj{\TS{k+1}{S}}{p}$.
\end{restatable}
\begin{proof}
  Pose $\kTS{S} = (N, s_0, \Delta)$ and $\TS{k{+}1}{S} = (N', s_0,
  \Delta')$.
    Recall that we have $ \Delta \subseteq \Delta'$ and $N \subseteq N'$.  $\neg(\epsproj{\kTS{S}}{p} \wbisim \epsproj{\TS{k+1}{S}}{p})$ for
  some $\p \in \PSet$.
    Then, there are $s \in N \cap N'$ and $\action$ (with $\subj{\action} = \p$) such that
    $     s \bTRANSS{\acts}{k{+}1} s'
    \bTRANSS{\action}{k{+}1} s''$ with $\proj{\acts}{p} = \emptyw$
      and
    \begin{equation}\label{eq:wbisim-hyp}
    \forall \acts' \in \ASet \qst \forall s'' \in \RS_k(S) \qst 
    s \kTRANSS{\acts'} s'' 
    \land \proj{\acts'}{p} = \emptyw
    \implies \neg (s'' \kTRANSS{\action})
  \end{equation}

  By Lemma~\ref{lem:exist-path-eqpeer}, there is 
      there is $t \in \RS_k(S)$ and $\actsb, \, \actsb' \in \ASetC$, such
  that $s \kTRANSS{\actsb} t$, $s'' \bTRANSS{\actsb'}{k+1} t$,
  $\eqpeer{\actsb}{\acts \concat \action \concat \actsb'}$.
    Hence, we have $s \kTRANSS{\actsb}$ with
  $\proj{\actsb}{p} = \action \concat \actsb''$ for some $\actsb''$
  with contradicts~\eqref{eq:wbisim-hyp}.
      \end{proof}

\begin{restatable}{lemma}{lemcompleteness}\label{lem:completeness}
  If $S$ is such that
  $\exists k \in \naturals_{>0} \qst \forall \p
  \in \PSet \qst \epsproj{\kTS{S}}{p} \wbisim
  \epsproj{\TS{k{+}1}{S}}{p}$, then $S$ is $k$-exhaustive.
  \end{restatable}
\begin{proof}
    Assume by contradiction that there is some $k \in \naturals_{>0}$
  such that
  \begin{equation}\label{eq:wbisim}
    \forall \p \in \PSet \qst \epsproj{\kTS{S}}{p} \wbisim
    \epsproj{\TS{k{+}1}{S}}{p}
  \end{equation}
  and $S$ is \emph{not} $k$-exhaustive.

  Pose $\kTS{S} = (N, s_0, \Delta)$ and
  $\TS{k{+}1}{S} = (N', s_0, \Delta')$.
    Recall that we have $ \Delta \subseteq \Delta'$ and
  $N \subseteq  N'$.
  
  Since $S$ is \emph{not} $k$-exhaustive, there are
  $s = \csconf{q}{w} \in \RS_k(S)$ and $\ptp{pq} \in \CSet$ such that
  $s \TRANSS{\PSEND{pq}{a}}$ and
    \begin{equation}\label{eq:nosend}
        \forall \acts \in \ASetC \qst \forall s' \in \RS_k(S) \qst s
    \kTRANSS{\acts} s' \land \p \notin \acts \implies \neg (s' \kTRANSS{\PSEND{pq}{a}})
  \end{equation}
  Since $s \in \RS_k(S)$ and $\neg (s \kTRANSS{\PSEND{pq}{a}})$, we must
  have $\lvert \word_{\ptp{pq}} \rvert = k$.
    Hence, $s \bTRANSS{\PSEND{pq}{a}}{k+1}$ and therefore
  \begin{equation}\label{eq:deltap}
    (s, {\PSEND{pq}{a}},s'') \in \Delta' \qquad \text{for some } s'' \in N'
  \end{equation}
    
  By~\eqref{eq:wbisim} and the fact that $ \Delta \subseteq \Delta'$
  and $N \subseteq N'$, we must have 
    \[
  \epsproj{(N, s, \Delta)}{p} \wbisim \epsproj{(N', s, \Delta')}{p}
  \]
  which is clearly a contradiction with~\eqref{eq:nosend}
  and~\eqref{eq:deltap}.
\end{proof}

\corbisimcsa*
\begin{proof}
  Take $S$ such that
  $\exists k \qst \forall \p \in \PSet \qst \epsproj{\kTS{S}}{p} \wbisim
  \epsproj{\TS{k+1}{S}}{p}$.
    Then, by Lemma~\ref{lem:completeness}, $S$ is $k$-exhaustive.
    Since $S$ is $\OBI{k}$ and $\IBI{(k{+}1)}$ by assumption, $S$ is
  $n$-exhaustive for any $n \geq k$, by
  Lemma~\ref{lem:k-exhau-kplus-exhau}.
    Hence, by Lemma~\ref{lem:kexh-imp-weakbisim-sys}, we have
  $\forall \p \in \PSet \qst \epsproj{\TS{n}{S}}{p} \wbisim
  \epsproj{\TS{n{+}1}{S}}{p}$ (for any $n \geq k$).
\end{proof}

\thmcompleteness*
\begin{proof}
      Part (1)
    follows from Lemma~\ref{lem:completeness}
                      and
  Part (2)
  follows from Lemma~\ref{lem:kexh-imp-weakbisim-sys}.
                  \end{proof}

 \section{Synchronous multiparty compatibility}\label{app:smc}

We adapt the definition of (synchronous) multiparty compatibility
from~\cite[Definition 4]{BLY15} to our setting (this definition is adapted
from~\cite[Definition 4.2]{DY13}).

We write $\sync{\acts}$ iff $\acts = \emptyw$, or
$\acts = \PSEND{pq}{a} \concat \PRECEIVE{pq}{a} \concat \acts'$ and
$\sync{\acts'}$.
We say that $s$ is \emph{stable} iff $s = \stablecsconf{q}$ and define
$\RS_0(S)$ as follows:
\[
\RS_0(S) \defi 
\{ s \st s_0 \bTRANSS{\acts}{1} s \land \sync{\acts} \} \cup
\{ s \st s_0 \bTRANSS{\acts \concat \PSEND{pq}{a}}{1} s \land \sync{\acts} \}
\]

\begin{definition}[$\TS{0}{S}$]\label{def:ts-zero}
    The synchronous transition system of $S$, written $\TS{0}{S}$, is
  the labelled transition system $(N, s_0, \Delta)$ such that
  $N = \RS_0(S)$, $s_0$ is the initial configuration of $S$,
  $\Delta \subseteq N {\times} \ASet {\times} N$ is the transition
  relation such that
  \begin{itemize}
  \item $(s, \PSEND{pq}{a}, s'') \in \Delta$ iff
    $\exists s' \in N \qst s \bTRANSS{\PSEND{pq}{a}}{1} s'
    \bTRANSS{\PRECEIVE{pq}{a}}{1} s'' $; and
      \item $(s', \PRECEIVE{pq}{a}, s'') \in \Delta$ iff
    $\exists s \in N \qst s \bTRANSS{\PSEND{pq}{a}}{1} s'
    \bTRANSS{\PRECEIVE{pq}{a}}{1} s'' $.
  \end{itemize}
      We write $s_1 \bTRANSS{\acts}{0} s_{n+1}$ if
  $\acts = \action_1 \cdots \action_n$ and
  $\forall 1 \leq i \leq n \qst (s_i, \action_i, s_{i+1}) \in \Delta$.
\end{definition}

\begin{definition}[Synchronous multiparty compatibility~\cite{BLY15}]\label{def:smc}
  $S$ is \emph{synchronous multiparty compatible} (\SMC) if, letting
  $\TS{0}{S} = (N, s_0, \Delta)$, for all $\p \in \PSet$, for all
  $q \in Q_\p$, and for all \emph{stable} $\csconf{q}{w} \in N$, if
  $q = q_\p$, then
  \begin{enumerate}
  \item if $q_\p$ is a sending state, then
    $\forall (q,\action,q') \in \delta_\p \qst \exists \acts \qst   
        \sync{\acts} \land
        s  \bTRANSS{\acts}{0}
    \land \proj{\acts}{p} = \action
    $;
  \item if $q_\p$ is a receiving state, then
    $\exists (q,\action,q') \in \delta_\p \qst \exists \acts \qst   
        \sync{\acts} \land
        s  \bTRANSS{\acts}{0}
    \land \proj{\acts}{p} = \action
    $.
  \end{enumerate}
\end{definition}

\begin{restatable}{lemma}{lemsmcclosed}\label{lem:k-smc-closed}
    Let $S$ be \emph{directed} and \SMC. 
    For all
  \emph{stable} $s \in \RS_0(S)$, if $s \TRANSS{\PSEND{pq}{a}}$ and
  $\paset = \{ \acts \st s \bTRANSS{\acts}{0}\bTRANSS{\PSEND{pq}{a}}{0}
  \land \p \notin \acts \land \sync{\acts}\}$,
  then $\nclosed{1}{\paset \neq \emptyset}{s}$.
  \end{restatable}
\begin{proof}
  The proof is an instance of the proof of
  Lemma~\ref{lem:k-mc-kclosed}, noting that (1) \SMC\ guarantees the
  existence of a synchronous execution that includes all send actions
  enabled at a given sending state and (2) directedness implies
  $\OBI{1}$. 
        \end{proof}

\begin{lemma}\label{lem:smc-eixstpath-eqpeer}
  Let $S$ be \emph{directed} and \SMC, then for all $s \in \RS_1(S)$
  such that $s_0 \bTRANSS{\acts}{1} s$, there are
  $\acts', \actsb \in \ASetC$ and stable $t \in \RS_0(S)$ such that
  $s \bTRANSS{\acts'}{1} t$, $s_0 \bTRANSS{\actsb}{0} t$,
  $\acts \concat \acts' \eqpeerop \actsb$, and $\sync{\actsb}$.
\end{lemma}
\begin{proof}
  Since $S$ is directed and \SMC, we can use
  Lemma~\ref{lem:k-smc-closed} and Lemma~\ref{lem:closed-set-paths},
  to show that the result holds following the same reasoning as in
  Lemma~\ref{lem:exist-path-eqpeer}.
\end{proof}

\begin{theorem}\label{thm:smc-imp-one-mc}
  If $S$ is \SMC\ then it is $1$\MC. 
\end{theorem} 
\begin{proof}
  We show that $S$ is $1$-exhaustive, then show that it is $1$-safe.

  \proofsub{exhaustivity}
    We have to show that 
  for all
  $s = \csconf{q}{w} \in \RS_1(S)$ and $\p \in \PSet$,
    if $q_\p$ is a sending state, then
    $
  \forall (q_\p, \action, q'_\p) \in \delta_\p \qst \exists
  \acts \in \ASetC
  \qst
  s  \bTRANSS{\acts}{1}\bTRANSS{\action}{1}  \text{ and } \p \notin
  \acts.
  $

  By contradiction take $s = \csconf{q}{w} \in \RS_1(S)$ 
  \begin{equation}\label{eq:smc-imp-one-mc}
    s \TRANSS{\PSEND{pq}{a}} \text{ and } 
    \neg (s \bTRANSR{1}\bTRANSS{\PSEND{pq}{a}}{1} \text{ and } \p
    \notin \acts)
  \end{equation}
    By Lemma~\ref{lem:smc-eixstpath-eqpeer}, there is stable
  $t \in \RS_0(S)$ such that $s \bTRANSS{\acts}{1} t$.
    If $\p \notin \acts$, then $\PSEND{pq}{a}$ is still enabled in $t$
  and by \SMC\ there is a (synchronous) execution from $t$ that
  includes $\PSEND{pq}{a}$, a contradiction
  with~\eqref{eq:smc-imp-one-mc}.
    If $\p \in \acts$, then $\PSEND{pq}{a}$ can be fired from a state
  along $\acts$, a contradiction with~\eqref{eq:smc-imp-one-mc}.

  \proofsub{safety} We have to show that for all
  $ s = \csconf q w \in \RS_1(S)$:
  \begin{enumerate}
  \item \textbf{Eventual reception:}
    $\forall \p\q \in \CSet$, if $w_{\p\q} = \msg{a} \cdot w'$,
    then $s \kTRANSR \kTRANSS{\PRECEIVE{pq}{a}}$.
        This follows trivially from Lemma~\ref{lem:smc-eixstpath-eqpeer}
    since there is a $1$-bounded execution from $s$ to a stable
    configuration.
  \item \textbf{Progress:} $\forall\p \in \PSet$, if $q_\p$ is a
    \emph{receiving} state, then
    $s \kTRANSR \kTRANSS{\PRECEIVE{qp}{a}}$ for $\q \in \PSet$ and
    $\msg{a} \in \ASigma$.
        By Lemma~\ref{lem:smc-eixstpath-eqpeer}, there is a $1$-bounded
    execution $\acts$ from $s$ to a stable $t \in \RS_0(S)$.
        If the expected receive action occurs in $\acts$, then we have the
    required result.
        If the expected receive action does not occur in $\acts$, then
    \SMC\ guarantees that it will occur in a synchronous execution
    from $t$.
        \qedhere
  \end{enumerate}  
  \end{proof}

 \section{Proofs for Section~\ref{sec:por} (partial order reduction)}\label{app:rts}

Below, we say that a configuration $s \in \RS_k(S)$ is $\OBI{k}$
(resp.\ $\IBI{k}$) if it validates the corresponding condition, e.g., if
$\p$ can fire one send action from $s$, then all its send actions are
enabled.
We say that $S$ (resp.\ $s$) is $\BA{k}$ when it is $\OBI{k}$ and
$\IBI{k}$.

\begin{definition}\label{def:non-csa-mc-dep-reduced}
  We say that $S$ is \emph{reduced} $k$-\emph{chained input bound
    independent} (reduced $\CIBI{k}$)
    if for all $s = \csconf{q}{w} \in \RS_k(S)$ and for all 
  $\p \in \PSet$, if $s \rkTRANSS{\PRECEIVE{qp}{a}} s'$, then
    $ \forall (q_\p, \PRECEIVE{sp}{b}, q'_\p) \in \delta_\p \qst \s \neq
  \q \implies \!  
    \neg ( s \rkTRANSS{\PRECEIVE{sp}{b}})
  \land
    (
  \forall \acts \in \ASetC \qst 
  s' \rkTRANSS{\acts} \rkTRANSS{\PSEND{sp}{b}}
  \implies \dependsin{s}{\PRECEIVE{qp}{a}}{\acts}{\PSEND{sp}{b}})
    $.
  
\end{definition}

\lemobiimpreducedobi*
\begin{proof}
  
  By contradiction. Notice that Definition~\ref{def:reduced-obi}
  requires the same property than Definition~\ref{def:kobi} at the
  configuration level.
    Take $s \in \hat{N}$ s.t. $s$ violates the
  (reduced) $\OBI{k}$ condition, then $s \in \RS_k(S)$, and $s$ also
  violates $\OBI{k}$.
\end{proof}

\begin{lemma}\label{lem:dibi-imp-reduced-dibi}
  Let $S$ be a system, if $S$ is $\DIBI{k}$, then $S$ is
  also \emph{reduced} $\DIBI{k}$.
\end{lemma}
\begin{proof}
  By contradiction. Take $s \in \hat{N}$ s.t. it violates the
  (reduced) $\DIBI{k}$ condition. Note that we $s \in \RS_k(S)$.
    There are two cases:
  \begin{itemize}
  \item If there is $\p$ such that two receive actions are enabled for
    $\p$, then they are also enabled at $s$, a contradiction. 
      \item If there is $\p$ such that one receive action is enabled for
    $\p$, and there is $\kTRANSS{}$-path s.t.\ a conflicting send
    action is fired, then we have the situation in $\kTS{S}$, hence we
    have a contradiction. \qedhere
  \end{itemize}
\end{proof}

\begin{lemma}\label{lem:dibi-imp-reduced-CIBI}
  Let $S$ be a system, if $S$ is $\CIBI{k}$, then $S$ is
  also \emph{reduced} $\CIBI{k}$.
\end{lemma}
\begin{proof}
  By contradiction. Take $s \in \hat{N}$ s.t. it violates the
  (reduced) $\CIBI{k}$ condition. Note that we $s \in \RS_k(S)$.
    There are two cases:
  \begin{itemize}
  \item If there is $\p$ such that two receive actions are enabled for
    $\p$, then they are also enabled at $s$, a contradiction. 
      \item If there is $\p$ such that one receive action is enabled for
    $\p$, and there is $\kTRANSS{\acts}$-path s.t.\ a conflicting send
    action is fired, and there is not dependency chain in $\acts$,
    then we have the situation in $\kTS{S}$, hence we have a
    contradiction. \qedhere
  \end{itemize}
\end{proof}

Lemma~\ref{lem:indep-partition} states that any transition in a given
set $L_i$ cannot be disabled by a sequence of transitions not in
$L_i$.

\begin{restatable}{lemma}{lemindeppartition}\label{lem:indep-partition}
  Let $S$ be a system, $s \in \RS_k(S)$ s.t.\ $s$ is $\BA{k}$, and
  $L_1 \cdots L_n = \spartition{s}$ (with $n \geq 1$).
    For all $L_i$ (with $1 \leq i \leq n$) and for all
  $\acts = \action_1 \cdots \action_m$ such that
  $\forall 1 \leq j \leq m \qst \action_j \notin L_i$, if
  $s \kTRANSS{\acts} s'$, then
  $\action \in L_i \implies s' \kTRANSS{\action}$.
\end{restatable}
\begin{proof} 
  Take $s \in \kTS{S}$, $L_1 \cdots L_n = \spartition{s}$, $L_i$
  ($1 \leq i \leq n$), and $\acts$ as defined in the statement.
      Take any $\action \in L_i$ and assume there is $s'$ such that
  $s \kTRANSS{\acts} s'$.
    We show the result by induction on the length of $\acts$ with the
  additional property that $\subj{\action} \notin \acts$ (note that
  this implies $q_\p = q'_\p$).
  
  If $\acts = \emptyw$, then $s= s'$ and we have the result
  immediately ($s \kTRANSS{\action}$ by
  Definition~\ref{def:partition}).

  Assume the result holds for $\acts$ and let us show that it holds
  for $\acts \concat \action'$ with $\action' \notin L_i$.
    Assume we have $s'$ such that
  $s \kTRANSS{\acts} s' \kTRANSS{\action'} s''$. 
    We have to show that $s'' \kTRANSS{\action}$, knowing that,
    by induction hypothesis, we have that $s' \kTRANSS{\action}$ and
  $q_\p = q'_\p$.
    There are two cases:
  \begin{itemize}
  \item If $\subj{\action} = \subj{\action'}$, then since $s$ is
    $\BA{k}$, we have $s \kTRANSS{\action'}$, hence
    $\action' \in L_i$, which implies that the premises of this lemma
    do not hold: a contradiction.
              \item If $\subj{\action} \neq \subj{\action'}$, then we have
    $q_\p = q'_\p = q''_\p$ and therefore
        $q''_\p \TRANSS{\action}$.
        \begin{itemize}
    \item If $\action = \PSEND{pq}{a}$. The only possibility for
      $\action$ to be disabled in $s''$ and enabled in $s'$ is if
      $\lvert \word''_{\ptp{pq}} \rvert > k$ which is not possible
      since $\subj{\action'} \neq \p$.
    \item If $\action = \PRECEIVE{qp}{a}$. The only possibility for
      $\action$ to be disabled in $s''$ and enabled in $s'$ is if
      $ \word''_{\ptp{pq}} = \emptyw $ which is not possible since
      $\subj{\action'} \neq \p$. 
            \qedhere
    \end{itemize}
  \end{itemize}
  \end{proof}

\begin{restatable}{lemma}{lemrtsevtfire}\label{lem:rts-evt-fire}
  Let $S$ be a system, then for all $s \in \kRTS{S}$ s.t.\ $s$ is
  $\BA{k}$ and $\action \in \ASet$, if $s \kTRANSS{\action}$, then
  there is $\acts \in \ASetC$ such that
  $s \rkTRANSS{\acts} \rkTRANSS{\action}$ with
  $\subj{\action} \notin \acts$.
\end{restatable}
\begin{proof}
    By assumption that $s \in \kRTS{S}$, $s$ is visited by
  Algorithm~\ref{algo:reduction}.
  
  If $\spartition{s}$ is invoked on $s$, the fact that
  $\subj{\action} \notin \acts$ follows from
  Definition~\ref{def:partition}, while the fact that $\action$ is
  eventually fired follows from the fact that the list of sets of
  transition decreases at each iteration in
  Algorithm~\ref{algo:reduction} and Lemma~\ref{lem:indep-partition}.

  If $\spartition{s}$ is not invoked, then we have that $E$ is not
  empty when $s$ is visited.
    Let $t$ be a the last node visited before $s$ such that
  $\spartition{t}$ is invoked.
    Pose $L_1 \cdots L_m = \spartition{t}$ and assume
  $E = L_i \cdots L_m$ ($i>1$) when $s$ is visited.
    If there is $L_j$ such that $\action \in L_j$ ($i \leq j \leq m$),
  we have the result as above.
    Otherwise, there are two cases
  \begin{itemize}
  \item If $\action$ is independent from all the actions in
    $L_i \cdots L_m$, then $\action$ will still be enabled once the
    list is entirely processed, and therefore $\action$ will be
    included in the partition resulting from the next invocation of
    $\spartition{\_}$.
  \item If $\action$ depends on some partition $L_j$, then we have a
    contradiction: either $\action$ is included in $L_j$ (it must have
    been enabled at $t$) or the list returned by $\spartition{t}$ is
    not a partition.
    \qedhere
  \end{itemize}
  \end{proof}

\begin{restatable}{lemma}{lemrtsreplaceaction}\label{lem:rts-replace-action}
  Let $S$ be a system. If
  $s_0 \rkTRANSS{\acts_1} s \rkTRANSS{\action} s' \rkTRANSS{\acts_2}
  t$
  such that $s$ is $\BA{k}$, $\subj{\action} \notin \acts_2$,
  $\chan{\action} \notin \acts_2$, and $s \kTRANSS{\action'}$ with
  $\subj{\action} = \subj{\action'}$ then
  $ s \rkTRANSS{\action'} s'' \rkTRANSS{\acts_2} t'$ for some $s''$
  and $t'$.
  \end{restatable}
\begin{proof}
  Assume that $E = L_1 \cdots L_m$ when $s$ is visited by
  Algorithm~\ref{algo:reduction}, then we have
  $\action, \action' \in L_1$ and $s \rkTRANSS{\action'} s''$ for some
  $s''$.
    When both $s'$ and $s''$ are visited next, we have
  $E = L_2 \cdots L_m$, hence it is easy to show they have the same
  behaviour while $E$ is not empty.
    Say $s_m$ (resp.\ $s'_m$) is the first state reachable from $s'$
  (resp.\ $s''$) when $E$ is empty.
        Note that if $\action$ is a receive action, then we must have
  $\action = \action'$ since $s$ is $\BA{k}$.
    Thus, the only differences between $s_m$ and $s'_m$ are:
  \begin{itemize}
  \item the local state of $\subj{\action}$
  \item the last message of channel $\chan{\action}$
  \end{itemize}
  In terms of enabled transition, this means that for all
  $\hat{\action}$ such that $\subj{\hat{\action}} \neq \subj{\action}$
  and $\chan{\hat\action} \neq \chan{\action}$ is enabled at both
  $s_m$ and $s'_m$. Hence, posing
  \[
  L'_1 \cdots L'_j = \spartition{s_m}
  \qquad \text{and} \qquad
  L''_1 \cdots L''_l = \spartition{s'_m}
  \]
  and assuming that the position of the partition of $\subj{\action}$
  is $i$ (with $1 \leq i \leq j$ and $i \leq l$), it must be the case
  that all paths of length less than $i$ and not involving
  $\chan{\action}$ nor $\subj{\action}$ are the same from both $s_m$
  and $s'_m$. Instead, any path longer than $i$ must use an action
  whose subject is $\subj{\action}$ at position $i$, hence does not
  satisfy the premises of this lemma.
    \end{proof}

\begin{restatable}{lemma}{lemportraceequiv}\label{lem:por-trace-equiv}
    Let $S$ be a \emph{reduced} $\BA{k}$ system such that
  $\kTS{S} = (N, s_0, \Delta)$,
  $\kRTS{S} = (\hat{N}, s_0, \hat\Delta)$, and
  $\metaszero \in N \cap \hat{N}$.
    The following holds:
  \begin{enumerate}
  \item \label{en:por-trace-equiv-ts}
    If $\metaszero \rkTRANSS{\acts} s$, then $\metaszero \kTRANSS{\acts} s$, for some $s$.
      \item \label{en:por-trace-equiv-rts} If
    $\metaszero \kTRANSS{\acts} s$, then there is $\actsb$ and
    $\acts'$ such that $\metaszero \rkTRANSS{\actsb} t$ and
    $s \kTRANSS{\acts'} t$ and $\acts \concat \acts' \eqpeerop \actsb$,
    for some $t$.
      \end{enumerate}
\end{restatable}
\begin{proof}
  Item~\eqref{en:por-trace-equiv-ts} follows trivially from
  Definition~\ref{def:partition} and Algorithm~\ref{algo:reduction}, since
  only transitions that exist in $\kTS{S}$ are copied in $\kRTS{S}$.

  We show Item~\eqref{en:por-trace-equiv-rts} by induction on the
  length of $\acts$.
    If $\acts = \emptyw$, then we have the result with
  $\acts' = \actsb = \emptyw$.
    Assume the result holds for $\acts$ and let us show that it holds
  for $\acts \concat \action$.
    We have the following situation, where the dotted arrows represent
  executions in $\kRTS{S}$ and $t$ is in $\kRTS{S}$.\footnote{Note
    that executions in $\kRTS{S}$ are also in $\kTS{S}$ by
    Item~\eqref{en:por-trace-equiv-ts}.}
  \[
  \begin{tikzpicture}[baseline=(current bounding box.center)]
    \node (s1) {$\metaszero$};
    \node[right=of s1] (s2) {$s$};
    \node[right=of s2] (s5) {$s'$};
        \node at (s2|-s3) (s4) {$t$};
        \path[->] (s1) edge [above] node {$\acts$} (s2);
    \path[->] (s2) edge [above] node {$\action$} (s5);
        \path[->,densely dotted] (s1) edge [below, bend right] node {$\actsb$} (s4);
    \path[->] (s2) edge [right] node {$\acts'$} (s4);
  \end{tikzpicture}
    \]
  Next, we show that there are $t'$, $\hat{s}$, $s''$, and $\actsb'$
  such that we have:
  \[
  \begin{tikzpicture}[baseline=(current bounding box.center)]
    \node (s1) {$s$};
    \node[right=of s1] (s2) {$s'$};
    \node[below=of s1] (s3) {$t$};
    \node[below=of s3] (s5) {$t'$};
    \node at (s2|-s5) (s4) {$s''$};
    \node at (s2|-s3) (s6) {$\hat{s}$};
        \path[->] (s1) edge [above] node {$\action$} (s2);
    \path[->] (s3) edge [above] node {$\action$} (s6);
    \path[->,densely dotted] (s5) edge [above] node {$\action$} (s4);
        \path[->] (s1) edge [left] node {$\acts'$} (s3);
    \path[->,densely dotted] (s3) edge [left] node {$\actsb'$} (s5);
    \path[->] (s2) edge [right] node {$\acts'$} (s6);
    \path[->] (s6) edge [right] node {$\actsb'$} (s4);
  \end{tikzpicture} 
  \qquad
  \text{with $\acts' \concat \actsb' \concat \action
    \eqpeerop \action \concat \acts' \concat \actsb'$}
  \]
    We show this by induction on the length of $\acts'$.
    If $\acts' = \emptyw$, then we have $s=t$ and $s' = \hat{s}$.
  There are two cases:
  \begin{itemize}
  \item $E = []$ when $t$ is visited by
    Algorithm~\ref{algo:reduction}.  In this case, the algorithm
    continues with $ E = L_1 \cdots L_m = \spartition{s}$, and by
    Definition~\ref{def:partition} there must be $1 \leq i \leq m$
    such that $\action \in L_i$ (since $\action$ is enabled at $t$).
        Since $\action$ is independent with all $\action_j$ such that
    $1 \leq j < i$, we have:
    \[
    s = t \rkTRANSS{\action_1 \cdots \action_{i-1}} t' \rkTRANSS{\action} s''
    \quad 
    \text{and}
    \quad
    s = t \kTRANSS{\action} s' = \hat{s} \kTRANSS{\action_1 \cdots \action_{i-1}} s''
    \]
    We have the required result with
    $\actsb' =  \action_1 \cdots \action_{i-1}$.
                          \item $E = L_i \cdots L_m$ ($i > 0$) when $t$ is visited by
    Algorithm~\ref{algo:reduction}. Then we have two cases:
    \begin{itemize}
    \item There is $i \leq j \leq m$ such that $\action \in L_j$ and
      we reason as in the case where $E == []$ (but starting at $i$
      instead of $1$).
    \item If $\action \notin \bigcup_{i \leq j \leq m} L_j$, then
      $\action$ was not enabled when $\spartition{\hat{t}}$ was
      invoked (for $\hat{t}$ a node visited on the path to $s$).
            Hence, $\action$ is independent with all actions in
      $\bigcup_{i \leq j \leq m} L_j$ and for all $t''$ such that
      $t \rkTRANSS{\action_i \cdots \action_m} t''$ with
      $\forall i \leq j \leq m \qst \action_j \in L_j$, we have
      $t'' \kTRANSS{\action}$.
            Pose $L'_1 \cdots L'_n = \spartition{t''}$, then we have that
      there is $1 \leq j \leq n$ such that $\action \in L'_j$. 
            Reasoning as above, we have
      \[
      s = t \rkTRANSS{\action_i \cdots \action_{m}} t'' 
            \rkTRANSS{\action'_1 \cdots \action'_{j-1}} t' \rkTRANSS{\action} s''
      \]
            and
            \[
      s = t \kTRANSS{\action} s' \kTRANSS{\action_i \cdots \action_{m}} 
            \kTRANSS{\action'_1 \cdots \action'_{j-1}} s''
      \]
      We have the required result with
      $\actsb' =  \action_i \cdots \action_{m} \concat \action'_1 \cdots \action'_{j-1}$.
    \end{itemize}
  \end{itemize}
  Now, assuming the inner induction hypothesis holds, let us show the
  result for $\acts' \concat \action'$. We have the following
  situation, where the red parts are what is to be shown:
  \[
  \begin{tikzpicture}[baseline=(current bounding box.center)]
    \node (s1) {$s$};
    \node[right=of s1] (s2) {$s'$};
    \node[below=of s1] (s3) {$t_i$};
    \node[below=of s3] (s5) {$t$};
    \node[below=of s5,red] (s7) {$t'$};
    \node[red] at (s2|-s5) (s4) {$\hat{s}$};
    \node at (s2|-s3) (s6) {$s_i$};
    \node[red] at (s2|-s7) (s8) {${s''}$};
        \path[->] (s1) edge [above] node {$\action$} (s2);
    \path[->] (s3) edge [above] node {$\action$} (s6);
            \path[->] (s1) edge [left] node {$\acts'$} (s3);
    \path[->] (s3) edge [left] node {$\action'$} (s5);
    \path[->] (s2) edge [right] node {$\acts'$} (s6);
    \path[->,red] (s6) edge [right] node {$\action'$} (s4);
    \path[->,red] (s4) edge [right] node {$\actsb'$} (s8);
    \path[->,red] (s5) edge [right] node {$\actsb'$} (s7);
    \path[->,red] (s7) edge [above] node {$\action$} (s8);
        \path[->,red] (s5) edge [above] node {$\action$} (s4);
  \end{tikzpicture} 
    \]
  There are two cases.
  \begin{itemize}
  \item If $\subj{\action} \neq \subj{\action'}$, then the two actions
    commute from $t_i$ and we have the result with $\actsb' = \emptyw$.
    
  \item If $\subj{\action} = \subj{\action'}$, then there are two cases:
    \begin{itemize}
    \item If $\action = \action'$, then $t = s_i$ (by determinism) and
      we have the result with $\acts' = \emptyw$.
    \item If $\action \neq \action'$, then we must have
      $\actsb = \actsb_1 \concat \action' \concat \actsb_2$ with
      $\subj{\action'} \notin \actsb_2$
      (since
      $\actsb \eqpeerop \acts \concat \acts' \concat \action$ by (outer)
      induction hypothesis).
            Since $\action'$ and $\action$ have the same subject, there is
      $\hat{t} \in \kRTS{S}$ such that
      $\metaszero \rkTRANSS{\actsb_1} \hat{t}$ such that
      $\hat{t} \rkTRANSS{\action}$ and $\hat{t} \rkTRANSS{\action'}$
      by $\BA{k}$.

      Thus, by Lemma~\ref{lem:rts-replace-action}, we also have
      $\metaszero \rkTRANSS{\actsb_1} \hat{t} \rkTRANSS{\action}
      \rkTRANSS{\actsb_2} t''$ for some $t''$.
            By (outer) induction hypothesis, we have
      $\actsb = \actsb_1 \concat \action' \concat \actsb_2 \eqpeerop
      \acts \concat \acts' \concat \action'$
      and since $\subj{\action'} \notin \actsb_2$, we also have
      $\actsb_1 \concat  \actsb_2 \eqpeerop \acts
      \concat \acts'$
            and
      $\actsb_1 \concat \action \concat \actsb_2 \eqpeerop
      \acts \concat \acts' \concat \action$ hence $s_i = t''$.
                  Since $t''$ is in $\kRTS{S}$, we have the required result with
      $\actsb' = \emptyw$.
   
    \end{itemize}
  \end{itemize}

  \noindent
  Going back to the outer induction, we have to show that
  \[
  \actsb \concat \actsb' \concat \action \eqpeerop \acts \concat \action \concat \acts''
  \]
  In other words, $ \acts \concat \action \in \kTS{S}$ can be extended with
  $\acts''$ so that there is an equivalent execution in $\kRTS{S}$,
  i.e., $\actsb \concat \actsb' \concat \action$.
        By induction hypothesis, we have
  $\actsb \eqpeerop \acts \concat \acts'$, hence we have
  \[
  \actsb \concat \actsb' \concat \action 
  \eqpeerop 
  \acts \concat \acts'  \concat \actsb' \concat \action 
  \]
  From the inner induction, we know that
  $ \acts' \concat \actsb' \concat \action \eqpeerop \action \concat
  \actsb''$, hence, we have
  \[
  \acts \concat \acts'  \concat \actsb' \concat \action 
  \eqpeerop
    \acts \concat \action \concat \acts''
  \]
  and thus we have the required result.
\end{proof}

\begin{lemma}\label{lem:ts-path-rts}
  Let $S$ be reduced $\BA{k}$, for all $s \in \RS_k(S)$, there is
  $t \in \kRTS{S}$ such that $s \kTRANSS{\acts} t$.
\end{lemma}
\begin{proof}
  Since $s \in \RS_k({S})$, there is $\actsb$ such that
  $s_0 \kTRANSS{\actsb} s$.
    Since $s_0 \in \kRTS{S}$, we can apply
  Lemma~\ref{lem:por-trace-equiv} and obtain the required result.
\end{proof}

\begin{restatable}{lemma}{lemreduceddibiimpdibi}\label{lem:reduced-dibi-imp-dibi}
  If $S$ is reduced $\OBI{k}$ and reduced $\DIBI{k}$, then $S$ is
  $\DIBI{k}$.
\end{restatable}
\begin{proof}
  By contradiction. Take $s_0 \kTRANSS{\acts} s = \csconf{q}{w} \in \RS_k(S)$.
  \begin{itemize}
  \item If $ s \kTRANSS{\PRECEIVE{pr}{a}} s_1$ and
    $ s \kTRANSS{\PRECEIVE{sr}{b}} s_2$.
        Then, by Lemma~\ref{lem:por-trace-equiv}, there is $t \in \hat{N}$
    s.t.\ $s_o \rkTRANSS{\actsb} t$ and $s_1 \kTRANSS{\acts''} t$ and
    $\acts \concat \PRECEIVE{pr}{a} \concat \acts'' \eqpeerop
    \actsb$. 
        Then both $\PRECEIVE{pr}{a}$ and $\PSEND{sr}{b}$ must appear in
    $\actsb$, which contradicts the assumption that $S$ is reduced
    $\DIBI{k}$.
      \item If $ s \kTRANSS{\PRECEIVE{pr}{a}} s_1$ and there is
    $(q_\rr, \PRECEIVE{sr}{b}, q'_\rr) \in \delta_\rr$ s.t.\
    $s \kTRANSS{\acts'} \kTRANSS{\PSEND{sp}{b}}s'$.
        Then we have a contradiction with the assumption that $S$ is
    reduced $\DIBI{k}$, via by Lemma~\ref{lem:por-trace-equiv} as
    above, with
        $\acts \concat \PRECEIVE{pr}{a} \concat \acts' \concat
    \PSEND{sp}{b} \concat \acts'' \eqpeerop \actsb$.
            \qedhere
  \end{itemize}
\end{proof}

\begin{restatable}{lemma}{lemreduceddibiimpCIBI}\label{lem:reduced-dibi-imp-CIBI}
  If $S$ is reduced $\OBI{k}$ and reduced $\CIBI{k}$, then $S$ is
  $\CIBI{k}$.
\end{restatable} 
\begin{proof}
  By contradiction. Take $s_0 \kTRANSS{\acts} s = \csconf{q}{w} \in \RS_k(S)$.
  \begin{itemize}
  \item If $ s \kTRANSS{\PRECEIVE{pr}{a}} s_1$ and
    $ s \kTRANSS{\PRECEIVE{sr}{b}} s_2$.
        Then, by Lemma~\ref{lem:por-trace-equiv}, there is $t \in \hat{N}$
    s.t.\ $s_o \rkTRANSS{\actsb} t$ and $s_1 \kTRANSS{\acts''} t$ and
    $\acts \concat \PRECEIVE{pr}{a} \concat \acts'' \eqpeerop
    \acts$ and $\actsb$. 
        Clearly, we must have both $\PSEND{pr}{a}$ and $\PSEND{sr}{b}$
    in $\actsb$.
    \begin{itemize}
    \item If we have
      \[
      \begin{array}{ccl} 
        \actsb & = &\actsb_1 \concat \PSEND{pr}{a} \concat \actsb_2
                     \concat \PSEND{sr}{b} \concat \actsb_3 \concat \PRECEIVE{pr}{a} \concat
                     \actsb_4, \text{or}
        \\
                                \actsb &= &\actsb_1 \concat \PSEND{sr}{b} \concat \actsb_2
                    \concat \PSEND{pr}{a} \concat \actsb_3 \concat \PRECEIVE{pr}{a} \concat
                    \actsb_4
      \end{array}
      \]
      where $\actsb_1$, $\actsb_2$, and $\actsb_3$ have been chosen
      appropriately so that the send actions are one matched at $s$,
            then we have a contradiction with the assumption that $S$ is
      $\CIBI{k}$ (both messages can be consumed).
          \item Assume we have
      \[
      \actsb = \actsb_1 \concat \PSEND{pr}{a} \concat \actsb_2 \concat
      \PRECEIVE{pr}{a} \concat \actsb_3 \concat \PSEND{sr}{b}\concat
      \actsb_4
      \]
      where $\actsb_1$, $\actsb_2$, and $\actsb_3$ have been chosen
      appropriately so that the send actions are one matched at $s$.
            Since $S$ is reduced $\CIBI{k}$, we must have
      $\dependsin{\hat{s}}{\PRECEIVE{pr}{a}}{\actsb_3}{\PSEND{sr}{b}}$,
      with $\hat{s}$ such that
      $s_0 \rkTRANSS{\actsb_1 \concat \PSEND{pr}{a} \concat \actsb_2}
      \hat{s}$.
            However, $\PSEND{pr}{a}$ and $\PSEND{sr}{b}$ appear in $\acts$,
      which contradicts the existence of a dependency chain between
      $\PRECEIVE{pr}{a}$ and $\PSEND{sr}{b}$ by
      Lemma~\ref{lem:dep-chain-equi}.
    \end{itemize}
          \item If $ s \kTRANSS{\PRECEIVE{pr}{a}} s_1$ and there is
    $(q_\rr, \PRECEIVE{sr}{b}, q'_\rr) \in \delta_\rr$ s.t.\
    $s_1 \kTRANSS{\acts'} \kTRANSS{\PSEND{sp}{b}}s'$ with
    $\neg (\dependsin{s}{\PRECEIVE{pr}{a}}{\acts'}{\PSEND{sr}{b}})$.
            Then, by Lemma~\ref{lem:por-trace-equiv}, there is $t \in \hat{N}$
    s.t.\ $s_0 \rkTRANSS{\actsb} t$, $s_1 \kTRANSS{\acts''} t$, and
    \[
    \acts \concat \PRECEIVE{pr}{a} \concat \acts' \concat
    \PSEND{sp}{b} \concat \acts'' \eqpeerop \actsb
    \]
    There are two cases depending on the structure of $\actsb$:
    \begin{itemize}
    \item If $\PSEND{sp}{b}$ appears before $\PRECEIVE{pr}{a}$ in
      $\actsb$, then we have a contradiction with the assumption that
      $S$ is reduced $\CIBI{k}$.
    \item  If $\PSEND{sp}{b}$ appears after $\PRECEIVE{pr}{a}$, then pose
      \[
      \actsb =
      \actsb_1 \concat \PRECEIVE{pr}{a} \concat \actsb_2 \concat
      \PSEND{sp}{b} \concat \actsb_3
      \]
      Since $S$ is reduced  $\CIBI{k}$, we must have
      $\dependsin{\hat{s}}{\PRECEIVE{pr}{a}}{\actsb_2}{ \PSEND{sp}{b} }$
      assuming $\hat{s}$ is such that $s_0 \rkTRANSS{\actsb_1} \hat{s}$.
            By Lemma~\ref{lem:dep-chain-equi}, we have a contradiction with
      the assumption that
            $\neg (\dependsin{s}{\PRECEIVE{pr}{a}}{\acts'}{\PSEND{sr}{b}})$.
                              \qedhere
    \end{itemize}
      \end{itemize}
\end{proof}

\begin{restatable}{theorem}{thmreduceddibiiffdibi}\label{thm:reduced-dibi-iff-dibi}
  Let $S$ be reduced $\OBI{k}$. $S$ is reduced $\DIBI{k}$ iff $S$ is
  $\DIBI{k}$.
\end{restatable}
\begin{proof}
  By Lemma~\ref{lem:reduced-dibi-imp-dibi} and
  Lemma~\ref{lem:dibi-imp-reduced-dibi}.
\end{proof}

\begin{restatable}{lemma}{lemkexhauimpredkexhau}\label{lem:kexhau-imp-red-kexhau}
  Let $S$ be reduced $\BA{k}$, if $S$ is $k$-exhaustive, then $S$ is
  also \emph{reduced} $k$-exhaustive.
\end{restatable}
\begin{proof}
  We show that Definition~\ref{def:exhaustive} applies to every state
  $s \in \kRTS{S} \subseteq \kTS{S}$. By assumption, we have that for
  every $\p \in \PSet$,
    if $q_\p$ is a sending state, then
    $ 
  \forall (q_\p, \action, q'_\p) \in \delta_\p \qst \exists
  \acts \in \ASetC
  \qst
  s  \kTRANSS{\acts}\kTRANSS{\action}  \text{ and } \p \notin
  \acts.$
    By Lemma~\ref{lem:por-trace-equiv}, there is $\acts'$ and $\actsb$ such that
  $s \rkTRANSS{\actsb}$ and $\acts \concat \action \concat \acts' \eqpeerop \actsb$.
    This implies that we have
  $\actsb = \actsb_1 \concat \action \concat \actsb_2$ with
  $\subj{\action} \notin \actsb_1$, and
  $s \rkTRANSS{\actsb_1 \concat \action}$, the required result.
\end{proof}

\begin{restatable}{lemma}{lemredkexhauimpkexhau}\label{lem:red-kexhau-imp-kexhau}
  Let $S$ be reduced $\BA{k}$, if $S$ is \emph{reduced}
  $k$-exhaustive, then $S$ is also $k$-exhaustive.
\end{restatable}
\begin{proof}
  By contradiction, take $s \in \kTS{S}$ such that the
  $k$-exhaustivity property does not hold (i.e., there is
  $\PSEND{pq}{a}$ that cannot be fired within bound $k$).
    By Lemma~\ref{lem:ts-path-rts}, there is $t \in \kRTS{S}$ and
  $\acts$ such that $s \kTRANSS{\acts} t$.
    Then either $\PSEND{pq}{a}$ is in $\acts$, i.e., we have a
  contradiction, or $\p$ is in the same state in $t$.
    By assumption, there is $\actsb$ such that
  $t \rkTRANSS{\actsb \concat \PSEND{pq}{a}}$, and by
  Lemma~\ref{lem:por-trace-equiv} we also have
  $t \kTRANSS{\actsb \concat \PSEND{pq}{a}}$, a contradiction.
\end{proof}

\begin{restatable}{theorem}{themkexhauiffredkexhau}\label{them:kexhau-iff-red-kexhau}
  Let $S$ be reduced $\BA{k}$, $S$ is \emph{reduced} $k$-exhaustive iff $S$ is
  $k$-exhaustive.
\end{restatable}
\begin{proof}
  By Lemma~\ref{lem:kexhau-imp-red-kexhau}
  and Lemma~\ref{lem:red-kexhau-imp-kexhau}.
\end{proof}

\themsafeiffredsafekexhau*
\begin{proof}
  By Theorem~\ref{them:safe-iff-red-safe} and Theorem~\ref{them:kexhau-iff-red-kexhau}
\end{proof}

\begin{restatable}{lemma}{lemredsafeimpsafe}\label{lem:red-safe-imp-safe}
  Let $S$ be reduced $\BA{k}$, if $S$ is $k$-safe, then $S$ is
  also \emph{reduced} $k$-safe.
\end{restatable}
\begin{proof}
  The proof works similarly to the proof of
  Lemma~\ref{lem:kexhau-imp-red-kexhau}.
    We show that Definition~\ref{def:k-safety} applies to every state in
  $s \in \kRTS{S} \subseteq \kTS{S}$.
    Each condition follows easily by showing the existence of an
  equivalent execution, by Lemma~\ref{lem:por-trace-equiv}.
\end{proof}

\begin{restatable}{lemma}{lemsafeimpredsafe}\label{lem:safe-imp-red-safe}
  Let $S$ be reduced $\BA{k}$, if $S$ is \emph{reduced} $k$-safe,
  then $S$ is also $k$-safe.
\end{restatable}
\begin{proof}
  The proof works similarly to the proof of Lemma~\ref{lem:red-kexhau-imp-kexhau}.
    By contradiction, we assume that there is a state $s$ for which the
  properties of Definition~\ref{def:k-safety} do not hold.
    Using Lemma~\ref{lem:ts-path-rts}, we show that there is an
  execution from $s$ to a state in $\kRTS{S}$ for which the properties
  hold by assumption.
\end{proof}

\begin{restatable}{theorem}{themsafeiffredsafe}\label{them:safe-iff-red-safe}
  Let $S$ be reduced $\BA{k}$, $S$ is \emph{reduced} $k$-safe iff $S$ is
  $k$-safe.
\end{restatable}
\begin{proof}
  By Lemma~\ref{lem:red-safe-imp-safe} and Lemma~\ref{lem:safe-imp-red-safe}.
\end{proof} 

\lemporoptimal*
\begin{proof}
  We show that  $\acts \neq \acts' \implies \neg (\acts \eqpeerop
  \acts'$).
    Let $\actsb$ be the longest common prefix of $\acts$ and $\acts'$. 
    Take $s$ such that $s_0 \rkTRANSS{\actsb} s$.
    Since $\acts \neq \acts'$, we must have $\action$ and $\action'$
  such that $s \rkTRANSS{\action}$ and  $s \rkTRANSS{\action'}$.
    However, since $\acts \eqpeerop \acts'$, it must be the case that
  $\subj{\action} \neq \subj{\action'}$; which gives us a
  contradiction since we have 
  that $s \rkTRANSS{\action}$ and $s \rkTRANSS{\action'}$, while
  $\action$ and $\action'$ must be in different sets $L_i$ and $L_j$.
\end{proof}

\begin{restatable}{theorem}{thmreduceddibiiffCibi}\label{thm:reduced-dibi-iff-CIBI}
  Let $S$ be reduced $\OBI{k}$. $S$ is reduced $\CIBI{k}$ iff $S$ is
  $\CIBI{k}$.
\end{restatable} 
\begin{proof}
  By Lemma~\ref{lem:dibi-imp-reduced-CIBI} and~\ref{lem:reduced-dibi-imp-CIBI}.
\end{proof}

\thmreduceddibiiffdibiBOTH* 
\begin{proof}
  By Theorem~\ref{thm:reduced-dibi-iff-dibi} and Theorem~\ref{thm:reduced-dibi-iff-CIBI}.
\end{proof}

 \section{Proofs for Section~\ref{sec:exist-bounded}}

\begin{lemma}\label{lem:kb-imp-kmb}
  Let $S$ be a system.
  If $s_0 \kTRANSS{\acts}$, then $\acts$ is $k$-\mbounded{}.
\end{lemma}
\begin{proof}
  We first note that $\acts$ is valid, by Lemma~\ref{lem:valid-word}.
  We have to show that for any prefix $\actsb$ of $\acts$, we have
  \[
  \mathit{min}\{
  \lvert \esndproj{\actsb}{pq} \rvert 
  ,
  \lvert \ercvproj{\acts}{pq} \rvert 
  \}
  -
  \lvert \ercvproj{\actsb}{pq} \rvert 
  \leq k
    \]
      There are two cases:
    \begin{itemize}
  \item If $\lvert \esndproj{\actsb}{pq} \rvert \leq \lvert
    \ercvproj{\acts}{pq} \rvert$, we have the result immediately since 
    \[
    \lvert \esndproj{\actsb}{pq} \rvert - \lvert \ercvproj{\actsb}{pq}
    \rvert \leq k 
    \]
    by hypothesis (and the definition of $k$-boundedness).
  \item  If $\lvert \esndproj{\actsb}{pq} \rvert > \lvert
    \ercvproj{\acts}{pq} \rvert$ then the following holds
    \[
    \lvert \ercvproj{\acts}{pq} \rvert 
    -
    \lvert \ercvproj{\actsb}{pq} \rvert 
    <
    \lvert \esndproj{\actsb}{pq} \rvert 
    -
    \lvert \ercvproj{\actsb}{pq} \rvert 
    \leq
    k
    \]
    by hypothesis, and we have the required result.
    \qedhere
  \end{itemize}
\end{proof}

\subsection{Proofs for Section~\ref{sub:kuske-exist} (Kuske \& Muscholl's boundedness) }

\begin{restatable}{lemma}{lemboundedexeclast}\label{lem:bounded-exec-last}
  If $\acts \concat \action \concat \acts' \in \ASetC$ is a valid
  $k$-\mbounded{} execution such that $\subj{\action} \notin \acts'$
  and $\acts \concat \acts'$ is also valid,
  then
  $\acts \concat \acts'$ is a $k$-\mbounded{} execution.
\end{restatable}
\begin{proof}
  We note that we only have to consider the number of messages on the
  channel of $\action$, as the others are unchanged.
    There are two cases depending on the direction of
  $\action$. 
            \begin{itemize}
  \item If $\action = \PSEND{pq}{a}$, then the result follows
    trivially since the number of send actions strictly decreases.
  \item If $\action = \PRECEIVE{pq}{a}$, we separate the prefixes of
    $\acts \concat \action \concat \acts'$ depending on whether they
    include $\action$ or not. 
            \begin{enumerate}
    \item For each prefix $\actsb$
      of $\acts$, we have
      \[
      \mathit{min}\{
      \lvert \esndproj{\actsb}{pq} \rvert 
      ,
      \lvert \ercvproj{\acts \concat \action \concat \acts'}{pq} \rvert 
      \}
      -
      \lvert \ercvproj{\actsb}{pq} \rvert 
      \leq k
      \]
      by hypothesis.
      We have to show that
      \[
      \mathit{min}\{
      \lvert \esndproj{\actsb}{pq} \rvert 
      ,
      \lvert \ercvproj{\acts \concat \acts'}{pq} \rvert 
      \}
      -
      \lvert \ercvproj{\actsb}{pq} \rvert 
      \leq k
      \]
      which follows trivially since
      $\lvert \ercvproj{\acts \concat \acts'}{pq} \rvert = \lvert
      \ercvproj{\acts \concat \action \concat \acts'}{pq} \rvert -1$.
                \item For each prefix of $\actsb$ of $\acts'$, we have to show that
      \[
      \mathit{min}\{
      \lvert \esndproj{\acts \concat \actsb}{pq} \rvert 
      ,
      \lvert \ercvproj{\acts \concat \acts'}{pq} \rvert 
      \}
      -
      \lvert \ercvproj{\acts \concat \actsb}{pq} \rvert 
      \leq k
      \]
      By hypothesis ($\subj{\action} \notin \acts'$), we have
      $\lvert \ercvproj{\acts'}{pq} \rvert = 0$ and since
      $\acts \concat \action \concat \acts'$ is valid by assumption,
      we have
      $ \lvert \esndproj{\acts}{pq} \rvert \geq \lvert
      \ercvproj{\acts}{pq} \rvert$, hence we are left to show that
      \[
      \lvert \ercvproj{\acts}{pq} \rvert 
      -
      \lvert \ercvproj{\acts \concat \actsb}{pq} \rvert 
      \leq k
      \]
                  Similarly, we know that
      $ 
      \lvert \esndproj{\acts \concat \action}{pq} \rvert 
      -
      \lvert \ercvproj{\acts \concat \action \concat \actsb}{pq} \rvert 
      \leq k
      $. We have the result since
      $   \lvert \ercvproj{\acts}{pq} \rvert = \lvert \ercvproj{\acts \concat \action}{pq} \rvert  -1$
            and
            $ \lvert \ercvproj{\acts \concat \actsb}{pq} \rvert = \lvert
      \ercvproj{\acts \concat \action \concat \actsb}{pq} \rvert -1 $.
            \qedhere
          \end{enumerate}
      \end{itemize}
  \end{proof}

\begin{restatable}{lemma}{lemkmcimpexistentially}\label{lem:kmc-imp-existentially}
  If $S$ is (reduced) $\OBI{k}$, $\infIBI$, and $k$-exhaustive system,
  then it is existentially $k$-bounded.
\end{restatable}
\begin{proof}

  Take $s$ and $\acts$ such that $s_0 \TRANSS{\acts} s$.
  By Lemma~\ref{lem:exist-path-eqpeer-general}, there is $t \in \RS_k(S)$, $\acts'$ and $\actsb$ such that
  $s \TRANSS{\acts'} t$, $s_0 \kTRANSS{\actsb} t$, and
  ${\acts \concat \acts'}\resche{\actsb}$.
    Note that $\actsb$ is valid by Lemma~\ref{lem:valid-word} and
  $k$-\mbounded{} by Lemma~\ref{lem:kb-imp-kmb}.
    We show that there is a $k$-\mbounded{} execution that leads to $s$
  by inductively deconstructing $\acts'$, starting from its last element. 
  \proofsub{Base case} If $\acts' = \emptyw$, then we have the results
  immediately by Lemma~\ref{lem:exist-path-eqpeer-general}, i.e., we have
  ${\acts \concat \emptyw}\resche{\actsb}$ with $\actsb$ $k$-\mbounded{}.

  \proofsub{Inductive case} Take $\acts' = \acts_1 \concat \action$.
    From Lemma~\ref{lem:exist-path-eqpeer-general}, there is $\actsb$
  ($k$-bounded) such that
  ${\acts \concat \acts_1 \concat \action}\resche{\actsb}$.
    Since the two executions are $\resche$-equivalent, we must have
  $\actsb = \actsb_0 \concat \action \concat \actsb_1$ with
  $\subj{\action} \notin \actsb_1$.
    Hence, we have the following situation, where the dashed execution
  is due to the fact that $\subj{\action} \notin \actsb_1$ (i.e.,
  $\action$ is independent from $\actsb_1$):
      \[
  \begin{tikzpicture}[baseline=(current bounding box.center)]
        \node (s0) {$s_0$};
    \node[right=of s0, xshift=1cm] (s) {$s$};

        \node[below=of s0] (sp) {$s'$};
    \node[below=of sp] (spp) {$s''$};
    \node at (s|-sp) (t) {$t$};
    \node at (t|-spp)  (tp) {$t'$};
        \path[->] (s0) edge [above] node {$\acts$} (s);
    \path[->] (s) edge [right] node {$\acts_1$} (t);
    \path[->] (s0) edge [left] node {$\actsb_0$} (sp);
    \path[->,dashed] (sp) edge [above] node {$\actsb_1$} (t);
    \path[->] (spp) edge [above] node {$\actsb_1$} (tp);
    \path[->] (t) edge [right] node {$\action$} (tp);
    \path[->] (sp) edge [left] node {$\action$} (spp);
  \end{tikzpicture}
    \]
  where $\actsb_0 \concat \actsb_1$ is valid by
  Lemma~\ref{lem:valid-word}, and $k$-\mbounded{} by
  Lemma~\ref{lem:bounded-exec-last}.
    Next, we repeat the procedure posing
  $\actsb \coloneqq \actsb_0 \concat \actsb_1$ and
  $\acts' \coloneqq \acts_1$. We note that the procedure always
  terminates since the execution $\acts'$ strictly decrease at each
  iteration.
    \end{proof}

\begin{lemma}\label{lem:mbounded-pref-imp-kbounded}
  If $\acts_0 \concat \acts_1$ is $k$-\mbounded{} and 
  \[
  \forall \ptp{pq} \in \CSet \qst 
    \lvert \esndproj{\acts_0}{pq} \rvert 
  \leq
  \lvert \ercvproj{\acts_0 \concat \acts_1}{pq} \rvert
  \]
  then $\acts_0$ is $k$-bounded for $s_0$.
\end{lemma}
\begin{proof}
  Pick any $\ptp{pq} \in \CSet.$
    By definition of  $k$-\mbounded{},
  for each prefix $\actsb$ of $\acts_0 \concat \acts_1$, we have:
  \[
  \mathit{min}\{
  \lvert \esndproj{\actsb}{pq} \rvert 
  ,
  \lvert \ercvproj{\acts_0 \concat \acts_1}{pq} \rvert 
  \}
  -
  \lvert \ercvproj{\actsb}{pq} \rvert 
  \leq k
  \]
  In particular, for each prefix $\actsb_0$ of $\acts_0$, we have
  $
  \mathit{min}\{
  \lvert \esndproj{\actsb_0}{pq} \rvert 
  ,
  \lvert \ercvproj{\acts_0 \concat \acts_1}{pq} \rvert 
  \}
  -
  \lvert \ercvproj{\actsb_0}{pq} \rvert 
  \leq k
  $.
      By assumption and the fact that
  $\actsb_0$ is a prefix of $\acts_0$, we
  have
  \[
  \lvert \esndproj{\actsb_0}{pq} \rvert 
  \leq
  \lvert \esndproj{\acts_0}{pq} \rvert 
  \leq
  \lvert \ercvproj{\acts_0 \concat \acts_1}{pq} \rvert
  \]
  Hence,  $\mathit{min}\{
  \lvert \esndproj{\actsb_0}{pq} \rvert 
  ,
  \lvert \ercvproj{\acts_0 \concat \acts_1}{pq} \rvert 
  \} =   \lvert \esndproj{\actsb_0}{pq} \rvert$
  and
  $
    \lvert \esndproj{\actsb_0}{pq} \rvert 
    -
    \lvert \ercvproj{\actsb_0}{pq} \rvert 
    \leq k
    $, as required.
    \end{proof}

\begin{restatable}{lemma}{lemlisbonexistsafeimpkexh}\label{lem:lisbon-exist-safe-imp-kexh}
  If $S$ is $\exists$-$k$-bounded and has the eventual reception
  property, then $S$ is $k$-exhaustive.
\end{restatable}
\begin{proof}
  \proofsub{$k$-eventual reception}
  We first show that for all $s = \csconf{q}{w} \in \RS_k(S)$, if
  $w_{\ptp{pq}} = \msg{a} \concat \word$, then
  $s\kTRANSR\kTRANSS{\PSEND{pq}a}$.
    Take $\acts_0$ such that $s_0\kTRANSS{\acts_0} s$. 
  By eventual reception, we have that $s \TRANSS{\acts_1}
  \TRANSS{\PRECEIVE{pq}{a}} t$, for some $\acts_1$ and $t$.
    Take $\acts_2$ such that $t \TRANSS{\acts_2}$ and
    \[
    \forall \ptp{pq} \in \CSet \qst 
    \lvert \esndproj{\acts_0 \concat \acts_1}{pq} \rvert 
  \leq
  \lvert \ercvproj{\acts_0 \concat \acts_1 \concat \PRECEIVE{pq}{a}
    \concat \acts_2}{pq} \rvert
  \]
  there is such $\acts_2$ by the eventual reception property.
    Since $S$ is existentially bounded, there is $\actsb$ such that
  $\actsb$ is $k$-\mbounded{} and
  $\actsb \resche \acts_0 \concat \acts_1 \concat \PRECEIVE{pq}{a}
  \concat \acts_2$.

  Next, remove all actions in $\acts_0$ from $\actsb$ as follows.
    Take the first action in $\acts_0$ (i.e., a send action) and remove
  it from $\actsb$ as well as its receive counterpart, if any.
    If this action is not received within $\acts_0$, then store it in
  $\hat\actsb$.
    Repeat until all actions from $\acts_0$ have been removed, so to
  obtain the sequence: $\hat\actsb \concat \actsb_1$ which is
  $k$-\mbounded{} and valid, so that we have
    $  
  \hat\actsb \concat \actsb_1 \resche \hat\actsb \concat \acts_1
  \concat \PRECEIVE{pq}{a} \concat \acts_2
  $.
  
  Pose $\actsb_1 = \actsb_2 \concat \PRECEIVE{pq}{a} \concat \actsb_3$
  and let us show that $\actsb_2 \concat \PRECEIVE{pq}{a}$ is
  $k$-bounded for $s$, by showing that 
  $\hat\actsb \concat \actsb_2 \concat \PRECEIVE{pq}{a}$ is
  $k$-bounded.
    We have to show that all prefixes are $k$-bounded. This is trivial
  for any prefix of  $\hat\actsb$ since $s \in \RS_k(S)$.
    For any prefix $\hat\actsb_2$ of $\actsb_2$ we have to show that 
  \[
    \forall \ptp{pq} \in \CSet \qst 
        \lvert \esndproj{\hat\actsb \concat \hat\actsb_2}{pq} \rvert 
    -
    \lvert \ercvproj{\hat\actsb \concat \hat\actsb_2}{pq} \rvert
    \leq k
  \]
  Since $ \hat\actsb \concat \actsb_1$ is $k$-\mbounded{}, we have
 \[
    \forall \ptp{pq} \in \CSet \qst 
        \mathit{min}\{
    \lvert \esndproj{\hat\actsb \concat \hat\actsb_2 }{pq} \rvert 
    ,
     \lvert \ercvproj{\hat\actsb \concat \actsb_2 \concat \PRECEIVE{pq}{a} \concat \actsb_3}{pq} \rvert
    \}
    -
    \lvert \ercvproj{\hat\actsb \concat \hat\actsb_2}{pq} \rvert
    \leq k
  \]
  By construction, we have
  $ \lvert \esndproj{\hat\actsb \concat \hat\actsb_2 }{pq} \rvert
    \leq
    \lvert \ercvproj{\hat\actsb \concat \actsb_2 \concat
    \PRECEIVE{pq}{a} \concat \actsb_3}{pq}\rvert$,
    hence we have the required result.

  \proofsub{$k$-exhaustivity}
  We show the rest by contradiction. Assume there is $s \in \RS_k(S)$
  for which the $k$-exhaustivity condition does not hold.
    Hence, there must be $\ptp{pq} \in \CSet$ such that $\lvert
  \word_{\ptp{pq}} \rvert = k \geq 1 $.
    From the result above, we have $s \kTRANSR
  \kTRANSS{\PRECEIVE{pq}{a}} t$ for some $\msg{a}$, and therefore we
  have $t \kTRANSS{\PSEND{pq}{b}}$, for any $\msg{b}$, a contradiction.
\end{proof}

\begin{restatable}{lemma}{lemexistentiallyextend}\label{lem:existentiallyextend}
  If $S$ is existentially $k$-bounded and safe, then for any
  $k$-\mbounded{} $\acts$ such that $s_0\TRANSS{\acts} s$, there are 
  $\actsb$ and $\acts'$ such that $s_0 \kTRANSS{\actsb} t$ and $s
  \TRANSS{\acts'} t$ and $\actsb \resche \acts \concat \acts'$.
\end{restatable}
\begin{proof}
  Take  $\acts$  $k$-\mbounded{}  s.t.\ $s_0\TRANSS{\acts} s$.
    By safety, there is $\acts'$ such that $s \TRANSS{\acts'}$ with
  $\forall \ptp{pq} \in \CSet \qst \lvert \esndproj{\acts}{pq}
  \rvert \leq \lvert \ercvproj{\acts \concat \acts'}{pq} \rvert $,
  i.e., we extend $\acts$ with an execution that consumes all messages
  sent in $\acts$.
  
  Since $S$ is existentially bounded, there is
  $\actsb \in \equivclass{\acts \concat \acts'}{\resche} \cap
  \bounded{\ASetC}{k}$.
    Take prefix $\actsb_0$ of $\actsb$ such that
  $\exists \acts'' \qst \forall \p \in \PSet \qst \proj{\actsb_0}{p} =
  \proj{\acts \concat \acts''}{p}$.
    If $\actsb_0$ is $k$-bounded, we have the required result,
  otherwise, there must be a prefix $\actsb_1$ of $\actsb_0$ such that
  \[
    \lvert \esndproj{\actsb_1}{pq} \rvert 
    -
    \lvert \ercvproj{\actsb_1}{pq} \rvert 
    > k
  \]
  However, since $\actsb$ is $k$-\mbounded{}, we have
  \[
    \mathit{min}\{
    \lvert \esndproj{\actsb_1}{pq} \rvert 
    ,
    \lvert \ercvproj{\actsb}{pq} \rvert 
    \}
    -
    \lvert \ercvproj{\actsb_1}{pq} \rvert 
    \leq k
  \]
  and by construction of $\actsb \resche \acts \concat \acts'$, we have
  $\lvert \esndproj{\actsb_1}{pq} \leq \lvert \ercvproj{\actsb}{pq}
  \rvert$, i.e., a contradiction.
\end{proof}

\begin{restatable}{theorem}{thmexistentiallyiffkmc}\label{thm:existentially-iff-kmc}
  (1) If $S$ is (reduced) $\OBI{k}$, $\infIBI$, and $k$-exhaustive,
  then it is existentially $k$-bounded.
    (2) If $S$ is existentially $k$-bounded and has the eventual
  reception property, then it is $k$-exhaustive.
  \end{restatable}
\begin{proof}
  Part (1) follows from Lemma~\ref{lem:kmc-imp-existentially}  and
  Part (2) follows from Lemmas~\ref{lem:lisbon-exist-safe-imp-kexh}.
\end{proof}

\thmexistentiallyiffkmcnotreduced*
\begin{proof}
  By Theorem~\ref{thm:existentially-iff-kmc}.
\end{proof}

\subsection{Proofs for Section~\ref{sub:classical-exist} (stable boundedness)}

\begin{restatable}{lemma}{lemstableconfmbounded}\label{lem:stable-conf-mbounded}
    Let $S$ be a system and $\acts \in \ASetC$ such that $s_0
  \TRANSS{\acts} s = \stablecsconf{q}$, then
  $\acts$ is $k$-\mbounded{} if and only if $\acts$ is
  $k$-bounded for $s_0$.
  \end{restatable}
\begin{proof}
  The ($\Leftarrow$) direction follows from
  Lemma~\ref{lem:kb-imp-kmb}.
  The ($\Rightarrow$) direction follows from the fact that for
any prefix $\actsb$ of $\acts$, we have
\[
  \lvert \esndproj{\actsb}{pq} \rvert \leq \lvert \ercvproj{\acts}{pq}
  \rvert
\]
since all messages sent along $\acts$ are received (all channels in
$s$ are empty).
Hence we have
$ \lvert \esndproj{\actsb}{pq} \rvert - \lvert \ercvproj{\actsb}{pq}
\rvert \leq k $ by Definition~\ref{def:exist-bounded}, i.e., $\acts$
is $k$-bounded.
\end{proof}

\thmkuskeimpclassical*
\begin{proof}
  We show both statements by contradiction.
  \begin{enumerate}
  \item Assume by contradiction that $S$ is existentially $k$-bounded, but
    \emph{not} existentially stable $k$-bounded.
        Then, there must be $\acts$ such that $s_0 \TRANSS{\acts} s =
    \stablecsconf{q}$ where $\acts$ has no $\resche$ equivalent
    execution which is $k$-bounded for $s_0$.
            However, since $S$ is  existentially $k$-bounded, there is $\actsb
    \resche \acts$ such that $\actsb$ is $k$-\mbounded{}.
        Since $s_0 \TRANSS{\actsb} \stablecsconf{q}$, by
    Lemma~\ref{lem:stable-conf-mbounded}, $\actsb$ is $k$-bounded, a
    contradiction.
      \item Assume by contradiction that $S$ is existentially stable
    $k$-bounded and has the stable property, but \emph{not} existentially
    $k$-bounded.
        Then there is $\acts$ such that
    $s_0 \TRANSS{\acts} s = \csconf{q}{w}$ (with $\vec{q}$ not empty)
    such that $\acts$ has no $\resche$ equivalent execution which is
    $k$-\mbounded{} for $s_0$.
        Since $S$ has the stable property, we have $s \TRANSS{\acts'}$ and there
    is $\actsb \resche \acts \concat \acts'$ such that $\actsb$ is
    $k$-bounded (since $S$ is $\exists$S-$k$-bounded).
        Then we reason as for the proof of
    Lemma~\ref{lem:kmc-imp-existentially} and progressively
    deconstruct $\acts'$ to show that there is a subsequence of
    $\actsb$ that is $k$-\mbounded{} and $\resche$-equivalent to
    $\acts$, a contradiction.
    \qedhere
  \end{enumerate}
  \end{proof}

\begin{lemma}\label{lem:dlf-reach-k-stable-stable}
  Let $S$ be $\exists$-$k$-bounded, then for all \emph{stable}
  configurations $s$ and $s'$ in $\RS(S)$ such that
  $s \TRANSS{\acts} s'$, there is $\actsb \resche \acts$ such that
  $\actsb$ is $k$-bounded (for $s$).
\end{lemma}
\begin{proof}
  Since $s$ is stable and $S$ is $\exists$-$k$-bounded, there is
  $\acts_0$ $k$-bounded for $s_0$ such that $s_0 \kTRANSS{\acts_0} s$,
  and we have $\hat\actsb$ $k$-bounded such that
  $\hat\actsb \resche \acts_0 \concat \acts$.
    We show that we inductively remove the actions of $\acts_0$ from
  $\hat\acts$ while preserving its $k$-boundedness.
    Since $s$ and $s'$ are stable, we have
  $\acts_0 = \PSEND{pq}{a} \concat \acts'_1 \concat \PRECEIVE{pq}{a}
  \concat \acts'_2$, with $\ercvproj{\acts'_1}{pq} = \emptyw$.
    Hence, we can remove the first respective occurrence of
  $\PSEND{pq}{a}$ and $\PRECEIVE{pq}{a}$ from $\hat\actsb$ without
  affecting its $k$-boundedness: ($i$) the new execution is still valid
  since we remove a send and its receive and ($ii$) the bound is
  preserved since we remove a send and a receive simultaneously.
    We repeat the procedure until all the elements of $\acts_0$ have
  been removed and we obtain the required result.
\end{proof}

\begin{restatable}{lemma}{lemdlfreachkstable}\label{lem:dlf-reach-k-stable}
  Let $S$ be an existentially stable $k$-bounded system with the
  stable property, then for all $s \in \RS_k(S)$, there is $t$ stable
  such that $s \kTRANSR t$.
\end{restatable}
\begin{proof}
  First we observe that for any stable $t$, we have $t \in \RS_k(S)$
  since $S$ is $\exists$S-$k$-bounded, by
  Lemma~\ref{lem:stable-conf-mbounded}.
      Assume $t_0$ is stable and $t_0 \kTRANSS{\acts} s$.
    We show the result by induction on the length of $\acts$.
  
  If $\acts = \action$, then we have the result since $t_0$ is stable
  and there is stable $t'$ such that $t_0 \kTRANSS{\action} s \TRANSR
  t'$ since $S$ has the stable property.
    Finally, by Lemma~\ref{lem:dlf-reach-k-stable-stable}, we have
  $s \kTRANSR t'$.

    Assume the result holds for $\acts$ and let us show that it holds
  for $\acts \concat \action$.
    Pose $t_0 \kTRANSS{\acts} s \kTRANSS{\action} s'$.
    By induction hypothesis, we have that $s \kTRANSS{\acts'} t$ for
  some $t$ stable and $\acts' \in \ASetC$.
    We  have to show that $s' \kTRANSR t'$ with $t'$ stable. There are
  two cases:
  \begin{itemize}
  \item If $\subj{\action} \notin \acts'$, then we have
    $s' \TRANSS{\acts} t'$ and $t \kTRANSS{\action} t'$, and we only
    have to show that $s' \kTRANSS{\acts} t'$, which follows trivially
    from the fact that  $\subj{\action} \notin \acts'$ (i.e., there is
    no other send on the channel in $\acts'$).
    \item If $\subj{\action} \in \acts'$, then there are two sub-cases
      depending on the direction of $\action$.
      \begin{itemize}
      \item If $\action$ is a receive action, then the result follows
        trivially.
      \item If $\action$ is a send action. Assume w.l.o.g. that
        $\acts' = \acts'_1 \concat \action \concat  \acts'_2$ with
        $\subj{\action} \notin \acts'_1$, then we have $s'
        \kTRANSS{\acts'_1} \kTRANSS{\acts'_2} t' = t$, and we have the
        required result.
      \end{itemize}
  \end{itemize}
  We have shown that either there is stable $t$  such that
  $t \kTRANSS{\action} t'$, hence we are back to the base case, or
  $t=t'$, in which case the result follows trivially.
\end{proof}

\thmexistentiallydlfiffkexh*
\begin{proof}
  We first note that by Theorem~\ref{thm:kuske-imp-classical} we have
  that $S$ is both $\exists$S-$k$-bounded and $\exists$-$k$-bounded
  since it has the stable property.
    Assume by contradiction, that $S$ is not $k$-exhaustive.
   Then, there
  is $s$ such that $s_0 \kTRANSS{\acts} s = \csconf{q}{w}$ and $\p$
  such that $(q_\p, \PSEND{pq}{a}, q'_\p) \in \delta_\p$ and
  $\neg( s \kTRANSR \kTRANSS{\PSEND{pq}{a}})$.
    By Lemma~\ref{lem:dlf-reach-k-stable}, there is stable $t$  such that
  $s \kTRANSS{\actsb} t$.
    Then either $\p \in \actsb$ and therefore $\PSEND{pq}{a}$ can be
  fired in $\actsb$ and we have a
  contradiction, or $\p \notin \actsb$ and $t \kTRANSS{\PSEND{pq}{a}}$,
  i.e., another contradiction.
\end{proof}

\subsection{Proofs for Section~\ref{sec:cav-sync} (synchronisability)}
\begin{restatable}{lemma}{lemksynkexecexist}\label{lem:ksynk-exec-exist}
  Let $\acts$ be a valid execution.
  If $\acts$ is a $k$-exchange then it is a $k$-\mbounded{} execution.
\end{restatable} 
\begin{proof}
  Since $\acts$ is a $k$-exchange, it must be of the form
  \[\acts =
    \acts_1 \concat \actsb_1 \cdots \acts_n \concat \actsb_n \quad
    \text{where } \forall 1 \leq i \leq n \qst \acts_i \in \ASetSendC
    \land \actsb_i \in \ASetRcvC
    \land \lvert \acts_i \rvert \leq k
  \]
  We must show that for every prefix $\hat\acts$ of $\acts$ and every
  $\ptp{pq} \in \CSet$, the following holds:
  \[
  \mathit{min}\{
  \lvert \esndproj{\hat\acts}{pq} \rvert 
  , 
  \lvert \ercvproj{\acts}{pq} \rvert 
  \}
  -
  \lvert \ercvproj{\hat\acts}{pq} \rvert 
  \leq k
  \]
  We first observe that, for all $1\leq i \leq n$, if
  $\hat\acts = \acts_1 \concat \actsb_1 \cdots \acts_i$ is
  $k$-\mbounded{}, then so is $\hat\acts \concat \actsb_i$ (since
  $\actsb_i \in \ASetRcvC$), hence we only show the result for the
  former.
    Take $\ptp{pq} \in \CSet$ and pose
  $\hat\acts = \acts_1 \concat \actsb_1 \cdots \acts_i$ (with
  $1\leq i \leq n$). There are two cases:
    \begin{itemize}
  \item If for all $1 \leq j < i \qst  \esndproj{\acts_j}{pq} =
    \ercvproj{\actsb_j}{pq}$, then all messages sent on channel
    $\ptp{pq}$ are received within each exchange.
                \begin{itemize}
    \item Case $\lvert \esndproj{\hat\acts}{pq} \rvert \leq \lvert
      \ercvproj{\acts}{pq} \rvert$.
      We have
      \[
        \begin{array}{rcl}
          \lvert \esndproj{\hat\acts}{pq} \rvert 
          & = & \lvert \esndproj{\acts_1 \cdots \acts_i}{pq} \rvert 
          \\
          & = & \lvert \ercvproj{\actsb_1 \cdots \actsb_{i-1}}{pq} \rvert
                +  \lvert \esndproj{\acts_i}{pq} \rvert
          \\
          & = & \lvert \ercvproj{\hat\acts}{pq} \rvert
                +  \lvert \esndproj{\acts_i}{pq} \rvert
        \end{array}
      \]
      Hence,
      $\lvert \esndproj{\hat\acts}{pq} \rvert - \lvert
      \ercvproj{\hat\acts}{pq} \rvert = \lvert \esndproj{\acts_i}{pq}
      \rvert \leq k$, and we have the required result.
          \item Case $\lvert \esndproj{\hat\acts}{pq} \rvert > \lvert
      \ercvproj{\acts}{pq} \rvert$.
            Then, there is $i \leq m \leq n$ such that
      $\esndproj{\acts_m}{pq} \neq \ercvproj{\actsb_m}{pq}$ and
                        we have
          \[
            \begin{array}{rcl}
              \ercvproj{\acts}{pq} & = &  \lvert \ercvproj{\actsb_1 \cdots
                                             \actsb_{m}}{pq} \rvert 
              \\
                                       & \geq &      \lvert \ercvproj{\actsb_1 \cdots
                                                \actsb_{i}}{pq} \rvert
              \\            
                                   & = &   \lvert
                                         \esndproj{\acts_1 \cdots
                                         \acts_i}{pq} \rvert   
                                         =  \lvert \esndproj{\hat\acts}{pq} \rvert
            \end{array}
    \] 
    Hence, we obtain $ \ercvproj{\acts}{pq} \geq  \lvert
    \esndproj{\hat\acts}{pq} \rvert$, a contradiction with this case.
        \end{itemize}

  \item If there is $j < i$ such that
    $\esndproj{\acts_j}{pq} \neq \ercvproj{\actsb_j}{pq}$ (take
    smallest such $j$), then for all
    $j < m \leq n \qst \ercvproj{\actsb_m}{pq} = \emptyw$, i.e., all
    messages sent after $j$ are not matched.
            Hence, we have
    \begin{equation}\label{eq:lem-synk-allrcv}
      \lvert \ercvproj{\hat\acts}{pq} \rvert
      =
      \lvert \ercvproj{\actsb_1 \cdots \actsb_j}{pq} \rvert
      =
      \lvert \ercvproj{\acts}{pq} \rvert
    \end{equation}
    Thus, we have
    \[
      \begin{array}{rcl}
        \esndproj{\hat\acts}{pq} & = &  \lvert \ercvproj{\actsb_1 \cdots
                                       \actsb_{j-1}}{pq} \rvert
                                       + 
                                       \lvert \esndproj{\acts_j}{pq}
                                       \rvert    
                                       +
                                       \lvert \esndproj{\acts_{j+1} \cdots \acts_i}{pq}
                                       \rvert    
        \\
                                 & \geq &      \lvert \ercvproj{\actsb_1 \cdots
                                          \actsb_{j-1}}{pq} \rvert
                                          + 
                                          \lvert \ercvproj{\actsb_j}{pq}
                                          \rvert 
                                          +
                                          \lvert \esndproj{\acts_{j+1} \cdots \acts_i}{pq}
                                          \rvert    
        \\            
                                 & \geq &   \lvert \ercvproj{\actsb_1 \cdots \actsb_j}{pq} \rvert
                                          =
                                          \lvert \ercvproj{\acts}{pq} \rvert     
      \end{array}
    \] 
    Hence, we only have to show that
    $ \lvert \ercvproj{\acts}{pq} \rvert - \lvert
    \ercvproj{\hat\acts}{pq} \rvert \leq k$, which holds
    by~\eqref{eq:lem-synk-allrcv}.    
    \qedhere
  \end{itemize}
\end{proof}

\thmksynkrelkmc*
\begin{proof}
  Item (1) follows from Lemma~\ref{lem:ksynk-exec-exist}: for any execution
  of $S$, there is an equivalent $k$-exchange, which is a
  $k$-\mbounded{} execution.
    Item (2) follows from Item (1) and Item (2) of 
  Theorem~\ref{thm:kuske-imp-classical}.  
\end{proof}

\bibliographystyleout{abbrv}
\bibliographyout{kmc}

\end{document}
